\documentclass[11pt,english,american]{article}
\usepackage[T1]{fontenc}
\usepackage[utf8]{inputenc}
\usepackage{amsmath}
\usepackage{amsthm}
\usepackage{amssymb}
\PassOptionsToPackage{normalem}{ulem}
\usepackage{ulem}
\usepackage{svg}
\usepackage{enumitem}
\usepackage{hyperref}
\usepackage{subcaption, cleveref, enumitem}
\usepackage{times}
\usepackage{verbatim}
\usepackage{changepage}
\usepackage{thmtools}

\makeatletter
\theoremstyle{plain}
\newtheorem{theorem}{\protect\theoremname}
\theoremstyle{plain}

\theoremstyle{definition}
\newtheorem{definition}[theorem]{\protect\definitionname}
\theoremstyle{plain}
\newtheorem{conjecture}{Conjecture}

\newif\ifdraft

\ifdraft
\usepackage[colorinlistoftodos,prependcaption,textsize=tiny]{todonotes}
\usepackage{xargs}
\newcommand{\awr}[1]{\todo[color=blue!100!black!50]{AW: #1}}
\newcommand{\akr}[1]{\todo[color=red!100!black!50]{AK: #1}}
\newcommand{\dkr}[1]{\todo[color=green,size=\tiny]{DK: #1}}
\newcommand{\aw}[1]{\textcolor{blue}{#1}}
\newcommand{\ak}[1]{\textcolor{red}{#1}}
\newcommand{\dk}[1]{\textcolor{green!60!black}{#1}}

\else
    \newcommand{\awr}[1]{}
    \newcommand{\akr}[1]{}
    \newcommand{\dkr}[1]{}
    \newcommand{\aw}[1]{\textcolor{black}{#1}}
    \newcommand{\ak}[1]{\textcolor{black}{#1}}
    \newcommand{\dk}[1]{\textcolor{black}{#1}}
\fi

\usepackage{fullpage}

\newcommand{\lc}{$L\& C^*$}

\makeatother

\usepackage{babel}
\addto\captionsamerican{\renewcommand{\definitionname}{Definition}}
\addto\captionsamerican{\renewcommand{\lemmaname}{Lemma}}
\addto\captionsamerican{\renewcommand{\theoremname}{Theorem}}
\addto\captionsenglish{\renewcommand{\definitionname}{Definition}}
\addto\captionsenglish{\renewcommand{\lemmaname}{Lemma}}
\addto\captionsenglish{\renewcommand{\theoremname}{Theorem}}
\addto\captionsenglish{}

\providecommand{\definitionname}{Definition}
\providecommand{\lemmaname}{Lemma}
\providecommand{\theoremname}{Theorem}

\begin{document}
\global\long\def\OPT{\mathrm{OPT}}%
\global\long\def\B{\mathcal{B}}%
\global\long\def\epssmall{\epsilon_{\mathrm{small}}}%
\global\long\def\epslarge{\epsilon_{\mathrm{large}}}%
\global\long\def\epsthin{\epsilon_{\mathrm{thin}}}%
\global\long\def\R{\mathbb{R}}%
\global\long\def\C{\mathcal{C}}%
\global\long\def\eps{\epsilon}%
\global\long\def\N{\mathbb{N}}%
\global\long\def\epscsmall{\epsilon_{\mathrm{thin}}^{\mathrm{container}}}%
\global\long\def\epsclarge{\epsilon_{\mathrm{thick}}^{\mathrm{container}}}%
\global\long\def\Cl{\mathcal{C}_{\mathrm{thick}}}%
\global\long\def\SM{\mathcal{S}}%
\global\long\def\stripthickness{\eps \cdot \epsclarge}%
\global\long\def\RR{\mathcal{R}}%
\global\long\def\Ls{\mathsf{L}}%

\newcommand{\lcor}{\ensuremath{\mathsf{L}}}
\newcommand{\AM}{\mathcal{A}}
\newcommand{\TT}{\mathcal{T}}
\newcommand{\optr}{$\OPT_{\text{r-}L\&C^*}$}
\newcommand{\CC}{\mathcal{C}}
\newcommand{\DP}{\mathsf{DP}}
\newcommand{\Ihor}{I_H}
\newcommand{\Iver}{I_V}
\newcommand{\Irem}{I_R}
\newcommand{\II}{\mathcal{I}}
\newcommand{\optcont}{\mathrm{OPT}_{\text{cont}}}

\newcommand{\Rs}{\mathsf{R}}
\newtheorem{lemma}[theorem]{\protect\lemmaname}
\newtheorem{corollary}[theorem]{\protect\lemmaname}
\theoremstyle{definition}
\newcommand{\BB}{\mathcal{R}}
\newcommand{\ILTlong}{I_{LT_{\text{long}}}}
\newcommand{\ILTshort}{I_{LT_{\text{short}}}}
\newcommand{\epstiny}{\epsilon_{\mathrm{thin}}}
\newcommand{\optthin}{I'_T}
\newcommand{\hugeitem}{i^*}
\newcommand{\optshr}{\OPT_{\mathrm{shrink}}}
\newcommand{\lft}{\text{left}}
\newcommand{\rht}{\text{right}}
\newcommand{\bottom}{\text{bottom}}

\global\long\def\epsskew{\epsilon_{\mathrm{skew}}}%

\title{Approximation Schemes and Structural Barriers for the Two-Dimensional Knapsack Problem with Rotations}
\author{Debajyoti Kar\thanks{Department of Computer Science and Automation, Indian Institute of Science, India. Supported by Google PhD Fellowship and Microsoft Research India PhD Award. Email: \texttt{debajyotikar@iisc.ac.in}}
\and Arindam Khan\thanks{Department of Computer Science and Automation, Indian Institute of Science, India. Research partly supported by Google India Research Award, SERB Core Research Grant
(CRG/2022/001176) on “Optimization under Intractability and Uncertainty”, Ittiam Systems CSR grant, and the Walmart
Center for Tech Excellence at IISc (CSR Grant WMGT-23-0001). Email: \texttt{arindamkhan@iisc.ac.in}}
\and Andreas Wiese \thanks{Technical University of Munich, Munich, Germany. Email: \texttt{andreas.wiese@tum.de}}}
\date{}
\maketitle

\abstract{We study the two-dimensional (geometric) knapsack problem
with rotations (2DKR), in which we are given a square knapsack and
a set of rectangles with associated profits. The objective is to find
a maximum profit subset of rectangles that can be packed without overlap
in an axis-aligned manner, possibly by rotating some rectangles by $90^{\circ}$.
The best-known polynomial time algorithm for the problem has an approximation
ratio of $3/2+\eps$ for any constant $\eps>0$, with an improvement to $4/3+\eps$ in the cardinality
case, due to G{á}lvez, Grandoni, Heydrich, Ingala, Khan, and Wiese
(FOCS 2017, TALG 2021). Obtaining a PTAS for the problem, even in
the cardinality case, has remained a major open question in the setting of multidimensional
packing problems, as mentioned in the survey by Christensen, Khan, Tetali,
and Pokutta (Computer Science Review, 2017).

In this paper, we present a PTAS for the cardinality case of 2DKR.
In contrast to the setting without rotations, we show that there are
$(1+\eps)$-approximate solutions in which all items are packed greedily
inside a constant number of rectangular {\em containers}.
Our result is based on a new resource
contraction lemma, which might be of independent interest. 
With our techniques, we also obtain a $(1+\epsilon)$-approximation algorithm in the weighted
case when all given items are \emph{skewed}, i.e., each of them has sufficiently small height or sufficiently small width.
In contrast, for the general weighted case, we prove that this simple type of packing is not sufficient to obtain a better approximation ratio than $1.5$.
However, we break this structural barrier and design a $(1.497+\eps)$-approximation
algorithm for 2DKR in the weighted case. Our arguments also improve the best-known
approximation ratio for the (weighted) case {\em without rotations} to $13/7+\eps \approx 1.857+\epsilon$.

Finally, we establish a lower bound of $n^{\Omega(1/\eps)}$ on the running time of any $(1+\eps)$-approximation algorithm for our problem with or without rotations -- even in the cardinality setting, assuming the $k$-\textsc{Sum} Conjecture.
In particular, this shows that an approximation scheme for the case of rectangles of two-dimensional geometric knapsack requires much more running time than for the case of squares.

\thispagestyle{empty}
\newpage

\setcounter{page}{1}

\section{Introduction}

The (geometric) 2D-Knapsack (2DK) problem is a fundamental problem
in combinatorial optimization, computational geometry, and approximation
algorithms. In 2DK, we are given a square knapsack $K:=[0,N]\times[0,N]$
for some given value $N\in\N$ and a set of $n$ (rectangular) items $I$. Each
item $i\in I$ corresponds to an axis-aligned open rectangle with
a width $w(i)\in\N$, a height $h(i)\in\N$, and an associated profit
$p(i)\in\N$. The objective is to select a subset $S\subseteq I$
of the given items and place them non-overlappingly inside the knapsack.
Formally, we need to define a bottom-left corner $(\lft(i),\bottom(i))\in K$ for each item $i\in S$ such that the
resulting rectangle $(\lft(i),\lft(i)+w(i))\times(\bottom(i),\bottom(i)+h(i))$ is contained in $K$ and it does not
intersect with any rectangle corresponding to any other item in $S$.
The objective is to maximize the total profit $p(S):=\sum_{i\in S}p(i)$.
Starting from the classical works of Gilmore and Gomory in the 1960s
\cite{gilmore1965multistage}, 2DK and its variants have found extensive
applications in practice, particularly in domains such as logistics,
cutting-stock, and scheduling\aw{, see} \cite{ali2022line}
and references therein.

In many practical applications \cite{coffman1980performance, bengtsson1982packing, vasko1989practical,  dell2002lower} (e.g.,  cutting stock, container loading, VLSI/PCB layout), it is important that the items are packed (or cut) parallel to the two coordinate axes. However,
it is often beneficial to rotate some of the items by $90^\circ$ (also called orthogonal rotations), as we might be able to select more items in this way.
For example, commercial software tools such as CutList Optimizer \cite{CutListOptimizer2025} and optiCutter \cite{optiCutter2025} provide an option to allow orthogonal rotations.
Also from a theoretical perspective, orthogonal rotations \aw{in geometric packing problems have} been extensively studied; see, e.g.,
\cite{fujita2002two, epstein2003two, miyazawa2004packing, JansenS05}.
In fact, this setting was already suggested by 
 Coffman, Garey, Johnson, and Tarjan \cite{coffman1980performance} \aw{in 1980}.

In this paper, we study the rotational variant of 2DK (i.e., we allow each item to be rotated by $90^{\circ}$ prior to placing it inside the knapsack)
which we call the  {\em two-dimensional geometric knapsack with rotations
(2DKR)} problem.
Note that there are $2^n$ options for rotating the $n$ given items, which potentially allows much more possible solutions.
{As the setting without rotations, it is strongly NP-hard \cite{leung1990packing}.
Despite a lot of research \cite{Jansen2004, adamaszek2015knapsack, galvez2021approximating, Galvez00RW21, DBLP:conf/compgeom/BuchemDW24, GKW19}, there is still a significant
gap in our understanding of this problem. On the one hand, there is
a $(1+\eps)$-approximation \ak{algorithm} \aw{with} {\em quasi-polynomial \aw{running} time} when the input
numbers are quasi-polynomially bounded integers~\cite{adamaszek2015knapsack}, 
which suggests that the problem might admit a \ak{polynomial time approximation scheme~(PTAS)}. On the other hand,
the best-known polynomial time approximation ratio is only $3/2+\eps$,
and there is an improvement to $4/3+\eps$ in the cardinality case,
i.e., when $p(i)=1$ for each item $i\in I$~\cite{galvez2021approximating}.
When we allow pseudo-polynomial running time, the best-known factors
are a bit better, i.e., $4/3+\eps$ and $5/4+\eps$,
respectively~\cite{Galvez00RW21}.

\aw{However,} it is open whether a PTAS exists for 2DKR, even in
the cardinality case! The survey by Christensen, Khan, Pokutta, and
Tetali \cite{CKPT17} lists this as one out of ten major open problems.
Also, \cite{galvez2021approximating} explicitly highlighted
that “the main problem that remains open is to find a PTAS, if any,
for 2DK and 2DKR. This would be interesting even in the cardinality
case.'' However, despite progress on several special cases and variations
\cite{Galvez00RW21, GKW19, abs-2102-05854, khanMSW21, MariPP23, DBLP:conf/compgeom/BuchemDW24, KhanLMSW25},
after the results in~\cite{GGHIKW17} from 2017, there has been no progress on polynomial time approximation algorithms
for 2DK or 2DKR, even in the cardinality case.

\subsection{Our contribution}

We make the first progress in nearly ten years on polynomial time
algorithms for the two-dimensional geometric knapsack problem \aw{for rectangles}. Our
first result is that we resolve the long-standing open problem of
finding a PTAS for the cardinality case of 2DKR.
\begin{theorem}
\label{theorem:PTAS-cardinality-case} There is a PTAS for the cardinality
case of 2DKR.
\end{theorem}

\begin{figure}
  \centering

  \begin{subfigure}[b]{0.45\textwidth}
  \centering
  \hspace*{-0.7cm}
    \includegraphics[width=1.2\linewidth]{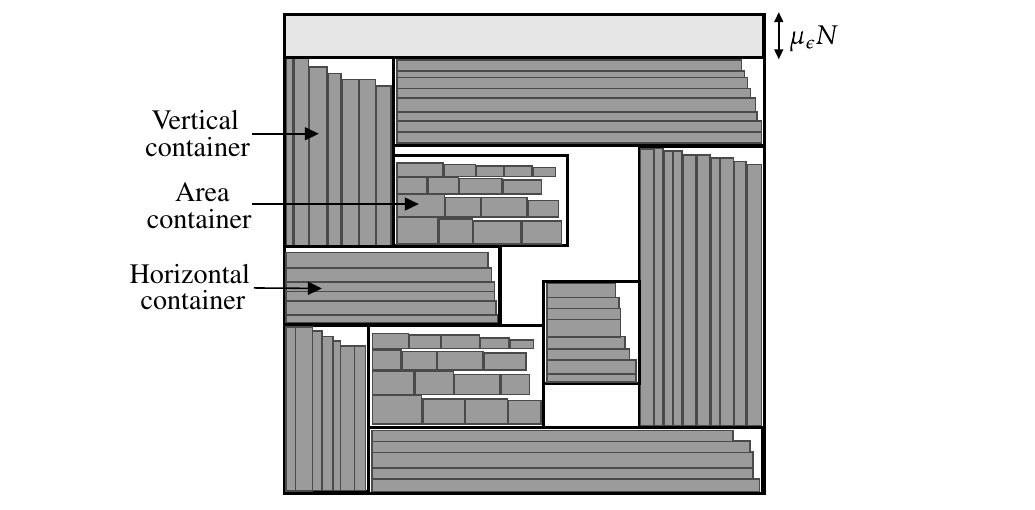}
    \label{fig:res-contraction}
  \end{subfigure}
  \hfill
  \begin{subfigure}[b]{0.45\textwidth}
  \centering
  \hspace*{-0.9cm}
    \includegraphics[width=1.2\linewidth]{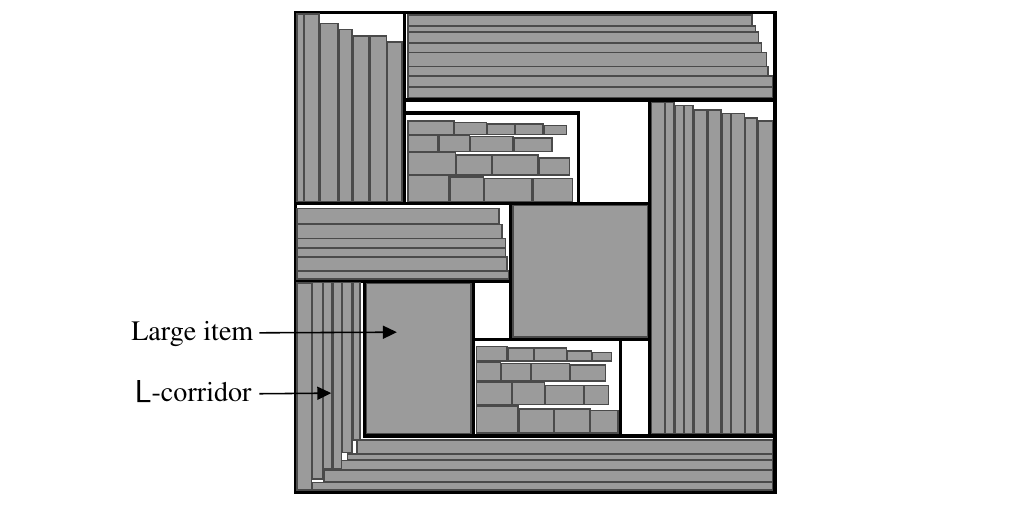}
    \label{fig:an-lc-packing}
  \end{subfigure}

  \caption{Left: A container packing with an empty strip of height $\mu_\eps N$. Right: a packing
  with rectangular containers and one additional container with the shape of an $\Ls$.}
  \label{fig:structured-packings}
\end{figure}

Our algorithm is based on the structural result that there always
exists a $(1+\eps)$-approximate packing in which the items are placed
via simple greedy algorithms inside $O_{\eps}(1)$ rectangular boxes\footnote{The notation $O_\eps(f (n))$ denotes that the implicit constant hidden in big-$O$ notation may depend on $\eps$.},
see Figure~\ref{fig:structured-packings}.
This is in stark contrast to the setting without rotations
for which it is known that such packings cannot give a better approximation
ratio than 2, already in the cardinality case~\cite{galvez2021approximating}.
We remark that the mentioned pseudo-polynomial time $(5/4+\eps)$-approximation algorithm for our setting \cite{Galvez00RW21} uses much more complicated
types of containers, e.g., L-shapes and spirals.
Our result shows
that this is not necessary, as our containers are just rectangular
boxes, and we even obtain a better approximation ratio and polynomial
running time. The backbone of our algorithm is a novel \emph{resource contraction
lemma} that might be of independent interest. We show that for any
$\eps>0$, we can pack at least $(1-\eps)|\OPT|-O_{\eps}(1)$
items inside a slightly smaller knapsack of height $(1-\mu_{\eps})N$
and width $N$ for some $\mu_\eps>0$, where $\OPT$ denotes an optimal solution. This is sufficient for a $(1+O(\eps))$-approximation
since we can compute an optimal solution easily in polynomial time by enumeration if $|\OPT|$ is $O_{\eps}(1)$.
Our packing is based on $O_{\eps}(1)$ boxes that are greedily packed; hence, we can compute a near-optimal
packing of this type easily using standard techniques via a reduction to the generalized assignment problem (GAP) \cite{galvez2021approximating}.

To prove our resource contraction lemma, we start with a $(1+\eps)$-approximate
solution in which the items are placed inside $O_{\eps}(1)$ thin
corridors (see \Cref{fig:corridor-decomposition}) whose existence was proven in~\cite{adamaszek2015knapsack}.
In~\cite{galvez2021approximating}, a routine was employed that intuitively
sacrifices an $\eps$-fraction of these items and places almost all
remaining items inside a constant number of boxes; the remaining items
have very small area in total and they \emph{would} fit inside a thin
strip of width $\epsthin N$ for some parameter $\epsthin>0$ that
can be chosen arbitrarily small. However, these remaining items are
\emph{not} placed inside the knapsack (yet) since we do not necessarily
have a free strip of the required size. In particular, it is not obvious how 
to generate \aw{the corresponding} empty space without increasing the approximation ratio
significantly.

The number of \aw{the resulting} boxes naturally depends on $\epsthin$ and in \cite{galvez2021approximating}
this number depended \emph{polynomially }on $1/\epsthin$. We present
a crucial improvement: we reduce the number of boxes
such that their number depends only \emph{polylogarithmically }on
$1/\epsthin$. \ak{This improvement is pivotal for 
\aw{our second step}.
}
Among the resulting boxes, we temporarily remove those that are relatively
thin. We argue that we can shrink the relatively thick boxes by dropping
a constant number of items per box (which we can easily afford in
the cardinality case) and losing only a factor of $1+\eps$ on the
number of the remaining items. We show that if we push the resulting
boxes maximally to the bottom and to the left, then we free up a thin
strip of width $\Omega(\epsthin N)$ either on the top or on the right of
the knapsack. In this strip, we pack the mentioned unplaced items
(that fit inside the strip of width $\epsthin N$) as well as the
items from the temporarily discarded thin boxes.  In particular,
our improved dependence of the number of (thin) boxes on $\epsthin$
allows us to find a choice for $\epsthin$ for which the thin boxes
fit in the free space \aw{since their total number is sufficiently small}.

A natural next step is the weighted setting of 2DKR. An important special case arises when all input items are \emph{skewed}, i.e., no input item is relatively large in both dimensions.
For related problems like bin packing \cite{KhanS23} or strip packing \cite{GalvezGAJKR23}, tight approximation algorithms are known for the case of skewed items, and it was noted in  \cite{KhanS23} that the ``inherent difficulty of these problems [2DK and another related problem called Maximum Independent Set of Rectangles] lies in instances containing skewed items''.
With our techniques, we can resolve the case of skewed items of 2DKR as well, i.e., we obtain
a polynomial time $(1+\eps)$-approximation algorithm if each item is sufficiently thin in at least one dimension, depending on~$\eps$. This result and our PTAS for the cardinality case can be found in \Cref{sec:ptas-cardinality}.

\begin{theorem}
\label{theorem:skewed-packing}
For each constant $\eps>0$ there is a value $\epsskew>0$ such that
there is a polynomial time $(1+\eps)$-approximation algorithm for
2DKR if each input item $i\in I$ satisfies $h(i)\le\epsskew N$ or
$w(i)\le\epsskew N$.
\end{theorem}

It seems natural to extend our approach to the general weighted case
of 2DKR and, for example, try to improve the mentioned polynomial
time $(3/2+\eps)$-approximation algorithm~\cite{galvez2021approximating}
for this setting. Perhaps surprisingly, we prove that this is impossible even if there is only a single large item and all other items are skewed!
We show that for the weighted case one needs $\Omega(\delta\log N)$ rectangular containers of the above type to obtain a $(3/2-\delta)$-approximation
for any $\delta>0$. However, to obtain polynomial running time, we
can afford only a constant number of such containers.
To prove this lower bound, we define a corresponding family of instances
that all have one huge item whose width equals $N$ (i.e., the width
of the knapsack) and whose height \emph{almost }equals $N$.

However, we show that this is (essentially) the only setting in which constantly
many boxes are not sufficient to obtain a better factor than $3/2$.
Therefore, \aw{in order to break the barrier of $3/2$,}
we define a different type of packing in this problematic case. Our packing
uses $O_{\eps}(1)$ boxes and one special container that has the shape
of an $\Ls$, similar to in~\cite{galvez2021approximating}. However,
if rotations are allowed then it is more difficult to pack items inside
this $\Ls$-container
since each item might be placed in its vertical arm
or in its horizontal arm. To address this, we argue that we need the
$\Ls$-container only in settings where its vertical arm is much shorter
than its horizontal arm and where each item in these arms is relatively
long compared to the respective arm. Another difficulty is that inside
the vertical arm, some items may need to be rotated. In particular,
the height of the $\Ls$ might be so small that items inside are almost squares!
We address this by sacrificing a factor of 2 on those items. With a careful analysis, involving the {\em corridor decomposition technique}  and a {\em resource contraction lemma for the weighted case} \cite{galvez2021approximating},
we show that overall we obtain an approximation
ratio that is strictly better than $3/2$. Our results for the (general) weighted case of 2DKR can be found in \Cref{sec:weighted-case}}.
\begin{theorem}
\label{theorem:rotation-weighted} There exists a polynomial time $(190/127+\epsilon)<(1.497+\eps)$-approximation
algorithm for 2DKR.
\end{theorem}

Furthermore, with our reasoning we can slightly improve and, at the same time,
simplify the best-known polynomial time result for 2DK (i.e., {\em without
rotations}), which has an approximation ratio of $17/9+\eps \approx 1.89+\eps$ \cite{galvez2021approximating} (see Appendix \ref{sec:non-rotation-improved} \aw{for our result}).
\begin{restatable}{theorem}{nonrotationimproved}
\label{theorem:non-rotation-improved-1} There exists a polynomial time
$(13/7+\epsilon)<(1.858+\eps)$-approximation algorithm for 2DK.
\end{restatable}

Given that we now have a PTAS for the cardinality case of 2DKR, a natural question
is how much time is needed for computing a $(1+\eps)$-approximation for the problem.
Note that for squares, there is a $(1+\eps)$-approximation
algorithm whose running time is only $n\log^2 n+(\log n)^{O_{\eps}(1)}$\cite{DBLP:conf/compgeom/BuchemDW24}.
However, for 2DK and 2DKR we prove a much higher lower bound on the running time (see Section \ref{sec:finehard}).

\begin{theorem}
\label{theorem:hardness-of-2dkr} Assuming the $k$-\textsc{Sum} Conjecture, 
an algorithm for 2DK or 2DKR
computing a $(1+\epsilon)$-approximation for every given $\eps >0$ must have a running time of $n^{\Omega(1/\eps)}$, even in the cardinality
case.
\end{theorem}

\subsection{Other related work}
The first approximation algorithms for 2DK and 2DKR are due to Jansen and Zhang \cite{jansen2004maximizing} and both have an approximation ratio of $2+\eps$. As mentioned above, Gálvez, Grandoni, Ingala, Heydrich, Khan, and Wiese~\cite{galvez2021approximating}
improved this factor to $17/9+\eps$ for 2DK and to $3/2+\eps$ for 2DKR, with further improvements in the respective cardinality cases. Gálvez, Grandoni, Khan, Ramírez-Romero, and Wiese showed that all of these ratios can be improved further by allowing pseudo-polynomial running time~\cite{Galvez00RW21}.

Bansal, Caprara, Jansen, Pr{\"a}del, and Sviridenko \cite{bansal2009structural} obtained a PTAS for 2DK and 2DKR when the profit-to-area ratio of the rectangles is bounded by a constant. Grandoni, Kratsch, and Wiese \cite{GKW19} considered parameterized
algorithms for 2DK for which the parameter~$k$ is the size of the
optimal solution. They showed that 2DK and 2DKR are W{[}1{]}-hard
for this parameter and provide a parametrized approximation scheme
for (the cardinality case of) 2DKR with a running time of $k^{O(k/\eps)}n^{O(1/\eps^{3})}$.
Buchem, Deuker, and Wiese~\cite{DBLP:conf/compgeom/BuchemDW24} presented
approximation algorithms for 2DK and 2DKR whose running times are
near-linear, i.e., $n\log^2 n+(\log n)^{O_{\eps}(1)}$. They also provided
dynamic algorithms with polylogarithmic query and update times.

The geometric knapsack problem has also been investigated for
different geometric shapes, including disks and regular polygons \cite{AcharyaBG0MW24},
as well as convex polygons \cite{MerinoW24} and also in higher dimensions,
for instance, in the context of hypercubes \cite{Jansen0LS22} and
cuboids \cite{JansenK0ST25}. It is worth noting that \aw{packing rectangles}
becomes significantly \aw{more} challenging when arbitrary rotations are permitted
(i.e., when items are not required to be axis-aligned); in fact, the
existence of even a polynomial time $O(1)$-approximation algorithm
remains open in this case.

\aw{A related problem is the 2D-Vector Knapsack problem which 
admits a PTAS \cite{frieze1984approximation} while there is a lower bound of $n^{o(1/\epsilon)}$ for the running time of any $(1+\epsilon)$-approximation for the problem, as shown by Jansen, Land, and Land \cite{jansen-lower-bound} (improving a previous lower bound 
of $n^{o(1/\sqrt{\epsilon})}$ by Kulik and Shachnai
\cite{kulik2010there}).
} 

2DK and 2DKR exhibit rich connections with numerous other geometric packing
and covering problems, including geometric bin packing \cite{bansal2014binpacking,Kar0R25}, strip packing \cite{HarrenJPS14, JansenS05}, generalized multidimensional knapsack~\cite{fstt1SS22,abs-2102-05854},
maximum independent
set of rectangles \cite{adamaszek2019approximation,GalvezKMMPW22},
unsplittable flow on a path \cite{amzingUFP2014,GMW22}, storage allocation
and round-SAP \cite{MomkeW20,Kar0W22}, guillotine packing~\cite{KhanLMSW25,khanMSW21},
and rectangle stabbing \cite{KhanSW24,KhanWW24}. For a comprehensive
overview, we refer the reader to the survey in \cite{CKPT17}.

\section{The Corridor Decomposition Framework}
\label{sec:corr-decomposition-appendix}

\begin{figure}
  \centering
  \begin{subfigure}[b]{0.45\textwidth}
  \centering
    \includegraphics[width=0.7\linewidth]{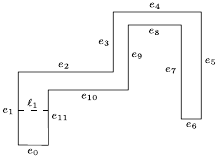}
    \caption{}
    \label{fig:open-corridor}
  \end{subfigure}
  \hfill
  \begin{subfigure}[b]{0.45\textwidth}
  \centering
    \includegraphics[width=0.7\linewidth]{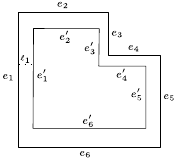}
    \caption{}
    \label{fig:closed-corridor}
  \end{subfigure}

  \caption{(a) An open corridor. (b) A closed corridor.}
  \label{fig:corridor-types}
\end{figure}
We give a brief description of the corridor decomposition framework used in \cite{adamaszek2015knapsack}, which was inspired by the work of \cite{adamaszek2019approximation}. An \textit{open corridor} is a face in the two-dimensional plane, bounded by a simple rectilinear polygon with $2k$ edges $e_0, \ldots, e_{2k-1}$, for some integer $k \ge 2$, such that for each pair of horizontal (resp. vertical) edges $e_i, e_{2k-i}$, $i \in [k-1]$, there exists a vertical (resp. horizontal) line segment $\ell_i$ such that both $e_i$ and $e_{2k-i}$ intersect $\ell_i$, and $\ell_i$ does not intersect any other edge. Similarly, a \textit{closed corridor} is a face bounded by two simple rectilinear polygons defined by the edges $e_1, \ldots e_{k}$ and $e'_1,\ldots, e'_{k}$ such that the second polygon is completely contained inside the first, and for each pair of horizontal (resp. vertical) edges $e_i,e'_i$, $i\in [k]$, there exists a vertical (resp. horizontal) line segment $\ell_i$ such that both $e_i$ and $e'_i$ intersect $\ell_i$ and $\ell_i$ does not intersect any other edge. The minimum length of the line segment $\ell_i$, $i\in [k]$ is called the \textit{width} of the corridor.

Next, observe that each open (resp. closed) corridor is the union of $k-1$ (resp. $k$) rectangular boxes, which we refer to as \textit{subcorridors}.
Each such subcorridor is a maximally large rectangle that is contained inside the corridor. A subcorridor is said to be \textit{horizontal} (resp. \textit{vertical}) if the boundary edges $e_{i}, e_{2k-i}$ or $e_i,e'_i$ are horizontal (resp. vertical). The \textit{length} of the subcorridor is defined as the length of the shorter edge among $e_i,e_{2k-i}$ or $e_i,e'_i$.

We now state the corridor decomposition lemma of \cite{adamaszek2015knapsack}.

\begin{lemma}[\cite{adamaszek2015knapsack}]
\label{lem:corr-decomposition-cardinality}
    Let $\eps>0$ and let $I$ be a set of items packed inside $K$. If every item in $I$ has height or width at least $\delta N$, for a given constant $\delta>0$, then there exists a corridor partition and a set $I_{\text{corr}}\subseteq I$ such that
    \begin{itemize}
        \item There is a subset $I_{\text{corr}}^{\text{cross}}\subseteq I_{\text{corr}}$ with $|I_{\text{corr}}^{\text{cross}}|\in O_{\eps, \delta}(1)$ such that each item in $I_{\text{corr}}\setminus I_{\text{corr}}^{\text{cross}}$ is fully contained in some corridor.
        \item $p(I_{\text{corr}})\ge (1-O(\eps))p(I)$.
        \item The number of corridors is $O_{\eps,\delta}(1)$ and each corridor has width at most $\delta N$ and has at most $1/\eps$ bends.
    \end{itemize}
\end{lemma}

\begin{figure}
    \centering
    \includegraphics[width=0.5\linewidth]{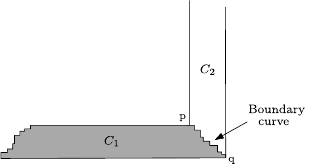}
    \caption{Boundary curve between subcorridors $C_1$ and $C_2$. The shaded region denotes the private region of $C_1$.}
    \label{fig:boundarycurve}
\end{figure}

Next, we define the \textit{boundary curve} between two consecutive subcorridors of a corridor. Consider two neighboring subcorridors $C_1$ and $C_2$ and assume w.l.o.g.~that $C_2$ is to the top right of $C_1$. Consider the two points $p$ and $q$ as shown in Figure \ref{fig:boundarycurve}. The boundary curve is defined as any simple rectilinear monotonically decreasing curve joining $p$ and $q$ that does not intersect any rectangle packed in the two subcorridors. Observe that each subcorridor has two boundary curves and we call the region delimited by these curves as the \textit{private region} of the subcorridor.

Before going into the PTAS for the unweighted case of 2DKR, we first discuss a \emph{resource augmentation lemma}, which informally states that it is possible to obtain a PTAS for 2DK and 2DKR if we are allowed to augment either the height or the width of the knapsack by a factor of $1+\delta$, for some constant $\delta>0$.

\subsection{Resource augmentation lemma}
\label{subsec:ra}
The following lemma allows us to convert an arbitrary packing inside a rectangular region into a container packing. Informally, if there exists an empty strip whose width is at least a $\delta$-fraction of the width of the region, then the items packed in the region can be rearranged into a constant (depending on $\delta$) number of containers by losing only negligible profit.
\begin{lemma}[\cite{galvez2021approximating}]
\label{lem:resource-augmentation}
    Let $I$ be a collection of \aw{$n$} items that can be packed into a box of size $a\times b$, and $\epsilon_{\text{ra}}>0$ be a given constant. Then there exists a container packing of $I'\subseteq I$ inside a box of size $a\times (1+\epsilon_{\text{ra}})b$ (resp., $(1+\epsilon_{\text{ra}})a\times b$) such that
    \begin{itemize}
        \item $p(I')\ge (1-O(\epsilon_{\text{ra}}))p(I)$;
        \item the number of containers is $O_{\epsilon_{\text{ra}}}(1)$ and their sizes belong to a set of cardinality $n^{O_{\epsilon_{\text{ra}}}(1)}$ that can be computed in time $n^{O_{\epsilon_{\text{ra}}}(1)}$.
    \end{itemize}
\end{lemma}

\section{PTAS for cardinality case of 2D Knapsack with rotations}
\label{sec:ptas-cardinality}
In this section, we present our PTAS for the cardinality case of \aw{2DKR}.
Let $\epsilon>0$
and let $\OPT$ be an optimal solution. We classify the items in $\OPT$
according to their heights and widths in the packing. For any subset $S\subseteq I$, we shall let $p(S)$ denote the total profit of the items in $S$. Note that in the cardinality case, it holds that $p(S)=|S|$ -- nevertheless we shall establish all the lemmas in this section for the weighted case as well, so that we can later reuse them to prove our result for the case of skewed items (\Cref{theorem:skewed-packing}) in \Cref{sec:skewed}. We will choose
two constants $\epssmall,\epslarge$ satisfying $0<\epssmall<\epslarge\le\epsilon^2$
according to which we classify an item $i$ as
\begin{itemize}
\item \textit{small} if $w(i)\le\epssmall N$ and $h(i)\le\epssmall N$,
\item \textit{large} if $w(i)>\epslarge N$ and $h(i)>\epslarge N$,
\item \textit{horizontal} if $w(i)>\epslarge N$ and $h(i)\le\epssmall N$,
\item \textit{vertical} if $h(i)>\epslarge N$ and $w(i)\le\epssmall N$,
and
\item \textit{intermediate} otherwise (i.e., $h(i)$ or $w(i)$ is in $(\epssmall N,\epslarge N]$).
\end{itemize}

Let $\OPT_{\text{small}},\OPT_{\text{large}},\OPT_{\mathrm{hor}},\OPT_{\mathrm{ver}}$,
and $\OPT_{\text{int}}$ denote the small, large, horizontal, vertical
and intermediate items in $\OPT$, respectively. 
We will later require that $\epssmall$ is much smaller than
$\epslarge$; formally, we will define a function $f:\R\rightarrow\R$
and require that $f(\epssmall)\le\epslarge$. Using the following
lemma, we can discard the items in $\OPT_{\text{int}}$ at negligible
loss in profit.
\begin{lemma}
[\cite{galvez2021approximating}] \label{lem:delete-intermediate-items}
For any constant $\epsilon>0$ and positive increasing function $f(\cdot)$
such that $f(x)>x$ for each $x\in(0,1)$, there exist values $\epslarge,\epssmall$
with $\epsilon^{2}\ge\epslarge\ge f(\epssmall)\ge \epssmall \ge \lambda_{\eps} \in \Omega_{\eps}(1)$, for some constant $\lambda_{\eps}$,
such that the total profit of items in $\OPT_{\mathrm{int}}$ is bounded
by $\epsilon \cdot p(\OPT)$. The pair $(\epssmall,\epslarge)$ is one pair
from a set of $O_{\epsilon}(1)$ pairs, and this set can be computed
in polynomial time.
\end{lemma}

We {\em guess}\footnote{By guessing, we mean that we \aw{try} all possible choices and \aw{finally output the best \aw{obtained} solution (over all choices).}} the pair $(\epssmall,\epslarge)$ according to Lemma~\ref{lem:delete-intermediate-items}. Our strategy is to partition the knapsack into $O_{\epsilon}(1)$ rectangular
containers, such that within each container the items are stacked
horizontally on top of each other, or vertically next to each other,
or all items are small in both dimensions compared to the container
and they are packed with the Next-Fit-Decreasing-Height (NFDH) algorithm (see \Cref{subsec:nf}). Formally, a \emph{container }is an open axis-parallel rectangle
$C\subseteq K$ with integral coordinates, together with the label
\emph{horizontal, vertical}, or \emph{area} \emph{container}\aw{; s}ee  \Cref{fig:containers}.

\begin{figure}
    \centering
    \includegraphics[width=0.7\linewidth]{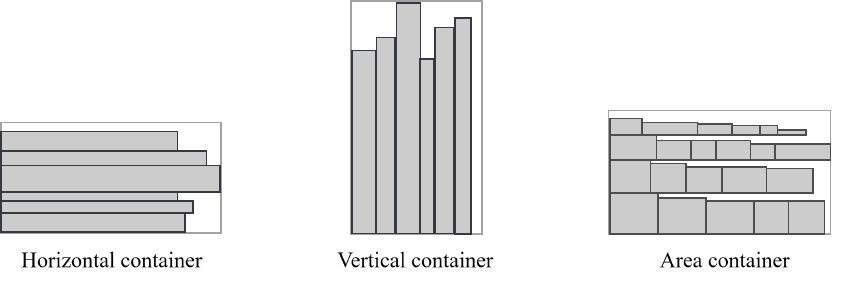}
    \caption{Types of containers.}
    \label{fig:containers}
\end{figure}

\begin{definition}[Container Packing]
Consider a packing of a set of items $I'\subseteq I$ inside $K$
and a set of containers $\C$. They form a \emph{container packing}
if the containers in $\C$ are pairwise disjoint and
\begin{itemize}
\item each item $i\in I'$ is contained in one container $C\in\C$,
\item inside each {\em horizontal} (respectively, {\em vertical}) container, the items are stacked
one on top of the other (respectively, one next to the other), and
\item if an item $i\in I'$ is packed in an {\em area} container of some height
$h$ and some width $w$, then $w(i)\le\epsilon\cdot w$ and $h(i)\le\epsilon\cdot h$.
\end{itemize}
\end{definition}

Container packings have been considered before in the literature \cite{Galvez2016, galvez2021approximating, Galvez00RW21, GalvezGAJKR23, Jansen0LS22}, 
and it is well-known that in polynomial time, we can essentially compute
the most profitable container packing with a constant number of containers
(even in the weighted case of our problem) via a reduction to the generalized assignment problem (GAP).

\begin{lemma}[\cite{galvez2021approximating}]
\label{lem:computer-container-packing}Let $\epsilon>0$ and $c\in\N$.
Consider an instance of the (weighted) two-dimensional knapsack problem
with rotations. Let $\OPT(c)$ denote the most profitable container
packing with at most $c$ containers. In time $n^{(c/\eps)^{O(1)}}$, 
we can compute a solution with a profit of at least $(1-\eps)p(\OPT(c))$.
\end{lemma}

Note that for a fixed value for $c$, it might be the case that $p(\OPT(c))$ is
significantly smaller than $p(\OPT)$. In particular, it is not clear
what approximation ratio one can achieve for our problem using container
packings with only a constant number of containers. Indeed, in \cite{galvez2021approximating}
it was shown that for the two-dimensional knapsack problem \emph{without}
rotations such container packings cannot yield a better approximation
factor than 2. However, we show that \emph{with} rotations we can
achieve an approximation ratio of $1+\eps$. If $|\OPT|\le O_{\eps}(1)$
then clearly there is a container packing with only $O_{\eps}(1)$ containers, 
since we can simply introduce one container for each item in $\OPT$. Hence,
it suffices to prove that there is a container packing that packs
$(1-\eps)|\OPT|-d_{\eps}$ items for some constant $d_{\eps}$ (yielding a $1+O(\eps)$-approximation
whenever $|\OPT|\ge d_{\eps}/\eps$).

\begin{lemma}[Resource contraction lemma]
\label{lem:profitable-container-packing}There are global constants
$c_{\eps}, d_{\eps}\in\N$ and $\mu_{\eps} >0$ such that for each instance of the cardinality case of 2DKR, there exists a container packing with
at most $c_{\eps}$ containers that packs at least $(1-\eps)|\OPT|-d_{\eps}$ items
such that each container is contained in $[0,N]\times [0, (1-\mu_{\eps})N]$.
\end{lemma}

\Cref{lem:computer-container-packing} and \Cref{lem:profitable-container-packing}
 yield Theorem~\ref{theorem:PTAS-cardinality-case}. We will prove Lemma~\ref{lem:profitable-container-packing}
in the remainder of this section.
As a first step, we
drop all items in $\OPT_{\text{large}}$ since already due to their
areas, there can be at most $O(1/\epslarge^{2})=O_{\epsilon}(1)$
of them. In the following, we will focus on the items of $\OPT_{\mathrm{hor}} \cup \OPT_{\mathrm{ver}}$;
we will add the small items in $\OPT$ at the very end with standard techniques, see e.g., \cite{galvez2021approximating}.
\begin{figure}
    \centering
    \includegraphics[width=0.9\linewidth]{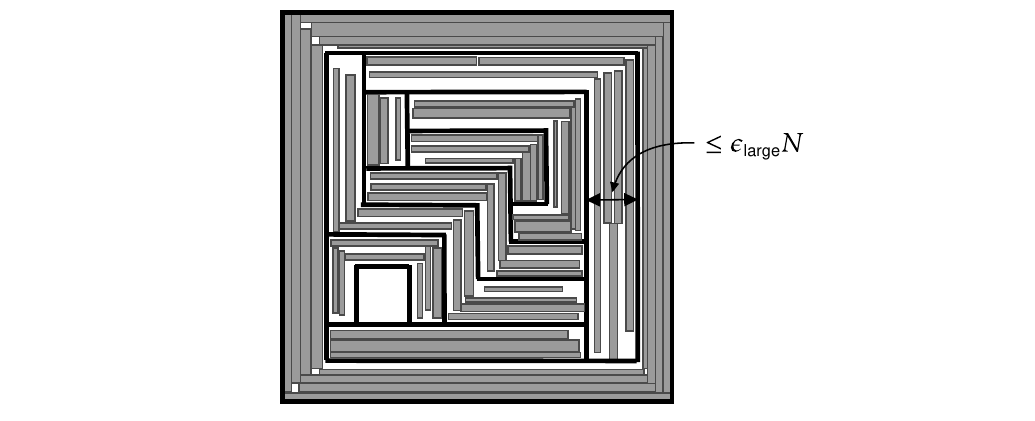}
    \caption{A partition of the knapsack into thin corridors, i.e., whose width is at most $\epslarge N$.}
    \label{fig:corridor-decomposition}
\end{figure}

Applying the \emph{corridor decomposition framework} (Lemma \ref{lem:corr-decomposition-cardinality}) to the packing of $\OPT_{\mathrm{hor}} \cup \OPT_{\mathrm{ver}}$, we obtain a packing of at least $(1-\eps)|\OPT_{\mathrm{hor}} \cup \OPT_{\mathrm{ver}}|$ items,
in which each item is contained in one out of $O_{\eps, \epslarge}(1)$ corridors, each having width at most $\epslarge N$, see Figure~\ref{fig:corridor-decomposition}. In particular, this implies that horizontal items can be placed only
in the horizontal parts of a corridor and the vertical items only
in the vertical parts. However, in the bend of a corridor there
may be both horizontal and vertical items and they may be packed in
a very complicated manner, see Figure~\ref{fig:corridor-decomposition}.

In \cite{galvez2021approximating}, a method was presented that takes a
packing of items of $\OPT_{\mathrm{hor}}\cup \OPT_{\mathrm{ver}}$ inside the $O_{\eps, \epslarge}(1)$ corridors as described above, removes
some of the items, and constructs a container packing for the remaining
items. More precisely, among the removed items, there are $O_{\eps, \epslarge}(1)$
items that we could simply discard since they are only constantly
many; for another subset we can bound the number of items by $O(\eps)\cdot |\OPT_{\mathrm{hor}}\cup\OPT_{\mathrm{ver}}|$
which is also small.
For the remaining removed items (the set $I'_{T}$ in the lemma below) we cannot bound their total number unfortunately; however, they
fit in a thin strip of size $[0,N]\times[0,\epsthin N)$
for some chosen parameter $\epsthin>0$. This parameter
may be chosen arbitrarily small. However, the resulting number of
containers may depend on $\epsthin$ and this number might
grow \emph{polynomially} with $1/\epsthin$.

\begin{lemma}[\aw{implicit in} \cite{galvez2021approximating}]
\label{lem:transformation-container-packing-old}
For any choice for the parameters $\epsilon>0$ and $\epsilon_{\mathrm{large}}>0$ there is a value $\Gamma(\epsilon, \epslarge)\in \N$ 
such that for any $\epsthin>0$ there exist sets $I',I''\subseteq \OPT_{\mathrm{hor}}\cup\OPT_{\mathrm{ver}}$ with $p(I')+p(I'')\ge(1-\epsilon)p(\OPT_{\mathrm{hor}}\cup\OPT_{\mathrm{ver}})$ and $|I''|=O_{\eps,\epslarge,\epsthin}(1)$
such that
\begin{itemize}
\item there exists a subset $I'_{P}\subseteq I'$ for which there is a container
packing with at most $(\frac{1}{\epsthin})^{\Gamma(\epsilon, \epslarge)}$
containers
such that each container is labeled horizontal or vertical, 
and
\item for the remaining items $I'_{T}=I'\setminus I'_{P}$ there exists
a packing inside $[0,N]\times[0,\epsthin N)$ (when
rotations are allowed).
\end{itemize}
\end{lemma}

We want to improve the dependence on $1/\epsthin$
in Lemma~\ref{lem:transformation-container-packing-old}. Therefore,
we present an alternative transformation of a given $(1+\eps)$-approximate
packing within $O_{\eps, \epslarge}(1)$ corridors to a container packing such
that the number of containers has only a \emph{polylogarithmic} dependence
on $1/\epsthin$. This improvement may seem minor;
however, we will see later that it makes a crucial difference in the
remainder of our reasoning. Our adjustment is inspired by (another)
transformation of a packing within corridors to a container packing
with at most $(\log N)^{O_{\eps}(1)}$ containers \cite{adamaszek2015knapsack} (which would be too much for our purposes though).

\begin{lemma}
\label{lem:transformation-container-packing-new}
For any choice for the parameters $\epsilon>0$ and $\epsilon_{\mathrm{large}}>0$ there is a value $\Gamma'(\epsilon, \epsilon_{\mathrm{large}})\in \N$
such that for any $\epsthin>0$ there exist sets $I',I''\subseteq \OPT_{\mathrm{hor}}\cup\OPT_{\mathrm{ver}}$ with $p(I')+p(I'')\ge(1-\epsilon)p(\OPT_{\mathrm{hor}}\cup\OPT_{\mathrm{ver}})$ and $|I''|=O_{\eps,\epslarge,\epsthin}(1)$ such that
\begin{itemize}
\item there exists a subset $I'_{P}\subseteq I'$ for which there is a container
packing with at most $(\log\frac{1}{\epsthin})^{\Gamma'(\epsilon, \epsilon_{\mathrm{large})}}$
containers such that each container is labeled horizontal or vertical, and
\item for the remaining items $I'_{T}=I'\setminus I'_{P}$ there exists
a packing inside $[0,N]\times[0,\epsthin N)$ (when
rotations are allowed).
\end{itemize}
\end{lemma}

Before proving \Cref{lem:transformation-container-packing-new}, we first show the following lemma which states that any packing of a subset of $\OPT_{\mathrm{hor}}$ (resp. $\OPT_{\mathrm{ver}}$) inside a rectangular region can be converted to a container packing using only horizontal (resp. vertical) containers, by discarding a negligible fraction of items.

\begin{lemma}
\label{lem:box-to-container}
    Let $\delta>0$ and let $I_B \subseteq \OPT_{\mathrm{hor}}$ be a set of items packed inside a box $B$. Then there exist subsets $I'_B, I''_B \subseteq I_B$ with $p(I'_B)+p(I''_B)\ge (1-\delta)p(I_B)$ and $|I''_B|=O_{\delta,\epslarge}(1)$ such that $I'_B$ can be packed into $O_{\delta}(1)$ horizontal containers all lying inside $B$. An analogous claim holds when $I_B \subseteq \OPT_{\mathrm{ver}}$.
\end{lemma}
\begin{proof}
    Let $w(B), h(B)$ denote the width and height of Box $B$. We draw $1/\delta -1$ equidistant horizontal lines inside $B$, partitioning the interior of $B$ into $1/\delta$ strips, each of width $w_B$ and height $\delta\cdot h(B)$. Since the width of each item of $\OPT_{\text{hor}}$ is at least $\epslarge N$, each line intersects at most $1/\epslarge$ horizontal items. Thus the total number of intersected items is bounded by $\frac{1}{\delta\cdot \epslarge}=O_{\delta,\epslarge}(1)$, which we assign to the set $I''_B$. Finally, we discard all items lying in the minimum profitable strip among the $1/\delta$ strips, thereby discarding a profit of at most $\delta\cdot p(I_B)$. This frees up a strip of height $\delta\cdot h(B)$ inside $B$, which we use for resource augmentation in order to obtain a container packing of profit at least $(1-O(\delta))p(I_B)$ using Lemma \ref{lem:resource-augmentation}. 
    
    The case for $I_B \subseteq \OPT_{\mathrm{ver}}$ can be shown analogously. This completes the proof. 
\end{proof}

We are now ready to prove \Cref{lem:transformation-container-packing-new}.

\begin{proof}[Proof of \Cref{lem:transformation-container-packing-new}]    
We apply Lemma \ref{lem:corr-decomposition-cardinality} to the optimal packing of the items in $\OPT_{\mathrm{hor}} \cup \OPT_{\mathrm{ver}}$, thus obtaining the sets $\OPT_{\text{corr}}$ and $\OPT_{\text{corr}}^{\text{cross}}$ with $|\OPT_{\text{corr}}^{\text{cross}}|=O_{\epsilon,\epslarge}(1)$, and a collection of $O_{\epsilon,\epslarge}(1)$ corridors in which the items of $\OPT_{\text{corr}}\setminus\OPT_{\text{corr}}^{\text{cross}}$ are packed. We assign the items of $\OPT_{\text{corr}}^{\text{cross}}$ to the set $I''$. Note that since the width of the corridors is at most $\epslarge N$, the horizontal subcorridors do not contain any vertical item and vice versa.

We now {\em process} the corridors into boxes. First, we discuss the processing for {\em open corridors} (see \Cref{fig:open-corridor}). Let $C_1$ be the first subcorridor of an open corridor $C$ and w.l.o.g.~assume that $C_1$ is horizontal with the shorter horizontal edge being the top one, and the next bend lying to the right of $C_1$ (see Figure \ref{fig:process}). Let $h$ be the height of $C_1$. For each $j=0,1,2,\ldots$ we draw a horizontal line at distance 
$\epstiny(1+\epsilon)^jh$ from the bottom horizontal edge of $C_1$. This partitions the private region of $C_1$ into $O(\frac{1}{\epsilon}\log\frac{1}{\epsthin})$ regions as shown in Figure \ref{fig:process}. We assign the items of $\OPT_{\mathrm{hor}}$ intersected by these lines to $I''$, noting that they are only $O_{\epsilon,\epsilon_{\text{large}},\epsilon_{\text{thin}}}(1)$ many.

Let $\RR_1,\RR_2,\ldots$ denote these regions from bottom to top. We mark as \textit{thin} the items in $\RR_1$, remove them from the packing and assign them to the set $\optthin$.
For each $j$, let $h(\RR_j)$ be the height of region $\RR_j$ and $w_j$ be the width of the top horizontal edge of $\RR_j$. Due to the geometrically increasing heights, notice that $h(\RR_{j+1})\le (1+\epsilon)h(\RR_j)$, for all $j\ge 1$, and due to the monotonic boundary curve between subcorridors (see \Cref{fig:boundarycurve}), $w_{j} \ge w_{j+1}$, for all $j \ge 1$.  Thus, $\RR_j$ is completely contained in a box of height $h(\RR_j)$ and width $w_{j-1}$; also, within $\RR_{j-1}$ we can place a rectangular box of height $h(\RR_j){/(1+\eps)\ge (1-\eps)h(\RR_j)}$ and width $w_{j-1}$, for each $j=2, 3, \dots$.
Hence if we are able to free up a horizontal strip of height $\epsilon\cdot h(\RR_{j+1})$ inside $\RR_{j+1}$, the remaining items in $\RR_{j+1}$ can be packed inside a box $B_j$ of width $w_j$ and height $(1-\epsilon)h(\RR_{j+1}) < h(\RR_j)$. 
This box can be completely placed inside the region $\RR_j$. 
To this end, we draw $1/\epsilon-1$ equidistant horizontal lines inside $\RR_{j+1}$ and put the at most $O_{\epsilon,\epsilon_{\text{large}}}(1)$ items of $\OPT_{\mathrm{hor}}$ intersected by these lines into $I''$.
These lines divide $\RR_{j+1}$ into $1/\epsilon$ strips, each of height $\eps\cdot h(\RR_{j+1})$. Among these strips, the strip that contains items with the least profit is called the {\em minimum profitable strip}. 
We then delete all items inside the minimum profitable strip, thereby discarding items whose profit is at most an $\epsilon$-fraction of the total profit of the items in $\RR_{j+1}$. 
Then we get packing of the remaining items in $\RR_{j+1}$ in a box of height $(1-\eps)h(\RR_{j+1})$ and width at most $w_j$. 
Altogether, we obtain a collection $\RR(C_1)$ of $O(\frac{1}{\epsilon}\log\frac{1}{\epstiny})$ boxes corresponding to $C_1$.

Next, we partition the remaining subcorridors. Let $C_2, C_3,\ldots$ be the sequence of subcorridors next to $C_1$. For each $B_j \in \RR(C_1)$, we create a path $p_j$ starting at the top right corner of $B_j$. The first part of $p_j$ is a vertical line segment completely inside $C_2$ joining the top right corner of $B_j$ and extending till it intersects the boundary curve between $C_2$ and $C_3$; from there we construct a horizontal line segment inside $C_3$ extending until it intersects the boundary curve between $C_3$ and $C_4$, and so on.
These collection of paths $\{p_j\}_{j: B_j \in \RR(C_1)}$ yields a decomposition of the region of $C$ not occupied by the boxes $\RR(C_1)$
into $O(\frac{1}{\epsilon}\log\frac{1}{\epstiny})$ smaller corridors each with one bend less. We assign all items of $\OPT_{\mathrm{hor}}\cup \OPT_{\mathrm{ver}}$ intersected by the paths $\{p_j\}_{j: B_j \in \RR(C_1)}$  to the set $I''$ noting that there are only $O_{\epsilon,\epsilon_{\text{large}},\epsilon_{\text{thin}}}(1)$ of them.
We continue this processing recursively, removing the thin items in each subcorridor (which we assign to the set $\optthin$) and obtaining boxes from the subcorridors.

\begin{figure}
    \centering
    \includegraphics[width=0.7\linewidth]{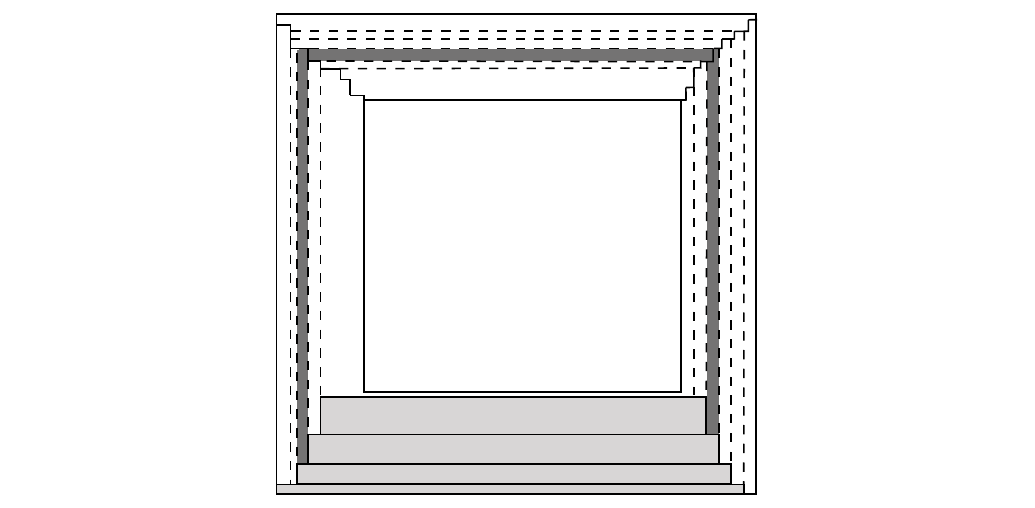}
    \caption{Processing a closed corridor.}
    \label{fig:closed-process}
\end{figure}

Now we discuss the processing of closed corridors (see \Cref{fig:closed-corridor}). For a closed corridor $C$, we let $C_1$ be the bottom-most horizontal subcorridor of $C$, i.e., the one with the lowest $y$-coordinate of the bottom edge. As in the case of an open corridor, we draw lines at geometrically increasing heights inside $C_1$, mark the items in the bottommost region as thin (i.e., put them in $\optthin$), and obtain a collection $\RR(C_1)$ of $O(\frac{1}{\epsilon}\log\frac{1}{\epstiny})$ boxes. For each box $B_j \in \RR(C_1)$, we construct two paths $p_j^{\text{left}}$ and $p_j^{\text{right}}$, starting from the top left and top right corners of $B_j$, respectively, similar to the construction of the path $p_j$ described before. The path $p_j^{\text{left}}$ (resp. $p_j^{\text{right}}$) ends after hitting the right (resp. left) edge of some box in $\RR(C_1)$.
These paths form a partition of the area of $C$ outside the boxes $\RR(C_1)$ into $O(\frac{1}{\epsilon}\log\frac{1}{\epstiny})$ smaller open corridors (see Figure~\ref{fig:closed-process}), to which we then apply the corridor processing described previously for open corridors.

Since each corridor has at most $1/\epsilon$ bends, the number of boxes obtained from each corridor is $(\frac{1}{\epsilon}\log\frac{1}{\epstiny})^{O(1/\epsilon)}\le (\log\frac{1}{\epstiny})^{O(1/\epsilon^2)}$, as $1/\eps \le (\log\frac{1}{\epstiny})^{O(1/\epsilon)}$. By Lemma \ref{lem:corr-decomposition-cardinality}, the number of corridors was $O_{\epsilon,\epslarge}(1)$, and therefore the total number of boxes in the knapsack can be bounded by $(\log \frac{1}{\epsthin})^{O_{\epsilon,\epslarge}(1)}$. 

Next, we apply \Cref{lem:box-to-container} with $\delta=\epsilon$ to each box $B$ thereby obtaining a packing of a set $I'_P$ into $(\log \frac{1}{\epsthin})^{O_{\epsilon,\epslarge}(1)}$ containers. We define $I'=I'_P\cup I'_T$. For each box $B$ we assign the items in the set $I''_B$ to $I''$, noting that they are only $O_{\eps,\epslarge,\epsthin}(1)$ many. It remains to show that the removed thin items, i.e., all items we assigned to $I'_T$, can all be packed inside a strip of sufficiently small width. To this end, observe that in any subcorridor, the items of $I'_T$ come from a thin strip whose area is at most an $\epsthin$-fraction of the area of the subcorridor.
Since the total area of all subcorridors is at most $2N^2$, the total area of the items of $I'_T$ is bounded by $2\epsthin N^2$. We orient each item so that its longer side is horizontal. Since by Lemma \ref{lem:corr-decomposition-cardinality} the height $h$ of each horizontal subcorridor is bounded by $\epslarge N$ and each item in $I'_T$  fits
in a box of height at most $\epsthin h$, the height of each item is trivially at most $\epsthin\cdot\epslarge N < \epsthin N$. Therefore, using an algorithm of Steinberg \cite{steinberg1997strip} (see \Cref{subsec:stein} for details on Steinberg's algorithm), the items of $I'_T$ can all be packed inside a strip of width $N$ and height $4\epsthin N$. Lemma \ref{lem:transformation-container-packing-new} follows by rescaling $\epsthin$ (i.e., replacing $\epsthin$ by $\epsthin/4$).
\end{proof}

\begin{figure}
    \centering
    \includegraphics[width=\linewidth]{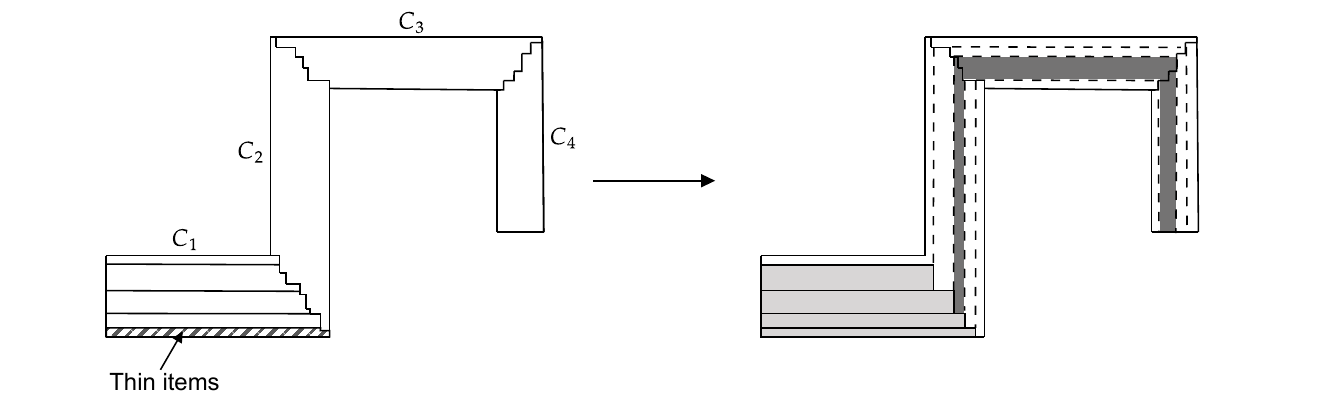}
    \caption{Processing a corridor into containers.}
    \label{fig:process}
\end{figure}

We invoke Lemma~\ref{lem:transformation-container-packing-new} with
$\epsthin:=(\eps/2)^{(2/\eps)^{\frac{\Gamma'(\eps,\epslarge)}{\eps\cdot \epslarge}}}$. Consider the sets $I'$, $I'_{P}$,
and $I'_{T}$ due to Lemma~\ref{lem:transformation-container-packing-new}
and let~$\C$ denote the at most $(\log\frac{1}{\epsthin})^{\Gamma'(\epsilon, \epsilon_{\mathrm{large}})}$
containers of the corresponding container packing for $I'_{P}$. Our
goal is to transform this container packing further such that we can
add the items in $I'_{T}$ to this packing. In the process, we will sacrifice at most $\eps\cdot|I'_{P}|\le O(\eps)\cdot |\OPT|$ items from $I'_{P}$.

We classify the containers in $\C$ as thick and thin containers. For a container $C\in\C$, we denote by $h(C)$
and $w(C)$ its height and width, respectively. We will choose two
constants $\epscsmall,\epsclarge$ with $0<\epscsmall<\epsclarge\le\epsilon$
and define that
\begin{itemize}
\item a horizontal container $C\in\C$ is \emph{thick} if $h(C)\ge\epsclarge N$
and \emph{thin} if $h(C)\le\epscsmall N$,
\item a vertical container $C\in\C$ is \emph{thick} if $w(C)\ge\epsclarge N$
and \emph{thin} if $w(C)\le\epscsmall N$,
\item a container $C\in\C$ is \emph{intermediate }if it is neither thick
nor thin.
\end{itemize}
The next lemma shows that there are choices for $\epscsmall,\epsclarge$
such that the intermediate containers have only very few items (similar
to Lemma~\ref{lem:delete-intermediate-items}).
We will need later that all thin horizontal (respectively, vertical) containers together have a total height (respectively, width) of at most $\Theta(\epsilon) \cdot \epsclarge N$, which is why we require that $\epsclarge/\epscsmall\ge 3|\C|/\epsilon$. This can be achieved with a standard shifting argument. However, we will also need that $\epsthin \le \frac{\eps}{6}\cdot \epsclarge$, which is why we need to choose $\epsthin$ sufficiently small. However, a smaller choice for $\epsthin$ implies a larger bound for $|\C|$ which, due to the shifting step, implies that $\epsclarge$ might be smaller; but then, the bound $\epsthin \le \frac{\eps}{6}\cdot \epsclarge$ might not be satisfied anymore. Our improved bound of
 $|\C|\le(\log\frac{1}{\epsthin})^{\Gamma'(\epsilon, \epsilon_{\mathrm{large})}}$ from Lemma~\ref{lem:transformation-container-packing-new} will ensure that we can satisfy our requirements despite this circular dependence of the mentioned quantities.

\begin{lemma}
\label{lem:weighted-classification}
    There exist values $\epscsmall,\epsclarge$,
with $\epslarge\ge\epsclarge>\epscsmall\ge6\epsthin/\epsilon$, 
and $\epsclarge/\epscsmall\ge 3|\C|/\epsilon$,
such that the total profit of items in intermediate containers is
at most $\epsilon\cdot p(\OPT)$.
For any choice of the parameters $\eps, \epslarge$, and $\epsthin$, there is a
global set of $O_{\epsilon}(1)$ pairs that we can compute in polynomial time and that contains the pair
$(\epscsmall,\epsclarge)$.
\end{lemma}
\begin{proof}
Let $k:= \frac{3}{\eps}(\log \frac{1}{\epstiny})^{\Gamma'(\epsilon,\epslarge)} \ge \frac{3}{\eps}|\CC|$. Since the width of each subcorridor was at most $\epslarge N$, the maximum height (resp.~width) of a horizontal (resp.~vertical) container is also at most $\epslarge N$. For $j\in \mathbb{N}$, let $\CC_j \subseteq \CC$ consist of horizontal containers having height in the range $(\frac{\epslarge}{k^{j}}N,\frac{\epslarge}{k^{j-1}}N]$ and vertical containers having width in the range $(\frac{\epslarge}{k^{j}}N,\frac{\epslarge}{k^{j-1}}N]$, for $j=1,2,\dots$. For convenience, let $p(\CC_j)$ denote the total profit of items packed in the containers of $\CC_j$. Since the sets $\{\CC_j\}_{j\in \mathbb{N}}$ are disjoint, there must exist a $j^*\in [1/\epsilon]$ such that $p(\CC_{j^*}) \le \epsilon\cdot p(OPT)$. We define $\epsclarge:= \epslarge/k^{j^*-1}$ and $\epscsmall := \epslarge/k^{j^*}$. Clearly then $\epsclarge/\epscsmall = k \ge 3|\CC|/\epsilon$, and since $j^* \le 1/\epsilon$, we have $\epscsmall \ge \epslarge/k^{1/\eps} = \left. \epslarge\middle/\left(\frac{3}{\eps}(\log \frac{1}{\epstiny})^{\Gamma'(\eps,\epslarge)}\right)^{1/\eps}\right.\ge 6\epstiny/\epsilon$, where the last inequality follows by our choice of $\epsthin$.
\end{proof}

We discard all intermediate containers in $\C$ and sacrifice the
items contained in them. \emph{Temporarily}, we also remove the thin
containers in $\C$ and their corresponding items. However, we will
later put their items back in our packing, together with the items
in $I'_{T}$.

Let $\Cl\subseteq\C$ denote the thick containers in $\C$. Our goal
is to sacrifice a few of their items in order to free up an empty strip
at the top edge or at the right edge of the knapsack. We will use
this strip to pack the items from the thin containers and also the
items in $I'_{T}$. For this, we will 
\emph{shrink} \aw{the} thick
containers. This means that we reduce their respective sizes in the
``thick'' dimension by a factor of $1-\eps$. In the next lemma,
we show that we can do this by discarding only constantly many items
packed inside a thick container and reducing the number of the remaining
items by a factor of $1-\eps$. The lemma is formulated for horizontal
containers, but by rotating items a symmetric statement immediately
holds also for vertical containers.
\begin{lemma}
\label{lem:weighted-shrinking}
    Let $\delta>0$. Consider a container packing
which includes a horizontal container~$C$ of height $h(C)$
and width $w(C)$ containing a set of items $I_{C}\subseteq I$. There exist subsets $I_{C}', I''_C\subseteq I_{C}$ with $p(I_{C}')+p(I''_C)\ge(1-\delta)p(I_{C})$ and $|I''_C|=O_{\delta,\epslarge}(1)$
such that $I'_{C}$ can be packed inside a horizontal container
$C'$ of height $h(C')=(1-\delta)h(C)$ and width $w(C')=w(C)$.
\end{lemma}
\begin{proof}
The proof is very similar to the proof of Lemma \ref{lem:box-to-container}. We construct $1/\delta - 1$ horizontal lines inside $C$ that partition the interior of $C$ into strips of height $\delta\cdot h(C)$. We include the horizontal items intersected by these lines in the set $I''_C$, noting that there are only $\frac{1}{\delta\cdot \epslarge}=O_{\delta,\epslarge}(1)$ of them. By the pigeonhole principle, one of the $1/\delta$ strips must contain a profit of at most $\delta\cdot p(I_C)$. By discarding the items lying in this strip, the remaining items $I'_C$ can be packed inside a shrunk container $C'$ of height $(1-\delta)h(C)$.
\end{proof}

We apply Lemma~\ref{lem:weighted-shrinking} with $\delta:=\eps$ to
each (horizontal or vertical) container $C\in\Cl$ and replace
the container by the corresponding container $C'$. Then, we push
all resulting containers down as much as possible such that they do
not intersect. Then, similarly, we push them to the left as much as
possible. Let $\Cl'$ denote the resulting containers.

\begin{figure}
    \centering
    \hspace{1.8cm}
    \includegraphics[width=0.5\linewidth]{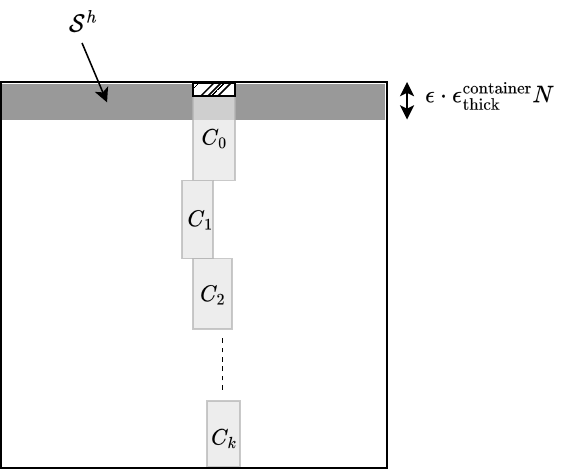}
    \caption{A chain of vertical containers.}
    \label{fig:chain}
\end{figure}

If we shrink a horizontal container $C\in\Cl$ and replace
it by the corresponding smaller container $C'$, then the \emph{absolute
}amount by which we shrink $C$ equals $h(C)-h(C')=\eps\cdot h(C)\ge\stripthickness\cdot N$
(similarly for the widths of vertical containers). Hence, after pushing
all shrunk containers down and to the left, one may hope that one
of the strips
\begin{itemize}
\item $\SM^{h}:=[0,N]\times[N-\stripthickness\cdot N,N]$\aw{, i.e.,} of height $\stripthickness \aw{\cdot N}$
at the top of \aw{$K$} and
\item $\SM^{v}:=[N-\stripthickness\cdot N,N]\times[0,N]$\aw{, i.e.,}  of width $\stripthickness\aw{\cdot N}$
on the right of \aw{$K$}
\end{itemize}
does not intersect with any container in $\Cl'$. Then, we could use
$\SM^{h}$ or $\SM^{v}$ to put back the items from the thin containers
and also items from $I'_{T}$. In the next lemma, we show that if the strip $\SM^h$ is not empty, then there must exist a ``chain''
of tightly stacked \aw{thick} vertical containers that reach from $\SM^{h}$
to the bottom of the knapsack, see Figure \ref{fig:chain}.


\begin{lemma}
\label{lem:chain-of-boxes}
    If $\SM^h$ intersects some container, then there exists a \textit{chain} of vertical containers $C_0,C_1,\ldots, C_k \in \CC'_{\text{thick}}$ for some integer $k$ such that
    \begin{enumerate}[label=(\roman*)]
        \item $C_0$ intersects $\SM^h$, and the top edge of $C_0$ does not touch the boundary of any other container;
        \item the top edge of $C_j$ touches the bottom edge of \ak{$C_{j-1}$}, for all $j\in [k]$;
        \item the bottom edge of $C_k$ touches the bottom boundary of the knapsack.
    \end{enumerate}
\end{lemma}
\begin{proof}
Since $\SM^h$ has a height of $\stripthickness N$ and the height of each horizontal container in $\CC_{\text{thick}}$ was shrunk by at least $\stripthickness N$, no horizontal container in $\CC'_{\text{thick}}$ can intersect $\SM^h$. Therefore, all containers intersecting $\SM^h$ must be vertical.

Let $C_0$ be any vertical container intersecting $\SM^h$ (see \Cref{fig:chain}). \ak{Let the two endpoints of the top edge of $C_0$ be $(l_0, y_0)$ and $(r_0,y_0)$ for $0 \le l_0\le r_0 \le N $ and $y_0 \in \aw{[N-\stripthickness N,N]}$.}  Since each vertical item has a height of at least $\epslarge N$, the height of any vertical container must be at least $\epslarge N$, and therefore no other container intersects the area above $C_0$ inside $\SM^h$, \ak{i.e., the region $[0,N]\times [y_0, N]$}. Now, if the bottom edge of $C_0$ touches the bottom boundary of the knapsack, we set $k=0$ and are done. Else, let $\CC_1 \subseteq \CC'_{\text{thick}}$ be the set of \ak{vertical} containers whose top edge touches the bottom edge of $C_0$. If $\CC_1 = \emptyset$, then the bottom edge of $C_0$ only touches horizontal containers, and since the height of each horizontal container shrank by at least $\stripthickness N$, $C_0$ must have been pushed down by a distance of at least $\stripthickness N$, contradicting the fact that $C_0$ intersects $\SM^h$. Hence, we have $\CC_1\neq \emptyset$. Again, if the bottom edge of some container in $\CC_1$ touches the bottom boundary of the knapsack, we set $k=1$ and are done. Otherwise, let $\CC_2 \subseteq \CC'_{\text{thick}}$ be the set of vertical containers whose top edge touches the bottom edge of some box in $\CC_1$. By a similar argument as before, if $\CC_2$ were to be empty, $C_0$ together with all boxes in $\CC_1$ would have been pushed down by a height of at least $\stripthickness N$, implying that $C_0$ would no longer intersect $\SM^h$, a contradiction. Therefore $\CC_2 \neq \emptyset$. If there exists some $C_2 \in \CC_2$ such that the bottom edge of $C_2$ touches the bottom boundary of the knapsack, then by definition of $\CC_2$, there exists a $C_1 \in \CC_1$ such that $C_0, C_1, C_2$ satisfy the conditions of the lemma, and we are done. Otherwise, we continue the above process until we obtain a collection of vertical containers $\CC_k \subseteq \CC'_{\text{thick}}$ such that the bottom edge of some container in $\CC_j$ touches the bottom boundary of the knapsack. Then we can find $C_j \in \CC_j$ for each $j\in [k]$ satisfying conditions (ii) and (iii) of the lemma, and we are done. \ak{As each vertical container has height \dk{at least} $\epslarge N$, we have $j \in [1/\epslarge]$. }
\end{proof}

Now it is easy to show that one of the strips $\SM^h$ or $\SM^v$ is empty.

\begin{lemma}
\label{lem:free-strip-at-boundary}
At least one of the strips $\SM^{h}$ and $\SM^{v}$ does not intersect
with any container in $\Cl'$.
\end{lemma}
\begin{proof}
    Suppose that $\SM^h$ intersects some container in $\CC'_{\text{thick}}$, otherwise we are done. As discussed in the proof of the previous lemma, all containers intersecting $\SM^h$ must be vertical. Consider the chain of vertical containers $C_0, C_1,\ldots, C_k$ guaranteed by \Cref{lem:chain-of-boxes}. Since the width of each of these containers shrank by at least $\stripthickness N$, after pushing all containers as much to the left as possible, the strip $\SM^v$ of width $\stripthickness N$ at the right boundary of the knapsack cannot intersect any container in $\CC'_{\text{thick}}$, and we are done.
\end{proof}

W.l.o.g., we assume that $\SM^{h}$ does not intersect with any container
in $\Cl'$. It turns out that $\SM^{h}$ is large enough to accommodate
all items in the thin containers in $\C$. \ak{This is true because the sum of 
the total height of \aw{the thin horizontal containers in $\C$ and the total width of the thin vertical containers in $\C$} is bounded by
$|\C|\cdot\epscsmall N\le \frac{1}{3} \epsilon\cdot \epsclarge N$}, where the inequality follows from Lemma \ref{lem:weighted-classification}. Recall that the items in $I'_{T}$ can be packed
in a thin area of width $N$ and height $\epsthin N\le\frac{1}{6} \epsilon\cdot \epsclarge N$
(see Lemma~\ref{lem:transformation-container-packing-new}), where the inequality again follows from Lemma \ref{lem:weighted-classification}. Hence,
we can easily find a container packing for the items in $I'_{T}$
with $O_{\eps}(1)$ containers inside one third of $\SM^{h}$ (i.e.,
twice as much space as the items in $I'_{T}$ would need) and pack
the thin containers in $\C$ into \aw{the second} third of $\SM^{h}$.
We leave the remaining third of $\SM^{h}$ empty, which has a width of at least $2\epsthin N$. 
Hence, we choose $\mu_{\eps}$ such that $\mu_{\eps} \le 2\epsthin $. Note that $\epsthin$ depends on $\epslarge$ which 
is defined via Lemma~\ref{lem:delete-intermediate-items}. However, since $\epslarge\ge \lambda_{\eps}$
we can show that for $\mu_{\eps}:= (\eps/2)^{(2/\eps)^{\frac{\Gamma'(\eps,\lambda_{\eps})}{\eps\lambda_{\eps}}}}$ 
the area $[0,N]\times ((1-\mu_{\eps}) N,N]=:K_{\mathrm{empty}}$ does not intersect any container. Thus, we obtain the following lemma.

\begin{lemma}
\label{lem:repack-missing-items}There exists a container packing
with a set of containers $\C''$ with $|\C''|\le|\C|+O_{\eps,\epslarge,\epsthin}(1)$
that packs all items in $I'_{T}$ and all items that are packed in
thin containers in $\C$, such that each container $C\in\C''$ is
contained in $\SM^{h} \setminus K_{\mathrm{empty}}$.
\end{lemma}

It remains to add the (small) items from $\OPT_{\text{small}}$ to our packing. 

\paragraph{Packing small items.}
In this subsection, we pack the items of $\OPT_{\text{small}}$. For this, we first show the following lemma, which enables us to bound the total wasted area inside the horizontal and vertical containers.

\begin{lemma}
\label{lem:split-container}
    Let $\delta >0$ and let $I_C \subseteq \OPT_{\mathrm{hor}}$ (resp. $\OPT_{\mathrm{ver}}$) be a set of items packed inside a horizontal (resp. vertical) container $C$. There exists a set $I'_C \subseteq I_C$ with $|I'_C|\ge |I_C|-O_{\delta,\epslarge}(1)$ and a collection of horizontal (resp. vertical) containers $C_1, C_2,\ldots,C_{1/\delta}$ all lying inside $C$, that together pack items of $I'_C$, such that the total area inside $\bigcup_{j\in [1/\delta]} C_j$ not occupied by any item is \ak{at most}  $\delta \cdot a(C)$. Here $a(C)$ denotes the area of the container $C$.
\end{lemma}

\begin{proof}
\begin{figure}
    \centering
    \includegraphics[width=1.1\linewidth]{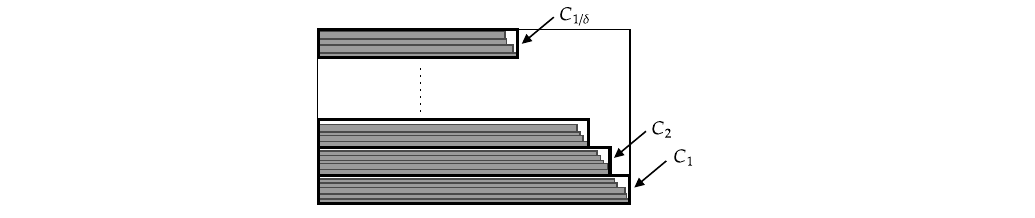}
    \caption{\dk{Splitting a container into multiple containers.}}
    \label{fig:split-containers}
\end{figure}
    Let $C$ be a horizontal container. Assume that the items $I_C$ are packed in non-decreasing order of width from top to bottom, and they are pushed to the left (and also to the bottom of the container as much as possible) so that all of them touch the left boundary of $C$. Also assume w.l.o.g.~that the height of $C$ equals the sum of the heights of the items of $I_C$, i.e., there is no empty horizontal region spanning the width of $C$. We draw $1/\delta - 1$ equidistant horizontal lines inside $C$ that partition the interior of $C$ into strips of height $\delta\cdot h(C)$. We discard the at most $\frac{1}{\delta\cdot \epslarge} =O_{\delta,\epslarge}(1)$ horizontal items intersected by these lines. Let $\RR_1,\RR_2,\ldots, \RR_{1/\delta}$ denote these strips from bottom to top. Inside each $\RR_j$, we create a container $C_j$ of height $\delta\cdot h(C)$ (i.e., equal to the height of $\RR_j$) and whose width equals the maximum width of an item packed inside $\RR_j$ (see Figure \ref{fig:split-containers}). Notice then that by our construction, for $j\in \{1,\ldots,1/\delta-1\}$, the box of width $w(C_{j+1})$ and height $h(C_j)$ inside $C_j$, touching the left boundary of $C_j$ is completely filled with items. Therefore, the free area inside $C_j$ is at most $(w(C_j)-w(C_{j+1}))\cdot\delta h(C)$. Summing over $j$, the total free area inside the containers $\bigcup_{j=1}^{1/\delta-1} C_j$ is at most $(w(C_1)-w(C_{1/\delta}))\cdot \delta h(C)$. 
    The free area inside the container $C_{1/\delta}$ is trivially upper bounded by $w(C_{1/\delta})\cdot \delta h(C)$. 
    Since $w(C_1)=w(C)$, the total free area within $\bigcup_{j=1}^{1/\delta} C_j$ thus is at most $\delta h(C)\cdot w(C) = \delta \cdot a(C)$, and we are done.
\end{proof}

Using the following lemma, we are able to free up enough space in the knapsack outside of the horizontal and vertical containers, which will enable us to pack almost all the small items. Recall that we require the region $K_{\text{empty}}= [0,N]\times ((1-\mu_{\eps}) N, N]$ inside $K$ should not intersect any container. Also $\delta \le 2\epsthin$, and therefore the area of the region $K_{\text{empty}}$ is at most $2\epsthin N^2$.

\begin{lemma}
\label{lem:final-skwed-packing}
    There exist subsets $\OPT'_{\text{skew}}, \OPT''_{\text{skew}} \subseteq \OPT_{\mathrm{hor}}\cup \OPT_{\mathrm{ver}}$, with $p(\OPT'_{\text{skew}}) + p(\OPT''_{\text{skew}})\ge (1-\epsilon)p(\OPT_{\mathrm{hor}}\cup \OPT_{\mathrm{ver}})$ and $|\OPT''_{\text{skew}}|=O_{\eps,\epslarge,\epsthin}(1)$, such that there exists a container packing of $\OPT'_{\text{skew}}$ into a collection $\overline{\CC}$ of horizontal and vertical containers inside $K\setminus K_{\text{empty}}$, with $|\overline{\CC}|=O_{\eps,\epslarge,\epsthin}(1)$. Furthermore, the total area of the region inside $K\setminus K_{\text{empty}}$ not occupied by any container from $\overline{\CC}$ is at least $\max\{\epsilon N^2,a(\OPT_{\text{small}})-\epsilon^3N^2\}$.
\end{lemma}
\begin{proof}
    At the end of Lemma \ref{lem:repack-missing-items}, we have a packing of a subset of $\OPT_{\mathrm{hor}}\cup \OPT_{\mathrm{ver}}$
    into the $O_{\eps,\epslarge,\epsthin}(1)$ containers $\CC'_{\text{thick}}\cup \CC''$. We apply Lemma \ref{lem:weighted-shrinking} with $\delta=2\epsilon$ to each container in $\CC'_{\text{thick}}\cup \CC''$. This frees up a $2\eps$-fraction of the area occupied by each container, and thus the total free area inside $K\setminus K_{\text{empty}}$ is at least $2\eps\cdot(1-2\epsthin) N^2>\eps N^2$, assuming $\epsthin < 1/4$.

    We next apply Lemma \ref{lem:split-container} with $\delta=\eps^3/2$ to each of the shrunk containers, thus obtaining a collection of containers $\overline{\CC}$ with $|\overline{\CC}|=\frac{2}{\eps^3}\cdot|\CC'_{\text{thick}}\cup \CC''|=O_{\eps,\epslarge,\epsthin}(1)$. Let $\OPT'_{\text{skew}}$ be the items packed in the containers, and $\OPT''_{\text{skew}}$ be the $O_{\eps,\epslarge,\epsthin}(1)$ items of $\OPT_{\mathrm{hor}}\cup \OPT_{\mathrm{ver}}$ discarded while applying Lemmas \ref{lem:transformation-container-packing-new}, \ref{lem:weighted-shrinking} and \ref{lem:split-container}. Then $p(\OPT'_{\text{skew}}) + p(\OPT''_{\text{skew}})\ge (1-O(\epsilon))p(\OPT_{\mathrm{hor}}\cup \OPT_{\mathrm{ver}})$. Observe that Lemma \ref{lem:split-container} guarantees that the total area inside the containers $\overline{\CC}$ that is not occupied by any item is bounded by $\frac{\eps^3}{2}N^2$. Therefore, a($\overline{\CC}$) (i.e., the total area of all containers in $\overline{\CC}$) is at most $a(\OPT_{\mathrm{hor}}\cup \OPT_{\mathrm{ver}})+\frac{\eps^3}{2}N^2$. Since $a(\OPT_{\mathrm{hor}}\cup \OPT_{\mathrm{ver}})+a(\OPT_{\text{small}})\le N^2$, it follows that
    the total free area inside $K\setminus K_{\text{empty}}$ not occupied by any container from  $\overline{\CC}$ is at least $N^2-a(\overline{\CC})-a(K_{\text{empty}})\ge N^2-(a(\OPT_{\mathrm{hor}}\cup \OPT_{\mathrm{ver}})+\frac{\eps^3}{2}N^2) -2\epsthin N^2 \ge a(\OPT_{\text{small}})-(2\epsthin+\frac{\eps^3}{2})N^2\ge a(\OPT_{\text{small}})-\epsilon^3N^2$. The last inequality follows by assuming $\epsthin \le \eps^3/4$.
    This completes the proof. 
\end{proof}

We shall now pack the items of $\OPT_{\text{small}}$ in the free area in $K\setminus K_{\text{empty}}$ outside the containers $\overline{\CC}$. For this, we shall need that $\epssmall \cdot |\overline{\CC}| \le \eps^4/4$ holds. For technical reasons that will be discussed in the proof of Theorem \ref{theorem:skewed-packing} in \Cref{sec:skewed}, we shall also require that $\epssmall\cdot |\OPT''_{\text{skew}}|\le \epsthin$. By an appropriate choice of the function $f$ in Lemma \ref{lem:delete-intermediate-items} (discussed later), we may assume $\epssmall$ to be sufficiently small such that both the aforementioned conditions hold.

\begin{lemma}
\label{lem:pack-small-items-in-free-space}
    There is a subset $\OPT'_{\text{small}}\subseteq \OPT_{\text{small}}$ with $p(\OPT'_{\text{small}})\ge (1-O(\epsilon))p(\OPT_{\text{small}})$ that can be packed into a collection of at most $O_{\epssmall}(1)$ area containers each of dimension $\frac{\epssmall}{\eps}\times \frac{\epssmall}{\eps}$, lying in the region of $K\setminus K_{\text{empty}}$ not occupied by the containers $\overline{\CC}$.
\end{lemma}
\begin{proof}
Let $\epsilon':=\epssmall/\epsilon$. We draw a uniform grid inside $K\setminus K_{\text{empty}}$ where each cell of the grid has height and width equal to $\epsilon'N$. We delete all grid cells that overlap with any of the containers of $\overline{\CC}$. Notice that for each container, there are at most $4/\epsilon'$ grid cells that partially overlap with it. The total area of such partially overlapping grid cells is at most
\[\frac{4}{\epsilon'}\cdot |\overline{\CC}| \cdot \epsilon'^2 N^2 \le \epsilon^3 N^2,\]
where the inequality follows by our choice of $\epssmall$. Therefore by Lemma \ref{lem:final-skwed-packing}, the area of the free cells is at least $\max\{(\epsilon-\epsilon^3) N^2,a(\OPT_{\text{small}})-2\epsilon^3N^2\}$. Note that the free cells are by a factor $1/\epsilon$ larger in each dimension than each small item.

Now if $a(\OPT_{\text{small}})\le \epsilon^2 N^2$, the free cells have a total area of at least $(\epsilon-\epsilon^3)N^2>\frac{\eps}{2}N^2$, assuming $\eps<1/2$. Therefore, the items of $\OPT_{\text{small}}$ can be fully packed into the free cells using NFDH. Hence, assume that $a(\OPT_{\text{small}})> \epsilon^2 N^2$, so that the area of the free cells is at least $a(\OPT_{\text{small}})-2\epsilon^3N^2 \ge (1-2\epsilon)a(\OPT_{\text{small}})$. We partition $\OPT_{\text{small}}$ into groups of total area at least $4\epsilon\cdot a(\OPT_{\text{small}})$, i.e., we iteratively pick items into a group until their total area exceeds $4\epsilon\cdot a(\OPT_{\text{small}})$, and then restart the procedure to create another group (the last group may have a smaller total area). Since the area of each item is at most $\epssmall^2N^2 < \eps^3N^2 <\eps\cdot a(\OPT_{\text{small}})$, the number of groups is at least $\frac{1}{5\epsilon}$. We delete the group having minimum profit among the ones with total area at least $4\epsilon\cdot a(\OPT_{\text{small}})$, and let $\OPT'_{\text{small}}$ be the remaining items. Then $p(\OPT'_{\text{small}})\ge (1-5\epsilon)p(\OPT_{\text{small}})$ and $a(\OPT'_{\text{small}})\le (1-4\epsilon)a(\OPT_{\text{small}})$. We pack $\OPT'_{\text{small}}$ into the free cells using NFDH. Observe that we do not run out of cells in this process, since otherwise by Lemma \ref{lem:NFDH-guarantee}, the total area of the packed items would be at least $(1-2\epsilon)\cdot(1-2\eps)a(\OPT_{\text{small}})>a(\OPT'_{\text{small}})$.
\end{proof}

\paragraph{Choice of \texorpdfstring{$f$}{f} in Lemma \ref{lem:delete-intermediate-items}.}
Now we discuss the choice of $f$. 
Remember from Lemma \ref{lem:final-skwed-packing}, we showed that there exist subsets $\OPT'_{\text{skew}}, \OPT''_{\text{skew}} \subseteq \OPT_{\mathrm{hor}}\cup \OPT_{\mathrm{ver}}$, where $|\OPT''_{\text{skew}}|=O_{\eps,\epslarge,\epsthin}(1)$ and there exists a container packing of $\OPT'_{\text{skew}}$ into a collection $\overline{\CC}$ of horizontal and vertical containers inside the knapsack. As discussed before, we require there that $\epssmall$  is sufficiently small such that we can guarantee that $\epssmall \cdot |\overline{\CC}| \le \eps^4/4$ and $\epssmall\cdot |\OPT''_{\text{skew}}|\le \epsthin$. 
Since both  $|\overline{\CC}|$ and $|\OPT''_{\text{skew}}|$ are $O_{\eps,\epslarge,\epsthin}(1)$ and $\epsthin$ is a function of $\eps$ and $\epslarge$ only, \aw{all these quantities} do not depend on $\epssmall$.
Thus, there exists a function $k_{\eps}(\cdot)$ such that $k_{\eps}(\epslarge)\le \min\left\{\frac{\eps^4}{4|\overline{\CC}|},\frac{\epsthin^2}{|\OPT''_{\text{skew}}|}\right\}$  for any choice of $\epslarge$. We may then choose the function $f$ in Lemma \ref{lem:delete-intermediate-items} such that $f:=k_{\eps}^{-1}$, so that the desired bounds on $\epssmall$ are satisfied.

\subsection{Extension to the case of skewed items}
\label{sec:skewed}
We prove Theorem \ref{theorem:skewed-packing} in this subsection. Recall that there is a global constant $\lambda_{\eps}$ such that $\epssmall \ge \lambda_{\eps}$ (see Lemma \ref{lem:delete-intermediate-items}). We let $\eps_{\mathrm{skew}}=\lambda_{\eps}$ so that $\eps_{\mathrm{skew}}\le \epssmall$ holds. Since each item has one side of length at most $\eps_{\mathrm{skew}}N$, it follows that $\OPT_{\text{large}}=\emptyset$, i.e., all items in $\OPT$ are either vertical or horizontal or small. We apply a similar argument\aw{ation} as in the cardinality case. Using Lemma \ref{lem:transformation-container-packing-new}, we obtain a container packing of a subset $I'_P \subseteq \OPT_{\mathrm{hor}}\cup \OPT_{\mathrm{ver}}$ into $(\log \frac{1}{\epsthin})^{O_{\eps,\epslarge}(1)}$ containers, and a set $I'_T$ that can be packed into a strip of width $\epsthin N$. We then classify the containers as thick, thin and intermediate, such that for an appropriate choice of the parameters $\epsclarge, \epscsmall$, the total profit of items packed in intermediate containers is at most $\eps \cdot p(\OPT)$ (see Lemma \ref{lem:weighted-classification}). We temporarily discard the thin containers. Following this, we apply Lemma \ref{lem:weighted-shrinking} to each of the thick containers and by pushing the resulting containers down and to the left as much as possible, we free up an empty strip of width $\eps\cdot\epsclarge N \ge 6\epsthin N$ inside the knapsack. We pack the thin containers and the items of $I'_T$ into this strip using Lemma \ref{lem:repack-missing-items}, which occupy a total height of at most $4\epsthin N$. Observe that this still leaves an empty strip $\SM$ of width $2\epsthin N$. We then apply Lemmas \ref{lem:weighted-shrinking} and \ref{lem:split-container} to the containers, which creates sufficient free area (outside $\SM$) to pack the items of $\OPT_{\text{small}}$. Let $\OPT'_{\text{skew}} \subseteq \OPT_{\mathrm{hor}}\cup \OPT_{\mathrm{ver}}$ be the items packed inside the horizontal and vertical containers and $\OPT''_{\text{skew}}$ be the $O_{\eps,\epslarge,\epsthin}(1)$ items that were discarded so far, i.e., while applying Lemmas \ref{lem:transformation-container-packing-new}, \ref{lem:weighted-shrinking} and \ref{lem:split-container}. We pack the items of $\OPT_{\text{small}}$ in the created free space (outside $\SM$) using Lemma \ref{lem:pack-small-items-in-free-space}.

The only remaining piece to handle here is the set $\OPT''_{\text{skew}}$. In the cardinality case, we could afford to discard these items, since they are only constantly many in number, and thus had negligible profit assuming $|\OPT|$ was large enough. However, this no longer holds in the weighted case. Hence, we shall repack them inside the strip $\SM$. Recall that we \aw{chose} $f$ in Lemma \ref{lem:delete-intermediate-items} such that $\epssmall\cdot |\OPT''_{\text{skew}}| \le \epsthin$ holds. Assuming w.l.o.g.~that the strip $\SM$ is horizontal, we rotate the items of $\OPT''_{\text{skew}}$ so that their shorter side (which has length at most $\epssmall N$) is vertical, and pack them one above the other in a single horizontal container inside $\SM$.

\section{Improved approximation for the non-rotation case}
\label{sec:non-rotation-improved}
In this section, we prove Theorem \ref{theorem:non-rotation-improved-1}. We first begin with a discussion of the corridor partitioning framework of \cite{galvez2021approximating}. Note that, unlike the cardinality case, now we cannot afford to discard $O_{\eps}(1)$ items at negligible loss in profit. Also the large items may occupy a very large fraction of the area of the knapsack, and therefore a global classification of items into $\OPT_{\mathrm{hor}}, \OPT_{\mathrm{ver}}, \OPT_{\text{small}}$ as done in the cardinality case will not work.

To handle this, \cite{galvez2021approximating} presented a more robust version of the corridor decomposition lemma that allows to handle small items and ensure that a specified set of $O_{\eps}(1)$ items are not deleted.
Given a packing of a set of items $I$ inside $K$ and a set of  \emph{untouchable }items $I'\subseteq I$ with $|I'|\in O_{\eps}(1)$, a non-uniform grid $G(I')$ is defined inside $K$ where the $x$-coordinates (resp. $y$-coordinates) of the grid cells correspond to the $x$-coordinates (resp. $y$-coordinates) of the items of $I'$. Let $\CC(I')$ denote the collection of these grid cells and let $I(C)$ denote the set of items that intersect a cell $C\in \CC(I')$. Let $w(i \cap C)$ and $h(i \cap C)$ denote the width and height of the intersection of rectangle $i$ with cell $C$. For some constants $\epslarge, \epssmall$, the items $I(C)$ intersecting cell $C$ are classified as follows.
\begin{itemize}
    \item $I_{\text{small}}(C):= \{i\in I(C) \mid w(i\cap C)\le \epssmall\cdot w(C) \text{ and } h(i\cap C)\le \epssmall\cdot h(C)\}$.
    \item $I_{\text{large}}(C) := \{i\in I(C) \mid w(i\cap C)> \epslarge\cdot w(C) \text{ and } h(i \cap C) > \epslarge\cdot h(C)\}$.
    \item $I_{\mathrm{hor}}(C) := \{i \in I(C) \mid w(i \cap C)> \epslarge\cdot w(C) \text{ and } h(i \cap C)\le \epssmall\cdot h(C)\}$.
    \item $I_{\mathrm{ver}}(C) := \{i\in I(C) \mid w(i \cap C) \le \epssmall\cdot w(C) \text{ and } h(i \cap C)>\epslarge\cdot h(C) \}$.
    \item $I_{\text{int}}(C) := I(C) \setminus (I_{\text{small}}(C)\cup I_{\text{large}}(C)\cup I_{\mathrm{hor}}(C)\cup I_{\mathrm{ver}}(C))$.
\end{itemize}

By an appropriate choice of $\epslarge$ and $\epssmall$, it can be ensured that the items in $I_{\text{int}}(C)$ have negligible profit. Let $I_{\text{small}}(I')$ be the set of items $i$ that belong to $I_{\text{small}}(C)$ for every cell $C$ that intersects $i$. The following lemma \aw{is very similar to a corresponding statement in \cite{galvez2021approximating} and it can be proven with almost exactly the same argumentation}.

\begin{lemma} \label{lem:corr-decomposition}
 Let $I$ be a set of
items that can be packed inside $K$. Let also $I'\subseteq I$ be
a given set of untouchable items with $|I'|\in O_{\epsilon}(1)$. Then,
there exists a corridor partition of the knapsack and a set of items
$I_{\text{corr}}\subseteq I$ satisfying: \begin{itemize} 
\item there exists a set of items $I_{\text{corr}}^{\text{cross}}\subseteq I_{\text{corr}}$
such that 
\begin{itemize}
    \item each item in $I_{\text{corr}}\setminus I_{\text{corr}}^{\text{cross}}$ is completely contained in some corridor of the partition,
    \item $I'\subseteq I_{\text{corr}}\setminus I_{\text{corr}}^{\text{cross}}$, and
    \item $|I_{\text{corr}}^{\text{cross}}\setminus I_{\text{small}}(I')|\in O_{\epsilon,\epslarge}(1)$
\end{itemize}
\item for each cell $C\in\CC(I')$ we have that $a(I_{\text{corr}}^{\text{cross}}\cap I_{\text{small}}(C)  \cap I_{\text{small}}(I'))\le\eps^{3}\cdot a(C)$, 
\item $p(I_{\text{corr}})\ge(1-O(\epsilon))p(I)$, and
\item the number of corridors is $O_{\epsilon,\epsilon_{\text{large}}}(1)$
and each corridor has at most $1/\epsilon$ bends and width at most
$\epsilon_{\text{large}}N$, except possibly for the corridors containing
items from $I'$ that correspond to rectangular regions matching exactly
the size of these items. \end{itemize} 
\end{lemma} 
\begin{proof}
    We consider the packing of $(\cup_{C\in \CC(I')} I_{\text{hor}}(C)\cup I_{\text{ver}}(C))\setminus I'$, and imagine stretching the non-uniform grid $G(I')$ into a uniform $[0,1]\times [0,1]$ grid, so that each grid cell is of length $\frac{1}{1+2|I'|}$. Then each item in $(\cup_{C\in \CC(I')} I_{\text{hor}}(C)\cup I_{\text{ver}}(C))\setminus I'$ has height or width at least $\frac{\epslarge}{1+2|I'|}$. We apply Lemma \ref{lem:corr-decomposition-cardinality} to this packing yielding a decomposition of $[0,1]\times [0,1]$ into $O_{\eps,\epslarge}(1)$ corridors and the set $I_{\text{corr}}$ satisfying $p(I_{\text{corr}})\ge (1-O(\eps))p(I)$. We then rescale back to the original non-uniform grid. 

    We assign the items of $(\cup_{C\in \CC(I')} I_{\text{hor}}(C)\cup I_{\text{ver}}(C))\setminus I'$ that do not completely lie inside a corridor to the set $I_{\text{corr}}^{\text{cross}}$. We now include back the items of $I'$. Note that these items can overlap with some corridor, but they can be circumvented by adding only $O_{\eps,\epslarge}(1)$ extra lines (see Figure 5 in \cite{galvez2021approximating}). 

    Finally, we include back the items of $I_{\text{small}}(I')$ in the packing. Those items which do not completely lie inside a corridor are assigned to the set $I_{\text{corr}}^{\text{cross}}$. 
    For any cell $C \in \CC(I')$, since the total number of lines defining the corridors inside $C$ is only $O_{\eps,\epslarge}(1)$, the total area of the items in $I_{\text{corr}}^{\text{cross}} \cap I_{\text{small}}(C)\cap I_{\text{small}}(I')$ can be bounded by $\eps^3 \cdot a(C)$, for sufficiently small $\epssmall$, and the lemma follows.   
\end{proof}

\aw{We want to apply Lemma \ref{lem:corr-decomposition} to the optimal packing $\OPT$. We first do this with $I':=\emptyset$. In the following, we will describe a procedure that discards the constantly many items $I_{\text{corr}}^{\text{cross}}\setminus I_{\text{small}}(I')$ and an additional set of constantly many items \dk{coming from the corridor processing described later}. If the total profit of these discarded items is at most $\eps\cdot p(\OPT)$ we are fine. Otherwise, we add the constantly many discarded items to the set $I'$ and iterate. After at most $1/\eps$ iterations, we must have that the discarded items of that iteration have a total profit of at most $\eps\cdot p(\OPT)$. Therefore, we assume in the following that we have 
}
a partition of the knapsack into $O_{\epsilon,\epslarge}(1)$ corridors and a set $\OPT_{\text{corr}}$ consisting of items packed inside the corridors. 
For each cell $C\in \CC(I')$, we classify the items $\OPT(C)$ intersecting $C$ into the sets $\OPT_{\text{small}}(C), \OPT_{\text{large}}(C), \OPT_{\mathrm{hor}}(C), \OPT_{\mathrm{ver}}(C), \OPT_{\text{int}}(C)$ as discussed before and let $\OPT_{\text{small}}$ be the set of items $i$ that belong to $\OPT_{\text{small}}(C)$ for every cell $C$ intersecting $i$.
By an appropriate choice of $\epslarge$ and $\epssmall$, it can be ensured that $p(\cup_{C\in \CC(I')}\OPT_{\text{int}}(C))\le \eps\cdot p(\OPT)$, and so the items of $\cup_{C\in \CC(I')}\OPT_{\text{int}}(C)$ are discarded.

\begin{figure}
  \centering
  \begin{subfigure}[b]{0.45\textwidth}
  \centering
    \hspace*{-3.4cm}
    \includegraphics[width=2\linewidth]{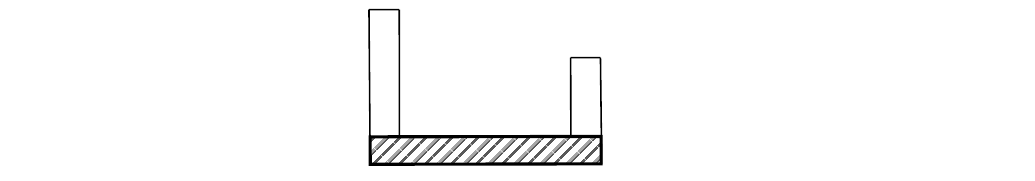}
    \caption{}
    \label{fig:sub1}
  \end{subfigure}
  \hfill
  \begin{subfigure}[b]{0.45\textwidth}
  \centering
    \hspace*{-4.4cm}
    \includegraphics[width=2.2\linewidth]{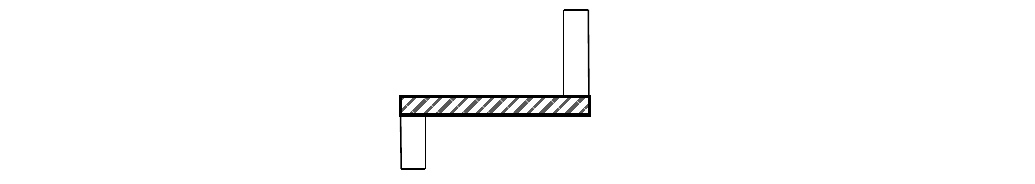}
    \caption{}
    \label{fig:sub2}
  \end{subfigure}

  \caption{(a) An acute subcorridor. (b) An obtuse subcorridor.}
  \label{fig:bends}
\end{figure}

\paragraph{Processing a subcorridor.}
The next step is to partition the corridors into boxes, retaining a large fraction of the profit of $\OPT_{\text{corr}}$. For simplicity, we discard the items of $\OPT_{\text{small}}$ for now, and explain later how to include them back. In our processing, we shall discard a set of $O_{\eps,\epslarge}(1)$ items $\OPT_{\text{kill}}$. By our discussion in the previous paragraph, we may assume that we are in the iteration where these items have a profit of at most $\eps\cdot p(\OPT)$. We first assign the items of $\OPT_{\text{corr}}^{\text{cross}}\setminus \OPT_{\text{small}}$ to $\OPT_{\text{kill}}$ (note that they are only $O_{\eps,\epslarge}(1)$ in number).

A subcorridor is said to be \textit{acute} if either it is the first or the last subcorridor of an open corridor, or the two adjacent subcorridors lie on the same side of the considered subcorridor; otherwise it is labeled \textit{obtuse} (see Figure \ref{fig:bends}). Consider an acute subcorridor $C_1$ of a corridor $C$ and assume that $C_1$ is horizontal with the shorter horizontal edge being the top one (the other cases are analogous). Similar to the proof of Lemma \ref{lem:transformation-container-packing-new}, we can consider strips of geometrically increasing heights inside $C_1$, mark as thin the items of the bottommost strip and remove them and continue. For the cardinality case, this was crucial to ensure the number of boxes is small enough, which was in turn essential for creating a sufficiently large empty strip for repacking thin items and small boxes. However, for the weighted case, we can \dk{work} with a simpler processing scheme similar to that used in \cite{galvez2021approximating} (which would not have worked out for the cardinality case though).

\aw{Formally,}
let $c_{\eps,\epslarge}$ be an upper bound on the total number of subcorridors, and let $\eps_{\text{box}}:= \eps^4/c_{\eps,\epslarge}$. Let $h$ be the height of $C_1$. We draw $1/\epsilon_{\text{box}}$ equidistant horizontal lines that partition the private region of $C_1$ into strips of height $\epsilon_{\text{box}}h$ and assign the horizontal items intersected by these lines to the set $\OPT_{\text{kill}}$ (there are only $O_{\epsilon, \epslarge}(1)$ such items). We mark as \textit{thin} the items \aw{inside} the bottommost, i.e., the widest such strip, and the remaining items as \textit{fat}. As before, if the thin items of a particular piece are deleted, the fat items of the piece can be repacked into boxes. Following this, we construct paths from the endpoints of the top edges of the boxes, that partition $C\setminus C_1$ into $1/\epsilon_{\text{box}}$ smaller corridors, as we did in the cardinality case. We call the above procedure as \textit{processing} the corridor piece $C_1$.

We call a subcorridor to be \textit{long} if its length exceeds $N/2$ and \textit{short} otherwise. \aw{In \cite{galvez2021approximating}, three different ways were presented} to partition the corridors into boxes, where in each case some subset of items were discarded and the others were packed into boxes.
\aw{They differ in the order in which the corridor pieces are processed with the method described above.}
Let $\OPT_T$ be the set of thin items that are discarded in at least one of the corridor processing methods, and let $\OPT_F$ be the items that are marked as fat in all cases. Further, let $\OPT_{LT}$ (resp. $\OPT_{ST}$) be the items in $\OPT_T$ that came from long (resp. short) subcorridors. Analogously, items in $\OPT_F$ are classified into $\OPT_{LF}$ and $\OPT_{SF}$. Note that in the weighted setting, the large items and $O_{\epsilon}(1)$ untouchable items were also included in the set $\OPT_{LF}$ in \cite{galvez2021approximating}.

In one of the three processing strategies, the short subcorridors were labeled as even and odd alternately, the odd (or even) short subcorridors were deleted, and then the even (or odd) short subcorridors were processed last. This gives the first bound in the lemma stated below. It was also shown that a subset of the thin items could be packed into an L-region at the knapsack boundary and a subset of the remaining items could be packed into containers in the free area outside the L~(referred to as an $L\& C$-packing), which gives the second profit bound in the following lemma. \aw{Here, $\OPT_{L\&C}$ denotes the maximum profit of a packing that
uses $O_{\epsilon}(1)$ containers and one special L-shaped corridor with two corridor pieces; the left edge of this corridor coincides with the left edge of the knapsack and its bottom edge coincides with the bottom edge of the knapsack.}

\begin{lemma}[\cite{galvez2021approximating}]
\label{lem:lcpacking-focs}
    The following statements hold:
    \begin{enumerate}[label=(\roman*)]
        \item $p(\OPT_{L\&C}) \ge (1-\epsilon)(p(\OPT_{LF})+\frac{1}{2}(p(\OPT_{SF})+p(\OPT_{ST})))$ and
        \item $p(\OPT_{L\&C}) \ge (1-\epsilon)(\frac{3}{4}p(\OPT_{LT})+p(\OPT_{ST})+\frac{1}{2}p(\OPT_{SF}))$.
    \end{enumerate}
\end{lemma}

\aw{In \cite{galvez2021approximating} two additional partitioning techniques were presented. The first one
processes} the subcorridors in any feasible order, which gives a packing of profit $p(\OPT_{LF})+p(\OPT_{SF})$, i.e., the full profit of the fat items. This loses out on the profit of the thin items. The \aw{second} processing strategy \aw{is} more involved and returned a packing of profit at least $p(\OPT_{LF})+\frac{1}{2}(p(\OPT_{SF})+p(\OPT_{LT}))$. We present an alternate corridor processing strategy that simplifies and improves upon both of \aw{these two} profit guarantees.

\paragraph{Alternate corridor processing.}
To begin with, we process all acute subcorridors that are short. Then, all remaining acute subcorridors will be long. We make the following simple observation.

\begin{lemma}
    \label{obs:obtuse-piece-length-new}
        In each of the remaining corridors, all obtuse subcorridors must be short.
\end{lemma}
    \begin{proof}
        Suppose there exists an obtuse subcorridor $C_1$ that is long, and w.l.o.g.~assume that $C_1$ is horizontal. Starting from the left (resp. right) bend of $C_1$, we traverse the subcorridors one by one until we reach the first acute subcorridor $C'$ (resp. $C''$). Since all the short acute subcorridors were already processed, it must be the case that both $C'$ and $C''$ are long. Since $C_1$ is a horizontal long subcorridor, neither of $C'$ and $C''$ can be horizontal, else the width of the knapsack would exceed $N$. But then both $C'$ and $C''$ are vertical long subcorridors having only obtuse subcorridors in between, implying that the height of the knapsack exceeds $N$, a contradiction.
    \end{proof}

    From the above observation, it follows that in the remaining corridors, all the long subcorridors are acute, while all the short subcorridors are obtuse. Notice also that between any two vertical long subcorridors, there must exist a horizontal long subcorridor, otherwise the height of the knapsack would exceed~$N$. We perform two different corridor processings as follows. In the first case, we first process all the horizontal long subcorridors in any order, followed by processing the obtuse subcorridors (which are all short). At the end, only the vertical long subcorridors survive, which we can now directly partition into boxes. In this way, we save all the thin items in the vertical long subcorridors. The other case is symmetric with the roles of the horizontal and vertical long subcorridors reversed. We therefore obtain the following lemma.

    \begin{lemma}
    \label{lem:new-processing}
        We have $p(\OPT_{L\&C}) \ge (1-\epsilon)(p(\OPT_{LF})+p(\OPT_{SF})+\frac{1}{2}p(\OPT_{LT}))$.
    \end{lemma}

    \aw{We remark that Lemma \ref{lem:lcpacking-focs} still holds if we define our sets $\OPT_{F}, \OPT_{T}$, etc. now according to the corridor decompositions from \cite{galvez2021approximating} and, additionally, our new alternative corridor processing. For the rest of this paper, all these sets are defined via all these corridor processing steps.} Also, note that in all the corridor processing steps, we need to discard a set of $O_{\eps,\epslarge}(1)$ items. We assign these items to the set $\OPT_{\text{kill}}$.
    Combining Lemmas \ref{lem:lcpacking-focs} and \ref{lem:new-processing}, we prove Theorem \ref{theorem:non-rotation-improved-1}.

    \nonrotationimproved*

    \begin{proof}
        We have
        \begin{alignat*}{2}
            2p(\OPT_{L\&C}) &\ge (1-\epsilon)(2p(\OPT_{LF})+p(\OPT_{SF})+p(\OPT_{ST})) &\quad& \text{[Lemma \ref{lem:lcpacking-focs}(i)]} \\
            6p(\OPT_{L\& C}) &\ge (1-\epsilon)\left(\frac{9}{2}p(\OPT_{LT})+6p(\OPT_{ST})+3p(\OPT_{SF})\right) &\quad& \text{[Lemma \ref{lem:lcpacking-focs}(ii)]} \\
            5p(\OPT_{L\& C}) &\ge (1-\epsilon)\left(5p(\OPT_{LF})+5p(\OPT_{SF})+\frac{5}{2}p(\OPT_{LT})\right) &\quad& \text{[Lemma \ref{lem:new-processing}]}
        \end{alignat*}
        Adding the above three inequalities, we obtain
        \begin{align*}
            13p(\OPT_{L\& C}) &\ge (1-\epsilon)(7p(\OPT_{LF})+ 9p(\OPT_{SF})+7p(\OPT_{LT})+7p(\OPT_{ST})) \\
            &\ge (7-7\epsilon)p(\OPT),
        \end{align*}
        thus completing the proof.
    \end{proof}

\paragraph{Handling small items.}
We now \aw{describe} 
how we pack the small items that we ignored so far. \aw{We handle this step in the same way as in \cite{galvez2021approximating} (see Section 6.3.1 in \cite{galvez2021approximating} for details)}. Let $\OPT_{\text{small}}^{\text{cross}}:= \OPT_{\text{small}}\cap \OPT_{\text{corr}}^{\text{cross}}$ be the items not completely lying inside the corridors.
We apply our corridor processing techniques to the packing of $\OPT_{\text{corr}}$ (i.e., by including the items of $\OPT_{\text{small}}\setminus \OPT_{\text{small}}^{\text{cross}}$) and let $\OPT_{\text{small}}^{\text{kill}}$ denote the items of $\OPT_{\text{small}}$ that are intersected by some line during the corridor processing. Let $\OPT'_{\text{small}}:= \OPT_{\text{small}}^{\text{cross}} \cup \OPT_{\text{small}}^{\text{kill}}$ be the items that we wish to repack. We assign each item $i \in \OPT'_{\text{small}}$ to the cell $C$ that maximizes the area of the intersection of $i$ with $C$. By an argument similar to the proof of Lemma \ref{lem:corr-decomposition}, the total area of the items assigned to $C$ can be bounded by $O(\eps^3)\cdot a(C)$, \aw{i.e., by choosing the function $f$ accordingly}.
Note that in the construction of boxes from the corridors in Lemmas \ref{lem:lcpacking-focs}(i) and \ref{lem:new-processing}, the items are moved only within a subcorridor. Assume for some cell $C\in \CC(I')$, there is a horizontal subcorridor $H$ intersecting \dk{the top or bottom boundary of} $C$ such that some items within $H$ were moved (vertically) to $C$ that were not in $C$ before. Since the part of the subcorridor $H$ lying within $C$ has a height of at most $\epslarge\cdot h(C)$, the total area of $C$ \dk{occupied by items from such horizontal subcorridors}
is at most $\epslarge\cdot a(C)$. Analogously, an area of \aw{at most} $\epslarge\cdot a(C)$ is \dk{occupied by items that were moved horizontally from other cells into $C$ through vertical subcorridors intersecting the left and right boundaries of $C$}. We now shrink each horizontal (resp. vertical) box arising after processing the corridors by a factor of $1-\eps$, by losing only an $\eps$-fraction of the profit of the box. The $O_{\eps,\epslarge}(1)$ items discarded during this process are assigned to the set $\OPT_{\text{kill}}$. Overall this creates a free area of at least $\Omega(\eps)\cdot a(C)$ inside the cell $C$. Since the area of the small items assigned to $C$ is only $O(\eps^3)\cdot a(C)$, they can be packed using NFDH into $O_{\epssmall}(1)$ area containers inside $C$.

Note that summing over all cells, the total area of the items of $\OPT'_{\text{small}}$ is bounded by $O(\eps^3)\cdot(N^2-a(I'))$. In the packings that we \aw{will} present in \Cref{sec:weighted-case}, we \aw{will use a different way to} repack the items of $\OPT'_{\text{small}}$. \aw{More precisely, we will} identify a certain region of large area inside the knapsack that is sufficient to pack $\OPT'_{\text{small}}$.

\section{Hardness of Container Packings for weighted case of 2DKR}
\label{sec:hardness-of-container-packing}
In Section \ref{sec:ptas-cardinality}, we showed that in the cardinality case of our problem,
there always exists a $(1+\eps)$-approximate container packing with
$O_{\eps}(1)$ containers. A natural question is whether this result
can be generalized to the weighted case. In \cite{galvez2021approximating}, it was shown that in the weighted case, there is always a $(1.5+\eps)$-approximate
solution using $O_{\eps}(1)$ containers. We show that this is essentially
the best possible.

\begin{theorem}
\label{lem:1.5-hardness}
For any $\delta > 0$, there is an instance of
the weighted case of the two-dimensional knapsack problem with rotations
such that any container packing providing a $(1.5-\delta)$-approximation
needs $\Omega(\delta\log N)$ containers.
\end{theorem}

\begin{figure}
  \centering

  \begin{subfigure}[b]{0.45\textwidth}
  \centering
    \includegraphics[width=0.55\linewidth]{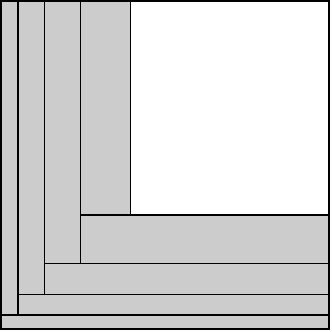}
    \vspace*{0.1cm}
    \caption{}
    \label{fig:non-rot-l-packing}
  \end{subfigure}
  \hfill
  \begin{subfigure}[b]{0.45\textwidth}
  \hspace*{-1cm}
  \centering
      \includegraphics[width=0.7\linewidth]{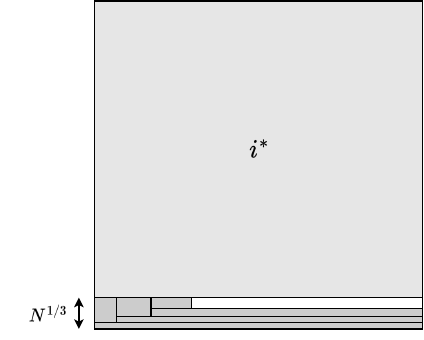}
    \caption{}
    \label{fig:container-hardness}
  \end{subfigure}

  \caption{(a) Hard instance for setting without rotations. (b) Hard instance for setting with rotations.}
  \label{fig:hard-instances}
\end{figure}

We prove \Cref{lem:1.5-hardness} in the remainder of this subsection.
We define an instance of 2DKR for which we will show that a container packing with only a few containers can obtain only (essentially) a $2/3$-fraction of the optimal profit.
Let the number of items $n$ be an odd integer and let the side length of the knapsack be $N:= 2^{\frac{3(n+1)}{2}}$. We consider the instance $I$ defined as follows. For each $j\in [\frac{n-1}{2}]$, we have an item $H_j$ of size $(N-(2^{j-1}-1)N^{2/3}) \times 2^{j-1}$, and an item $V_j$ of size $(2^{j-1}N^{2/3})\times \dk{(N^{1/3}-2^{j}+1)}$, each having a profit of $1$. Finally, we have a ``special" item $i^*$ of size $N \times (N-N^{1/3})$, and a profit of $\frac{n-1}{2}$ (see Figure \ref{fig:container-hardness}).

This is inspired by a family of instances for 2DK presented in \cite{galvez2021approximating} for which it was (essentially) shown that no container packing with few containers achieves an approximation ratio better than 2.
In the instance in \cite{galvez2021approximating}, the items are defined such that they fit tightly into a container with
the shape of an $\Ls$ (see \Cref{fig:non-rot-l-packing}). However, it is easy to pack them into one container by simply rotating the vertical items by 90 degrees. Therefore, in our construction we introduced the special item $i^*$ and scale down the heights of all items in the construction in \cite{galvez2021approximating} (see \Cref{fig:container-hardness}). Then, ignoring the item $i^*$, the items still fit inside a container with
the shape of an $\Ls$, however, its vertical arm is much smaller, i.e., $N^{1/3}$. Thus, the items in the horizontal arm do not fit in the vertical arm after rotation. Furthermore, we made the items for the vertical arm slightly wider, i.e., at least $N^{2/3}$, so that for each of them the width is larger than the height of the horizontal arm. Thus, they do not fit in the horizontal arm if we rotate them. Therefore, our obtained items behave like items in an instance of 2DK, since effectively we cannot rotate them.

\begin{lemma}
\label{lem:feasible-packing}
    There exists a feasible packing of all input items inside the knapsack.
\end{lemma}
\begin{proof}
    First, we orient the rectangle $i^*$ such that $w(i^*)=N$ and $h(i^*)=N-N^{1/3}$ and place it such that the top edge of $i^*$ coincides with the top boundary of the knapsack. In the remaining empty strip $\SM$ of height $N^{1/3}$ below $i^*$, we pack the rectangles of $\bigcup_{j\in [\frac{n-1}{2}]} H_j \cup V_j$ as follows. \aw{For each  $j\in [\frac{n-1}{2}]$} we orient the rectangles \aw{$H_j$ and $V_j$} such that $w(H_j)=N-(2^{j-1}-1)N^{2/3}$ and $h(H_j)=2^{j-1}$, and $w(V_j)=2^{j-1}N^{2/3}$ and \dk{$h(V_j)=N^{1/3}-2^{j}+1$}. We place $H_1$ \aw{such that its bottom edge touches the bottom boundary of the knapsack} and then place $V_1$ on top of $H_1$ as much to the left as possible -- note that $V_1$ \aw{fits inside} the strip $\SM$ since $h(V_1)+h(H_1)=N^{1/3}$. \dk{Next, note that $w(H_2)+w(V_1)=(N-N^{2/3})+N^{2/3}=N$, so we can place $H_2$ such that its bottom edge touches the top edge of $H_1$. We then place $V_2$ above $H_2$ as much to the left as possible (touching the right edge of $V_1$), noting that this is possible since $h(V_2)+h(H_2)+h(H_1)=N^{1/3}=h(\SM)$. In general, we have that \ak{$w(H_j)+\sum_{i=1}^{j-1}w(V_i)=(N-(2^{j-1}-1)N^{2/3})+\sum_{i=1}^{j-1}2^{i-1}N^{2/3}=(N-(2^{j-1}-1)N^{2/3})+(2^{j-1}-1)N^{2/3}=N$} and $h(V_j)+\sum_{i=1}^{j} h(H_i)= (N^{1/3}-2^j+1)+\sum_{i=1}^{j} 2^{i-1} = N^{1/3}$
    for all $j\ge 2$, hence after $H_1,\ldots,H_{j-1}$ and $V_1,\ldots, V_{j-1}$ have been packed, we can pack $H_j$ with its bottom edge touching $H_{j-1}$, and pack $V_j$ above $H_j$ touching the right edge of $V_{j-1}$. In this way, we obtain a packing of $\bigcup_{j\in [\frac{n-1}{2}]} H_j \cup V_j$ inside the strip $\SM$ as shown in Figure \ref{fig:container-hardness}.}
\end{proof}

We are now ready to prove Lemma \ref{lem:1.5-hardness}.

\begin{proof}[Proof of \Cref{lem:1.5-hardness}]
    Due to Lemma \ref{lem:feasible-packing}, we have $p(\OPT)=3\cdot \frac{n-1}{2}$. Therefore, any $(3/2-\delta)$-approximation must pack the rectangle $i^*$. W.l.o.g., we assume that $i^*$ is oriented such that $w(i^*)=\aw{N}$ and $h(i^*)=N-N^{1/3}$, and the top edge of $i^*$ coincides with the top knapsack boundary. This means that in the strip $\SM$ of height $N^{1/3}$ and width $N$ lying below $i^*$, at least $(1/2+\delta)(n-1)$ rectangles of $\bigcup_{j\in [\frac{n-1}{2}]} H_j \cup V_j$ must be packed into containers.

    Observe that for each $j\in [\frac{n-1}{2}]$, if $H_j$ is packed in the strip $\SM$, then it must be oriented such that $w(H_j)=N-(2^{j-1}-1)N^{2/3}$ and $h(H_j)=2^{j-1}$ \dk{(this is because the longer side of $H_j$ has length \ak{$N-(2^{j-1}-1)N^{2/3}\ge N-(2^{\frac{(n+1)}{2}-2})N^{2/3} \ge N-\frac{1}{4}N^{1/3}\cdot N^{2/3}=\frac{3}{4}N>h(\SM)$)}}. Also, since the \dk{longer} side of $V_j$ has length at least $2^{j-1}N^{2/3}\ge N^{2/3}>h(\SM)$, it follows that if $V_j$ is packed, it must be oriented such that $w(V_j)=2^{j-1}N^{2/3}$ and $h(V_j)=N^{1/3}-2^{j}+1$. Suppose there exists a container packing of a subset of items from $\bigcup_{j\in [\frac{n-1}{2}]} H_j \cup V_j$ into $c$ containers inside $\SM$, for some $c\in \mathbb{N}$. Note that there cannot be any area containers since $w(H_j)>N/2$ and $h(V_j)>h(\SM)/2$, for all $j$, \aw{using that $\delta \le 1/2$ since otherwise the claim is trivial}. For each $j$, we call the rectangles $H_j$ and $V_j$ a \textit{symmetric pair}. Observe that if the packing contains $s$ symmetric pairs, then the number of items packed \aw{apart from $i^*$} is at most \dk{$2s+(\frac{n-1}{2}-s)=s+\frac{n-1}{2}$}. We shall show that, if there are $c$ containers in the packing, then $s \le c$. \dk{If this holds, then since for a $(3/2-\delta)$-approximation, at least $(1/2+\delta)(n-1)$ items from $\bigcup_{j\in [\frac{n-1}{2}]} H_j \cup V_j$ must be packed.
    \aw{However, before we argued that at most $s+\frac{n-1}{2}$ of them are packed, which yields}  that $s+\frac{n-1}{2}\ge (\frac{1}{2}+\delta)(n-1)$, implying $c\ge s=\Omega(\delta n) = \Omega(\delta \log N)$.}

    We now establish that $s\le c$. Consider a $j\in [\frac{n-1}{2}]$ such that both $H_j$ and $V_j$ are packed. Let $C_j$ be the container that packs $H_j$. We first show that $H_j$ must be the rectangle of smallest width \aw{among the rectangles $\{H_{j'}\}_{j'\in [\frac{n-1}{2}]}$} in the container $C_j$. For this, we consider the following two cases.
    \begin{itemize}
        \item {\em $H_j$ and $V_j$ are both packed in $C_j$}: In this case, $C_j$ cannot be a vertical container, since otherwise its width would be at least 
        $w(H_j)+w(V_j)=N-(2^{j-1}-1)N^{2/3}+2^{j-1}N^{2/3}>N$, a contradiction. If $C_j$ also contains $H_{j'}$ for some $j'>j$, then the height of $C_j$ is at least $h(V_j)+h(H_{j'})=(N^{1/3}-2^{j}+1)+2^{j'-1}>N^{1/3}=h(\SM)$, a contradiction. Hence $H_j$ is the rectangle of smallest width inside $C_j$.
        \item {\em $V_j$ is packed in a distinct container $C'_j$}: \dk{Then we have $w(C_j)+w(C'_j)\ge w(H_j)+w(V_j)= N-(2^{j-1}-1)N^{2/3}+2^{j-1}N^{2/3}>N$, and thus there must exist a vertical line segment $\ell$ inside $\SM$ intersecting both $C_j$ and $C'_j$. Suppose $C_j$ also contains $H_{j'}$ for some $j'>j$. Since $w(H_j)+w(H_{j'})>N$, $C_j$ cannot be a vertical container. Therefore $C_j$ is a horizontal container and the height of $C_j$ is then at least $h(H_j)+h(H_{j'})>2^{j'-1}$. But then the length of the segment $\ell$ would exceed
        $2^{j'-1}+\aw{h(V_j)}=2^{j'-1}+(N^{1/3}-2^{j}+1)>N^{1/3}=h(\SM)$, a contradiction.}
    \end{itemize}
    \dk{Consider \aw{a map} $f$ that maps each packed symmetric pair $(H_j,V_j)$ to the container $C_j$ where the item $H_j$ is packed. Since we have established that $H_j$ is the rectangle of smallest width inside $C_j$, the mapping $f$ must be one-to-one. This implies that the number of symmetric pairs in the packing must be upper bounded by the number of containers, i.e., it holds that $s\le c$.}
\end{proof}

\section{Improved approximation for weighted case of 2DKR}
\label{sec:weighted-case}
We prove \Cref{theorem:rotation-weighted} in this section. Notice that the instances from the proof of \Cref{lem:1.5-hardness} have
one very large item that is as wide as the knapsack and almost (but
not exactly) as high as the knapsack. In fact, we can show that this
case is (essentially) the only bottleneck for getting a packing with
$O_{\eps}(1)$ containers and an approximation ratio strictly better
than $1.5$. \aw{This is our first step to prove \Cref{theorem:rotation-weighted}.}

Let $\eps>0$. We say that an item $i^*$ is \emph{huge}
if $h(i^*)\ge(1/2+\eps)N$ and $w(i^*)\ge(1/2+\eps)N$. Note that there can be at most one huge item in any feasible packing. It turns out that if $\OPT$ does not contain a huge item, the most profitable container packing already achieves an approximation ratio less than $1.5$.

\begin{restatable}{lemma}{nohugeitem}
\label{lem:no-huge-item}
    There exists a container packing with $O_{\epsilon}(1)$ containers and an approximation ratio of $\frac{190}{127}+\epsilon < 1.497+\eps$ if there is no huge item in the optimal solution.
\end{restatable}

\begin{figure}
    \centering
    \includegraphics[width=0.35\linewidth]{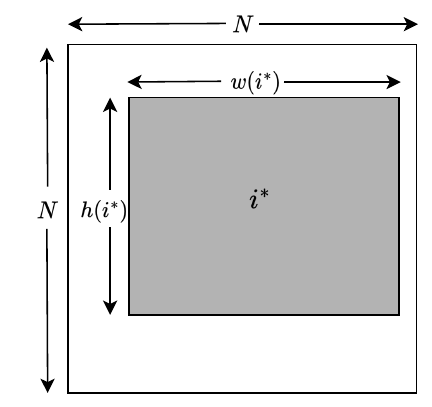}
    \caption{A huge item $i^*$ satisfying $w(i^*)\ge h(i^*)$ and $N-w(i^*)\ge 2\eps^2(N-h(i^*))$ if, e.g., $\eps=1/4$.}
    \label{fig:chain-and-dimensions-combined}
\end{figure}

We shall prove \Cref{lem:no-huge-item} in \Cref{sec:no-huge-item}. Suppose now that there is a huge item $i^*$. 
In the following lemma, the condition $w(i^*)\ge h(i^*)$ \aw{
(which is w.l.o.g.~since we can rotate items)}
in particular implies
that the total thickness of the area on the left and the right of~$i^*$ (i.e., $N-w(i^*)$) is thinner than the corresponding
area on the top and below~$i^*$ (i.e., $N-h(i^*)$); however, \aw{we assume that it is not much thinner, i.e.,}
$N-w(i^*)>2\eps^{2}(N-h(i^*))$, see Figure \ref{fig:chain-and-dimensions-combined}.

\begin{restatable}{lemma}{smallhugeitem}
\label{lem:container-packing=000020exists}
There exists a container
packing with $O_{\eps}(1)$ containers and an approximation ratio
of $\frac{37}{25}+\epsilon = 1.48+\eps$ if there is a huge item $i^*$ in the optimal solution with $w(i^*)\ge h(i^*)$
and $N-w(i^*)>2\eps^{2}(N-h(i^*))$.
\end{restatable}

\Cref{lem:container-packing=000020exists} will be proved in \Cref{sec:vertical-arms-large}.
Suppose now that there is
a huge item $i^*\in\OPT$ \aw{such that w.l.o.g. $w(i^*)\ge h(i^*)$ holds, but also} $N-w(i^*)\le2\eps^{2}(N-h(i^*))$. Our goal
is to argue that there is a {$\frac{130}{87}+\eps < 1.495$-approximate packing
using $O_{\eps}(1)$ containers and an $\Ls$-corridor (see Figure \ref{fig:structured-packings}),
similar as in \cite{galvez2021approximating}. 
We shall formally define an \lc-packing in \Cref{sec:lc*-packings}. Intuitively, an \lc-packing consists of $O_{\eps}(1)$ containers and a single $\Ls$-corridor whose horizontal (resp.~vertical) arm is aligned with the bottom (resp.~left) boundary of $K$, and the vertical arm has a height of at most $1/2$. As we shall see, the latter condition enables us to efficiently compute a partition of items that can be packed into the two arms of the $\Ls$. Using this, we shall give a PTAS to compute the most profitable \lc-packing. 
Finally, we show the existence of an \lc-packing of profit strictly more than $\frac{2}{3}p(\OPT)$ for the case when the total thickness of the regions on the left and right of $i^*$ is negligible compared to the regions on the top and below $i^*$. See \Cref{sec:large-huge-item} for a proof of the following lemma.

\begin{restatable}{lemma}{LCpackinghugeitem}
\label{lem:LC-packing-hugeitem}
If in the optimal solution
there is a huge item $i^{*}\in\OPT$ with $w(i^*)\ge h(i^*)$ and $N-w(i^{*})\le2\eps^{2}(N-h(i^{*}))$
then there is an L\&C{*}-packing with $O_{\epsilon}(1)$ containers
that yields an approximation ratio of $\frac{130}{87}+\eps < 1.495+\eps$.
\end{restatable}

Combining Lemmas \ref{lem:no-huge-item}, \ref{lem:container-packing=000020exists} and \ref{lem:LC-packing-hugeitem}, in the worst case, we obtain an approximation ratio of $\frac{190}{127}+\epsilon<1.497+\epsilon$ for the weighted case of 2DKR, which proves \Cref{theorem:rotation-weighted}.

\subsection{There is no huge item in the optimal packing}
\label{sec:no-huge-item}
We apply the corridor decomposition framework (Lemma \ref{lem:corr-decomposition}) to the optimal packing $\OPT$ and classify items into $\OPT_{LF},\OPT_{SF}, \OPT_{LT},\OPT_{ST}$ as in \Cref{sec:non-rotation-improved}. \dk{Note that the packings corresponding to Lemmas \ref{lem:lcpacking-focs}(i) and \ref{lem:new-processing} are all container packings, and hence their profit guarantees continue to hold even in the setting with rotations.} 

In the following, we present several ways of packing various sets of items into containers. In each case, the main idea is to free a strip of small width or height in the knapsack, where we can repack a subset of the thin items, i.e., the items in $\OPT_{LT}\cup \OPT_{ST}$. For this, we shall frequently make use of the following lemma.

\begin{lemma}
\label{lem:bound-thin-items}
    The items in $\OPT_{LT}\cup \OPT_{ST}$ can be packed inside a \aw{rectangular} box of \aw{size} $\eps^4 N \times N$. Also, items in $\OPT_{ST}$ can be packed in a \aw{rectangular} box of \aw{size} $\eps^4 N \times (1/2+\eps)N$.    
\end{lemma}
\begin{proof}
    Since the number of subcorridors is bounded by $c_{\epsilon,\epsilon_{\text{large}}}$ and each item of $\OPT_{LT}\cup \OPT_{ST}$ lies inside a strip of width at most $\epsilon_{\text{box}}N$ inside a subcorridor, the sum of the widths of these strips is bounded by $\epsilon_{\text{box}}N\cdot c_{\epsilon,\epsilon_{\text{large}}} \le \epsilon^4N$, by our choice of $\eps_{\text{box}}$. Therefore the items of $\OPT_{LT}\cup \OPT_{ST}$ can be packed inside a box of dimensions $\eps^4 N \times N$. Also, since the longer side of each item in $\OPT_{ST}$ has length at most $(1/2+2\epslarge)N\le (1/2+\eps)N$, they can be packed inside a box of width $(1/2+\eps)N$.
\end{proof}

We now state the {\em resource contraction lemma} for the weighted case from \cite{galvez2021approximating}.

\begin{lemma}[\cite{galvez2021approximating}]
\label{lem:resource-contraction}
        Let $I$ be a set of items packed into the knapsack such that no item $i \in I$ \aw{satisfies} both $w(i) \ge (1-\epsilon)N$ and $h(i) \ge (1-\epsilon)N$. Then it is possible to pack a set $I' \subseteq I$ with profit at least $\frac{1}{2}p(I)$ into a box of size $N\times (1-\epsilon/2)N$ if rotations are allowed.
\end{lemma}

We show the following lemma.

\begin{lemma}
\label{lem:res-contraction}
    There exists a container packing of profit at least $(1-\epsilon)(p(\OPT_{LT})+p(\OPT_{ST})+\frac{1}{2}(p(\OPT_{LF})+p(\OPT_{SF})))$.
\end{lemma}
\begin{proof}
    We first {\em process} all subcorridors (with parameter $\eps_{\mathrm{box}}$ as discussed in \Cref{sec:non-rotation-improved}) in any order and obtain a container packing of profit at least $(1-\epsilon)(p(\OPT_{LF})+p(\OPT_{SF}))$. Note that since we are in the case when there is no huge item in the optimal packing, no item can simultaneously satisfy $w(i) \ge (1-\epsilon)N$ and $h(i)\ge (1-\epsilon)N$. We apply Lemma \ref{lem:resource-contraction} with $I=\OPT_{LF}\cup \OPT_{SF}$ and obtain a packing of profit at least $\frac{1}{2}(p(\OPT_{LF})+p(\OPT_{SF}))$ which in addition contains an empty strip $\mathcal{S}$ of height $\epsilon N/2$. Due to Lemma \ref{lem:bound-thin-items}, we can pack the items in $\OPT_{LT}\cup \OPT_{ST}$ inside $\mathcal{S}$ and use the remaining height of $(\epsilon/2-\epsilon^4)N$ inside $\mathcal{S}$ for resource augmentation (Lemma \ref{lem:resource-augmentation}) to obtain a container packing of the desired profit stated in the lemma.
\end{proof}

Now we further classify $\OPT_{LF}$ into two categories. Let $\OPT_{LF_{\ell}} := \{i \in \OPT_{LF} \mid h_i > \epsilon N \text{ and } w_i > (1/2+\epsilon)N\} \cup \{i\in \OPT_{LF}\mid h_i > (1/2+\epsilon)N \text{ and } w_i > \epsilon N\}$, and $\OPT_{LF_s} := \OPT_{LF}\setminus \OPT_{LF_{\ell}}$. Note that the dimensions of the items in $\OPT_{LF_{\ell}}$ are sufficiently large so that all items of $\OPT_{ST}$ can be packed in the place of a single item of $\OPT_{LF_{\ell}}$. Using the following lemma, we obtain a container packing by deleting a {\em random strip} inside the knapsack. A similar random strip argument will be crucially used in multiple subsequent lemmas, where we delete all items intersected by a random strip and argue that we still retain sufficient profit. Then we use the empty strip region to pack some thin items (and small items).

\begin{lemma}
\label{lem:randomstrip}
    There exists a container packing of profit at least $(1-\epsilon)(\frac{1}{2}p(\OPT_{LF_s})+\frac{1}{4}p(\OPT_{LF_{\ell}})+ \frac{3}{4}p(\OPT_{SF})+p(\OPT_{LT})+p(\OPT_{ST}))$.
\end{lemma}
\begin{proof}
    We consider the packing of the items of $\OPT_{LF}\cup \OPT_{SF}$ inside the knapsack and consider a strip $\mathcal{S}$ of thickness $\epsilon N$, where $\mathcal{S}$ is chosen to be a random horizontal strip with probability $1/2$, and a random vertical strip with probability $1/2$. 
    \ak{Here, by random horizontal strip we mean a strip $\mathcal{S}^H:=[0,N]\times[aN, (a+\eps) N]$, where $a \in [0, 1-\eps]$ chosen uniformly at random. Similarly, a random vertical strip denotes $\mathcal{S}^V := [aN, (a+\eps) N] \times [0,N]$, where $a \in [0, 1-\eps]$ chosen uniformly at random. 
    Thus, both  $\mathcal{S}^H$ and $\mathcal{S}^V$ are fully contained in the knapsack. Finally, we choose $\mathcal{S}^H$ or $\mathcal{S}^V$, both with probability 1/2.}
    We delete all items of $\OPT_{LF}\cup \OPT_{SF}$ intersecting $\mathcal{S}$. Each item in $\OPT_{SF}$ is skewed, i.e., one of its dimensions has length at most $\epsilon_{\text{small}}N$ ($\le \epsilon N$), and the other has length at most $(1/2+2\epslarge)N\le (1/2+2\epsilon)N$, therefore it survives with probability at least \dk{$\frac{1}{2}\cdot (1-\frac{2\epsilon}{1-\eps})+\frac{1}{2}\cdot (\frac{1}{2}-\frac{3\epsilon}{1-\eps})= \frac{3}{4}-O(\epsilon)$}. Consider now an item $i\in \OPT_{LF_s}$. If the longer side of $i$ has length more than $(1/2+\epsilon)N$, the shorter side must have length at most $\epsilon N$, and thus the probability that item $i$ survives is at least \dk{$\frac{1}{2}(1-\frac{2\epsilon}{1-\eps})=\frac{1}{2}-O(\eps)$}. Otherwise both sides of $i$ have length at most $(1/2+\epsilon)N$, and then $i$ survives with probability at least $\frac{1}{2}\cdot (\frac{1}{2}-\frac{2\epsilon}{1-\eps})+\frac{1}{2}\cdot (\frac{1}{2}-\frac{2\epsilon}{1-\eps}) = \frac{1}{2}-O(\epsilon)$. Thus, every item in $\OPT_{LF_s}$ survives with probability at least $1/2 - O(\epsilon)$. Finally, since one of the sides of each item in $\OPT_{LF_{\ell}}$ does not exceed $(1/2+\epsilon)N$ \dk{(otherwise the item would have been huge)}, each such item survives with probability at least \dk{$\frac{1}{2}\cdot (\frac{1}{2}-\frac{2\epsilon}{1-\eps})=\frac{1}{4}-O(\epsilon)$}. Hence we retain an expected profit of at least $(1-O(\epsilon))(\frac{1}{2}p(\OPT_{LF_s})+\frac{1}{4}p(\OPT_{LF_{\ell}})+ \frac{3}{4}p(\OPT_{SF}))$, and free up a strip of thickness $\epsilon N$ in the process. Due to Lemma \ref{lem:bound-thin-items}, we can pack the items of $\OPT_{LT}\cup \OPT_{ST}$ (which use up a thickness of $\epsilon^4N$) and use the remaining empty space of thickness $(\epsilon-\epsilon^4)N$ for resource augmentation (Lemma \ref{lem:resource-augmentation}) to get a container packing.    
\end{proof}

In the case when $\OPT_{LF_{\ell}}=\emptyset$, Lemma \ref{lem:lcpacking-focs}(i), Lemma \ref{lem:new-processing} and Lemma \ref{lem:randomstrip} already gives a $(10/7+O(\epsilon))$-approximation.

\begin{lemma}
\label{lem:optlflempty}
    If $\OPT_{LF_{\ell}}=\emptyset$, there exists a container packing of profit at least $(7/10-\epsilon)p(\OPT)$.
\end{lemma}
\begin{proof}
    Let $\optcont$ be the maximum profit container packing. We have the following bounds on $\optcont$.
    \begin{alignat*}{2}
        2p(\optcont) &\ge (1-\epsilon)(2p(\OPT_{LF})+p(\OPT_{SF})+p(\OPT_{ST})) &\quad& \text{[Lemma \ref{lem:lcpacking-focs}(i)]} \\
        2p(\optcont) &\ge (1-\epsilon)\left(2p(\OPT_{LF})+2p(\OPT_{SF})+p(\OPT_{LT})\right) &\quad& \text{[Lemma \ref{lem:new-processing}]} \\
        6p(\optcont) &\ge (1-\epsilon)\left(3p(\OPT_{LF})+\frac{9}{2}p(\OPT_{SF})+6p(\OPT_{LT})+6p(\OPT_{ST})\right) &\quad& \text{[Lemma \ref{lem:randomstrip}]}
    \end{alignat*}
    Adding the three above inequalities, we obtain
    \begin{align*}
        10p(\optcont)&\ge (1-\epsilon)\left(7p(\OPT_{LF})+\frac{15}{2}p(\OPT_{SF})+7p(\OPT_{LT})+7p(\OPT_{ST})\right) \\
        &\ge (7-7\epsilon)p(\OPT),
    \end{align*}
    which completes the proof.
\end{proof}

We assume from now on that $\OPT_{LF_{\ell}}\neq \emptyset$. We further distinguish two cases depending on the number of items in $\OPT_{LF_{\ell}}$.

\paragraph{Case 1: $\mathbf{|\OPT_{LF_{\ell}}|\ge 2}$.}
In this case, by discarding the least profitable item of $\OPT_{LF_{\ell}}$, we are able to repack all items in $\OPT_{ST}$.

\begin{lemma}
\label{lem:delete-and-repack-st}
    There exists a container packing of profit at least $(1-\epsilon)(p(\OPT_{LF_s})+p(\OPT_{SF})+\frac{1}{2}p(\OPT_{LT})+p(\OPT_{ST}))+\frac{1}{2}p(\OPT_{LF_{\ell}})$.
\end{lemma}
\begin{proof}
    Using Lemma \ref{lem:new-processing}, we first get a container packing of profit at least $(1-\epsilon)(p(\OPT_{LF})+p(\OPT_{SF})+\frac{1}{2}p(\OPT_{LT}))$. \aw{By the construction of this packing in the proof of Lemma \ref{lem:new-processing}}, the items of $\OPT_{LF_{\ell}}$ are all present in this packing. Since one of the sides of every item in $\OPT_{LF_{\ell}}$ has length at least $(1/2+\eps)N$, by Lemma \ref{lem:bound-thin-items}, the items of $\OPT_{ST}$ can all fit in the space occupied by such an item.
    We delete the item of $\OPT_{LF_{\ell}}$ having minimum profit and pack items of $\OPT_{ST}$ in its place. Since the shorter side of the deleted item has length at least $\epsilon N$, there is still an empty space of thickness at least $(\epsilon-\epsilon^4)N$, which we can use for resource augmentation (Lemma \ref{lem:resource-augmentation}) in order to get a container packing of $\OPT_{ST}$. Since $\OPT_{LF_{\ell}}$ had at least two items, we are able to retain a profit of at least $\frac{1}{2}p(\OPT_{LF_{\ell}})$ and are done.
\end{proof}

Combining Lemmas \ref{lem:lcpacking-focs}(i), \ref{lem:new-processing}, \ref{lem:res-contraction} and \ref{lem:delete-and-repack-st} gives a $16/11+O(\epsilon)$-approximation for this case.

\begin{lemma}
\label{lem:optlflmorethantwo}
    If $|\OPT_{LF_{\ell}}|\ge 2$, there exists a container packing of profit at least $(11/16-\epsilon)p(\OPT)$.
\end{lemma}
\begin{proof}
    Let $\optcont$ be the maximum profit container packing. We have the following bounds on $\optcont$.
    \begin{alignat*}{2}
        2p(\optcont) &\ge (1-\epsilon)(2p(\OPT_{LF})+p(\OPT_{SF})+p(\OPT_{ST})) &\quad& \text{[Lemma \ref{lem:lcpacking-focs}(i)]} \\
        4p(\optcont) &\ge (1-\epsilon)\left(4p(\OPT_{LF})+4p(\OPT_{SF})+2p(\OPT_{LT})\right) &\quad& \text{[Lemma \ref{lem:new-processing}]} \\
        8p(\optcont) &\ge (1-\epsilon)\left(8p(\OPT_{LT})+8p(\OPT_{ST})+4p(\OPT_{LF})+4p(\OPT_{SF})\right) &\quad& \text{[Lemma \ref{lem:res-contraction}]} \\
        2p(\optcont) &\ge (1-\epsilon)(p(\OPT_{LF})+2p(\OPT_{SF})+ p(\OPT_{LT})+2p(\OPT_{ST})) &\quad& \text{[Lemma \ref{lem:delete-and-repack-st}]}
    \end{alignat*}
    Adding the four above inequalities, we obtain
    \begin{align*}
        16p(\optcont)&\ge (1-\epsilon)(11p(\OPT_{LF})+11p(\OPT_{SF})+11p(\OPT_{LT})+11p(\OPT_{ST})) \\
        &= (11-11\epsilon)p(\OPT),
    \end{align*}
    completing the proof.
\end{proof}

\paragraph{Case 2: $\mathbf{|\OPT_{LF_{\ell}}|=1}$.}
In this case, we cannot afford to discard the single item in $\OPT_{LF_{\ell}}$, as it may contain the entire profit of $\OPT_{LF}$. For instance, a distribution of profits that does not yield a better than $1.5$-approximation using all our previous lemmas is the following: $p(\OPT_{LF_{\ell}})=p(\OPT_{SF})=p(\OPT_{ST})=\frac{1}{3}p(\OPT)$. However, we show that it is possible to pack the single item of $\OPT_{LF_{\ell}}$ together with a constant fraction of the profit of $\OPT_{SF}$ below it. Then, we can free up an empty horizontal strip inside the packing of $\OPT_{SF}$, and repack items of $\OPT_{ST}$ in the strip.

\begin{lemma}
\label{lem:one-by-twenty-sf}
    There exists a container packing of profit at least $p(\OPT_{LF_{\ell}})+ (1-\epsilon)(p(\OPT_{LT})+p(\OPT_{ST})+\frac{1}{20}p(\OPT_{SF}))$.
\end{lemma}
\begin{proof}
\begin{figure}
    \centering
    \includegraphics[width=0.5\linewidth]{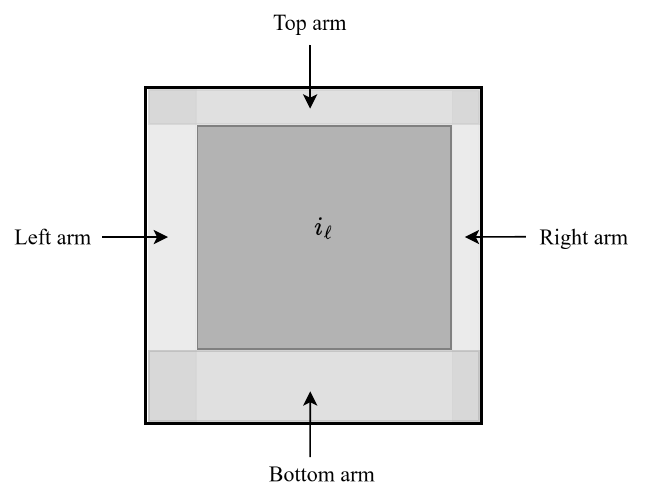}
    \caption{\dk{Partitioning the region outside $i_{\ell}$ into four arms.}}
    \label{fig:arms}
\end{figure}
Let $i_{\ell}$ denote the single item in $\OPT_{LF_{\ell}}$.
We divide the region surrounding the item $i_{\ell}$ into four arms as shown in Figure \ref{fig:arms}. The \textit{top} arm consists of all points inside the knapsack whose $y$-coordinate exceeds the $y$-coordinate of the top edge of $i_{\ell}$. Analogously, we define the \textit{bottom}, \textit{left} and \textit{right} arms. We assume w.l.o.g.~that the thickness of the bottom arm is the largest among the four arms. Let $h_b$ be the height of the bottom arm. Since $i_{\ell}$ is not a huge item, one of the sides of $i_{\ell}$ must have length at most $(1/2+\epsilon)N$, and therefore it follows that $h_b \ge \frac{1}{2}(\frac{1}{2}-\epsilon)N > (\frac{1}{4}-\epsilon)N$.

    \begin{figure}
    \centering
    \includegraphics[width=\linewidth]{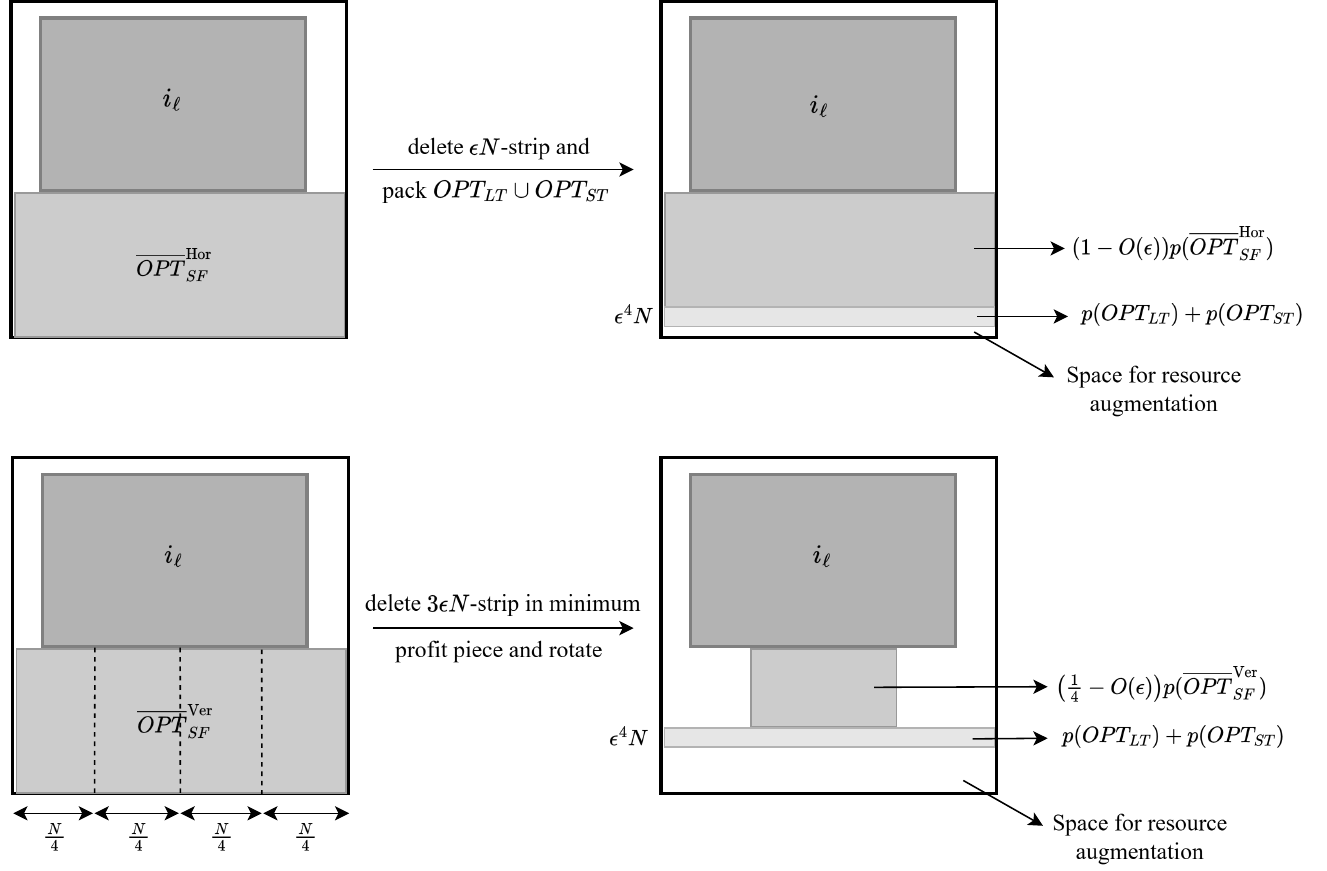}
    \caption{\dk{Packing corresponding to Lemma \ref{lem:one-by-twenty-sf}.}}
    \label{fig:nonmassivecase}
\end{figure}

    Consider now the items of $\OPT_{SF}$. Our goal is to pack a subset of these items of profit $\approx \frac{1}{20}p(\OPT_{SF})$ below the item $i_{\ell}$ so that there is still an $O(\epsilon)N$-height empty strip left that can be used for packing the items of $\OPT_T$ and for resource augmentation. For this, note that out of the top, bottom, left and right arms, there must be one arm such that the items of $\OPT_{SF}$ completely lying inside the arm have a profit of at least $\frac{1}{4}p(\OPT_{SF})$. Say, this holds for the bottom arm, otherwise we can transfer the items from the highest profit arm to the bottom arm. Now we have a packing of a subset $\overline{\OPT}_{SF} \subseteq \OPT_{SF}$ of profit $p(\overline{\OPT}_{SF}) \ge \frac{1}{4}p(\OPT_{SF})$ inside the bottom arm. Let $\overline{\OPT}_{SF}^{\text{Hor}} \subseteq \overline{\OPT}_{SF}$ be the horizontal items in $\overline{\OPT}_{SF}$, i.e., the items of height at most $\epsilon_{\text{small}}N$, and $\overline{\OPT}_{SF}^{\mathrm{ver}} := \overline{\OPT}_{SF}\setminus \overline{\OPT}_{SF}^{\text{Hor}}$. We consider the packing of $\overline{\OPT}_{SF}^{\text{Hor}}$ and discard items intersecting a random horizontal strip of height $\epsilon N$ in the bottom arm, \dk{similar to Lemma \ref{lem:randomstrip}}. Clearly, each item survives with probability at least $1-O(\epsilon)$, so we retain a profit of at least $(1-O(\epsilon))p(\overline{\OPT}_{SF}^{\text{Hor}})$, and are able to free up a strip of height $\epsilon N$ in the process. Now consider the packing of $\overline{\OPT}_{SF}^{\mathrm{ver}}$. We partition the bottom arm into four equal pieces by constructing vertical segments separated by a distance of $N/4$ (see Figure \ref{fig:nonmassivecase}). Then the items intersecting one of these pieces must have a profit of at least $\frac{1}{4}p(\overline{\OPT}_{SF}^{\mathrm{ver}})$. By deleting a random vertical strip of width $3\epsilon N$, we get a packing of profit at least $(1/4-O(\epsilon))p(\overline{\OPT}_{SF}^{\mathrm{ver}})$ into a box of width $(1/4-2\epsilon)N$ and height $h_b$. We now rotate this box and place it below the item $i_{\ell}$ (see Figure \ref{fig:nonmassivecase}). Since the height of the bottom arm is at least $(1/4-\epsilon)N$, in both cases, we get an empty horizontal strip of height $\epsilon N$, and pack items of $\OPT_{SF}$ having a profit of at least $(1-O(\epsilon))\max\{p(\overline{\OPT}_{SF}^{\text{Hor}}),\frac{1}{4}p(\overline{\OPT}_{SF}^{\mathrm{ver}})\} \ge (1-O(\epsilon))\cdot\frac{1}{5}p(\overline{\OPT}_{SF})\ge (\frac{1}{20}-O(\epsilon))p(\OPT_{SF})$. Finally in the empty strip, we pack the items of $\OPT_{LT}\cup \OPT_{ST}$ which occupy a height of at most $\epsilon^4N$ (Lemma \ref{lem:bound-thin-items}), and utilize the remaining height of $(\epsilon - \epsilon^4)N$ for resource augmentation (Lemma \ref{lem:resource-augmentation}).
\end{proof}

Combining Lemmas \ref{lem:lcpacking-focs}(i), \ref{lem:new-processing}, \ref{lem:res-contraction}, \ref{lem:randomstrip} and \ref{lem:one-by-twenty-sf}, we show the following result.

\begin{lemma}
\label{lem:optlflone}
    If $|\OPT_{LF_{\ell}}|=1$, there exists a container packing of profit at least $(127/190-\epsilon) p(\OPT)$.
\end{lemma}
\begin{proof}
    As before, let $\optcont$ denote the maximum profitable container packing. We have the following lower bounds on $p(\optcont)$.
    \begin{adjustwidth}{-0.5cm}{-1cm}
\begin{alignat*}{2}
42p(\optcont) &\ge (1-\epsilon)(42p(\OPT_{LF})+21p(\OPT_{SF})+21p(\OPT_{ST})) &\quad& \text{[Lemma \ref{lem:lcpacking-focs}(i)]} \\
42p(\optcont) &\ge (1-\epsilon)\left(42p(\OPT_{LF})+42p(\OPT_{SF})+21p(\OPT_{LT})\right) &\quad& \text{[Lemma \ref{lem:new-processing}]} \\
6p(\optcont)  &\ge (1-\epsilon)\big(6p(\OPT_{LT}) + 6p(\OPT_{ST}) + 3p(\OPT_{LF}) \\
              &\qquad\qquad\quad  + 3p(\OPT_{SF})\big) &\quad& \text{[Lemma \ref{lem:res-contraction}]} \\    
80p(\optcont) &\ge (1-\epsilon)\big(40p(\OPT_{LF_s}) + 20p(\OPT_{LF_{\ell}}) + 60p(\OPT_{SF}) \\
              &\qquad\qquad\quad + 80p(\OPT_{LT}) + 80p(\OPT_{ST})\big) &\quad& \text{[Lemma \ref{lem:randomstrip}]} \\
20p(\optcont) &\ge (1-\epsilon)\big(20p(\OPT_{LF_{\ell}}) + 20p(\OPT_{LT}) + 20p(\OPT_{ST})\\
              &\qquad\qquad\quad  + p(\OPT_{SF})\big) &\quad&  \text{[Lemma \ref{lem:one-by-twenty-sf}]} \\
\end{alignat*}
\end{adjustwidth}
    Adding the five inequalities above, we obtain
    \begin{align*}
        190p(\optcont) &\ge (1-\epsilon)(127p(\OPT_{LF})+127p(\OPT_{SF})+127p(\OPT_{LT})+127p(\OPT_{ST})) \\
        &= (127-127\epsilon)p(\OPT),
    \end{align*}
    which completes the proof.
\end{proof}

We are now ready to prove \Cref{lem:no-huge-item}.

\nohugeitem*
\begin{proof}
    If $\OPT_{LF_{\ell}}=\emptyset$, \Cref{lem:optlflempty} gives a $(10/7+\epsilon)$-approximation. Otherwise, if $\OPT_{LF_{\ell}}$ has at least two items, \Cref{lem:optlflmorethantwo} yields a $(16/11+\eps)$-approximation. Finally, for the case when $\OPT_{LF_{\ell}}$ consists of a single item, \Cref{lem:optlflone} achieves an approximation ratio of $(190/127+\eps)$. Therefore, the worst-case approximation ratio of our algorithm is $(190/127+\eps)$.
\end{proof}

\paragraph{Handling small items.}
\dk{We discuss now how to repack the items of $\OPT'_{\text{small}}$ \ak{(see handling of small  items in \Cref{sec:non-rotation-improved}), where we are left to pack  $\OPT'_{\text{small}}$ with area at most $O(\eps^3)N^2$}. Observe that in the packings constructed in Lemmas \ref{lem:res-contraction}, \ref{lem:randomstrip} and \ref{lem:one-by-twenty-sf}, we obtain an empty strip of height $\Omega(\eps)N$ and width $N$ or vice versa. Also, in Lemma \ref{lem:delete-and-repack-st}, after discarding the least profitable item in $\OPT_{LF_{\ell}}$ and packing the items of $\OPT_{ST}$, we still have an empty strip of height $\Omega(\eps)N$ and width at least $(1/2+\eps)N$ (or vice versa). Since $a(\OPT'_{\text{small}})\le O(\eps^3)N^2$ and the width and height of each item in $\OPT'_{\text{small}}$ is trivially bounded by $\epssmall N$, in each case, we can pack all items in $\OPT'_{\text{small}}$ using NFDH inside one half of the strip, and use the remaining half for resource augmentation.}

\subsection{There is a huge item \texorpdfstring{$\hugeitem$}{i*} with $N-w(\hugeitem) > 2\epsilon^2 (N-h(\hugeitem))$}
\label{sec:vertical-arms-large}
Let $w=N-w(\hugeitem)$ and $h=N-h(\hugeitem)$, so that $w\le h$. As in the proof of Lemma \ref{lem:one-by-twenty-sf}, we divide the region of the knapsack not occupied by the item $\hugeitem$ into the top, bottom, left, and right arms. First, we show that if the item $\hugeitem$ is discarded, then it is possible to obtain a container packing of the remaining items.

\begin{lemma}
\label{lem:delete-big-item}
    There exists a container packing of profit at least $(1-\epsilon)(p(\OPT)-p(\hugeitem))$.
\end{lemma}
\begin{proof}
    Note that since $h < (1/2-\epsilon)N$, the items completely lying in the top and bottom arms can be rotated and packed inside a box $B$ of height $N$ and width $(1/2-\epsilon)N$. We discard the item $\hugeitem$, and pack the box $B$ in its place, keeping all the items in the left and right arms at their original positions. Since $w(\hugeitem) > (1/2+\epsilon)N$, there is still an empty vertical strip of width $2\epsilon N$ inside the knapsack. We use this empty strip for resource augmentation (Lemma \ref{lem:resource-augmentation}) and obtain a container packing of profit at least $(1-\epsilon)(p(\OPT)-p(\hugeitem))$.
\end{proof}

We next present several ways of packing the item $\hugeitem$ together with some other subset of items. For this, we apply the corridor decomposition lemma (Lemma \ref{lem:corr-decomposition}) with $\hugeitem$ as an untouchable item. Although in \cite{galvez2021approximating}, the untouchable items were all included in the set $\OPT_{LF}$, in our case we put all untouchable items except the item $\hugeitem$ in the set $\OPT_{LF}$. From Lemmas \ref{lem:lcpacking-focs}(i) and \ref{lem:new-processing}, we directly get the following bounds on the profit of a container packing.

\begin{lemma}
\label{lem:lf-sf-st}
    There exist container packings of profit at least
    \begin{enumerate}[label=(\roman*)]
        \item $(1-\epsilon)(p(\OPT_{LF})+\frac{1}{2}(p(\OPT_{SF})+p(\OPT_{ST})))+p(\hugeitem)$, and
        \item $(1-\epsilon)(p(\OPT_{LF})+p(\OPT_{SF})+\frac{1}{2}p(\OPT_{LT}))+p(\hugeitem)$.
    \end{enumerate}
\end{lemma}

Now we present another restructuring of the optimal packing from which we can obtain almost the full profit of the thin items, together with a constant fraction of the profit of $\OPT_{SF}$.

\begin{lemma}
\label{lem:one-by-twelve-sf}
    There exists a container packing of profit at least $(1-\epsilon)(p(\OPT_{LT})+p(\OPT_{ST})+\frac{1}{12}p(\OPT_{SF}))+p(\hugeitem)$.
\end{lemma}
\begin{proof}
    First, we fix the item $\hugeitem$ in its original position inside the optimal packing. Let $h_t, h_b, w_{\ell},w_r$ be the widths of the top, bottom, left and right arms around $i^*$, respectively, and assume w.l.o.g.~that $h_b\ge h_t$ and $w_r \ge w_{\ell}$. Since we are in the case when $w\ge 2\epsilon^2h$, it holds that $w_r > \epsilon^2 h$. In the following, we shall also assume that the bottom arm is the thickest among the four arms, i.e., $h_b \ge w_r$; the case when the right arm is the thickest is similar and will be discussed at the end.

    Consider the items of $\OPT_{LT}$. Since the item $\hugeitem$ has both of its dimensions exceeding $N/2$, the horizontal (resp. vertical) items in $\OPT_{LT}$ must completely lie in the top and bottom (resp. left and right) arms. After processing the subcorridors with parameter $\epsilon_{\text{box}}$, the height of the horizontal items of $\OPT_{LT}$ inside each subcorridor is at most $\epsilon_{\text{box}}h$. Since the number of subcorridors is at most $c_{\eps,\epslarge}$, the total height of the horizontal items of $\OPT_{LT}$ is at most $\epsilon_{\text{box}}h\cdot c_{\eps,\epslarge} \le \epsilon^4 h \le \epsilon^2 w/2$, where the first inequality holds by our choice of $\eps_{\text{box}}$, and the second inequality holds since we are in the case when $w>2\epsilon^2 h$. Similarly, the total width of the vertical items can be bounded by $\epsilon^2w/2$. We rotate the horizontal items by 90 degrees, so that the total width of $\OPT_{LT}$ is now at most $\epsilon^2 w$. These will be packed together with a subset of items from $\OPT_{ST}$ at the right boundary of the knapsack.

    \begin{figure}
    \centering
    \includegraphics[width=0.9\linewidth]{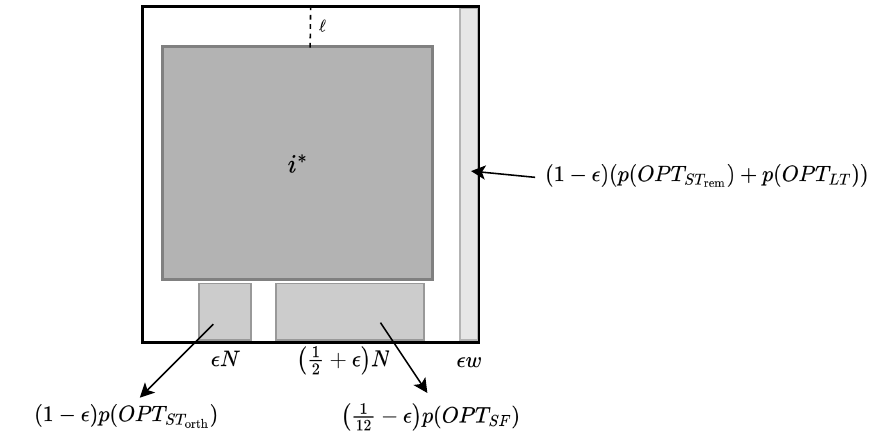}
    \caption{\dk{Packing corresponding to Lemma \ref{lem:one-by-twelve-sf}.}}
    \label{fig:massive-case}
\end{figure}

    Consider now the items of $\OPT_{ST}$. We classify them into two categories as follows. Let $\OPT_{ST_{\text{orth}}}\subseteq \OPT_{ST}$ be those items coming from vertical subcorridors completely lying in the top and bottom arms, and horizontal subcorridors completely lying in the left and right arms. After rotating the items of $\OPT_{ST_{\text{orth}}}$ coming from the left and right arms by 90 degrees, the total width of the items in $\OPT_{ST_{\text{orth}}}$ can be bounded by $\epsilon_{\text{box}}N\cdot c_{\epsilon,\epsilon_{\text{large}}} \le \epsilon^4 N$. Notice that the width $w_r$ of the right arm may be much smaller than $\epsilon^4 N$, and therefore it might not be possible to pack the items of $\OPT_{ST_{\text{orth}}}$ in the right arm -- in fact, it may even be the case that not a single item from the bottom arm fits into the right arm. However, since the bottom arm was assumed to be the thickest, the items of $\OPT_{ST_{\text{orth}}}$ can be packed into a box of width $\epsilon^4 N$ in the bottom arm, and by allowing an additional width of $(\epsilon-\eps^4) N$ for resource augmentation (Lemma \ref{lem:resource-augmentation}), we can obtain a container packing of these items (see Figure \ref{fig:massive-case}).

    Next let $\OPT_{ST_{\text{rem}}} := \OPT_{ST}\setminus \OPT_{ST_{\text{orth}}}$ be the remaining items of $\OPT_{ST}$. Similar to $\OPT_{LT}$, their total width (after rotating the horizontal items) can be bounded by $\epsilon_{\text{box}}h\cdot c_{\epsilon, \epsilon_{\text{large}}} \le \epsilon^2 w$. Thus, the total width of the items of $\OPT_{LT}\cup \OPT_{ST_{\text{rem}}}$ is at most $2\epsilon^2 w$, and therefore they can be packed into a strip of the same width at the right boundary of the knapsack. We increase the width of the strip to $\epsilon w$, so that now there is sufficient space for resource augmentation (Lemma \ref{lem:resource-augmentation}) in order to obtain a container packing (see Figure \ref{fig:massive-case}).

    It remains to pack (a subset of) items from $\OPT_{SF}$. For this, we first note that out of the top, bottom, left, and right arms, there must be an arm such that the items of $\OPT_{SF}$ completely lying inside the arm have a profit of at least $p(\OPT_{SF})/4$. Assume that this holds for the top arm (the other cases are symmetric). Consider the vertical line segment $\ell$ passing through the middle of the arm, as shown in Figure \ref{fig:massive-case}. We classify the items of $\OPT_{SF}$ lying in the top arm into three groups -- those lying to the left of $\ell$, those lying to the right of $\ell$, and those crossing the segment $\ell$. Clearly, one of these groups will have a profit of at least $p(\OPT_{SF})/12$ and the items in the group can be completely packed into a box of width $(1/2+2\epslarge)N$ (since the width of each short item is at most $(1/2+2\epslarge)N$). We place this box in the bottom arm (possibly after rotation if the box comes from one of the vertical arms), noting that this is always feasible by our assumption that the bottom arm is the thickest. Finally, we increase the width of the box to $(1/2+\epsilon)N$ so that there is a width of at least $(\epsilon-2\epslarge)N > \epsilon N/2$ for resource augmentation (Lemma \ref{lem:resource-augmentation}). See Figure \ref{fig:massive-case} for the final packing inside the knapsack.

    In the above argumentation, the bottom arm was assumed to be the thickest among the four arms around the item $\hugeitem$. This was essential to ensure that items from $\OPT_{ST_{\text{orth}}}$ and $\OPT_{SF}$ could be packed in the bottom arm, possibly after rotation. For the case when the right arm is the thickest, we employ a similar argumentation as before, and pack the items from $\OPT_{ST_{\text{orth}}}$ and $\OPT_{SF}$ in the right arm instead. Further, since $h_b\ge h/2 \ge \epsilon w$, the strip of width $\epsilon w$ that packs items from $\OPT_{LT}\cup \OPT_{ST_{\text{rem}}}$ can be packed at the bottom boundary of the knapsack, and we are done.

    \dk{Finally we pack the items of $\OPT'_{\text{small}}$. Again, we assume for simplicity that the bottom arm is the thickest; the case when the right arm is the thickest is similar. Observe that since $i^*$ is an untouchable item, any cell of the non-uniform grid formed by extending the edges of all untouchable items must completely lie inside one of the eight cells formed by extending the edges of $i^*$. Also note that one of the sides of each of these eight cells has a length of at most $h_b$, where $h_b$ is the height of the bottom arm. Therefore, each item of $\OPT'_{\text{small}}$ has one side of length at most $\epssmall\cdot h_b$ (and the length of the other side is trivially at most $\epssmall N$). Therefore, since $a(\OPT'_{\text{small}})\le O(\eps^3)\cdot(N^2-a(i^*))$, and the area of the bottom arm is at least $\frac{1}{4}(N^2-a(i^*))$, the items of $\OPT'_{\text{small}}$ can easily be packed using NFDH inside a box of width $\eps N$ and height $h_b$, which can be placed in the bottom arm.}
\end{proof}

Combining Lemmas \ref{lem:lf-sf-st} and \ref{lem:one-by-twelve-sf}, we obtain the following.

\begin{lemma}
\label{lem:container-packing-guarantee-1}
    There exists a container packing of profit at least $(13/25-\epsilon)(p(\OPT_{LF})+p(\OPT_{SF})+p(\OPT_{LT})+p(\OPT_{ST}))+p(\hugeitem)$.
\end{lemma}
\begin{proof}
    Let $\OPT_{\text{cont}}$ be the maximum profitable container packing. We have the following guarantees on $\optcont$.
    \begin{alignat*}{2}
        2p(\optcont) &\ge (1-\epsilon)(2p(\OPT_{LF})+p(\OPT_{SF})+p(\OPT_{ST}))+2p(\hugeitem) &\quad& \text{[Lemma \ref{lem:lf-sf-st}(i)]} \\
        11p(\optcont) &\ge (1-\epsilon)\left(11p(\OPT_{LF})+11p(\OPT_{SF})+\frac{11}{2}p(\OPT_{LT})\right) + 11p(\hugeitem) &\quad& \text{[Lemma \ref{lem:lf-sf-st}(ii)]} \\
        12p(\optcont) &\ge (1-\epsilon)(12p(\OPT_{LT})+12p(\OPT_{ST})+p(\OPT_{SF})) + 12p(\hugeitem) &\quad& \text{[Lemma \ref{lem:one-by-twelve-sf}]}
    \end{alignat*}
    Adding the three inequalities above, we get
    \begin{align*}
        25p(\optcont) &\ge (1-\epsilon)\left(13p(\OPT_{LF})+13p(\OPT_{SF})+\frac{35}{2}p(\OPT_{LT})+13p(\OPT_{ST})\right) + 25p(\hugeitem) \\
        &\ge (13-13\epsilon)(p(\OPT_{LF})+p(\OPT_{SF})+p(\OPT_{LT})+p(\OPT_{ST})) + 25p(\hugeitem),
    \end{align*}
    completing the proof.
\end{proof}

Finally, we combine Lemmas \ref{lem:delete-big-item} and \ref{lem:container-packing-guarantee-1} to achieve an overall approximation ratio of $(37/25+\eps)$, thus proving \Cref{lem:container-packing=000020exists}. 

\smallhugeitem*
\begin{proof}
    Letting $\optcont$ be the maximum profitable container packing, Lemma \ref{lem:container-packing-guarantee-1} gives us that $p(\optcont)\ge (13/25-\epsilon)p(\OPT)+(12/25+\epsilon)p(\hugeitem)$, and therefore
    \[ 25p(\optcont) \ge (13-25\epsilon)p(\OPT)+(12+25\epsilon)p(\hugeitem)\]
    Also, from Lemma \ref{lem:delete-big-item}, we have
    \[ 12p(\optcont) \ge (1-\epsilon)(12p(\OPT)-12p(\hugeitem))\]
    Adding the above two inequalities gives us the desired guarantee stated in the lemma.
\end{proof}

\subsection{\lc-packings}
\label{sec:lc*-packings}

We now formally define the notion of \lc-packings that will be the main technical ingredient in order to obtain an improved approximation ratio than $1.5$ under the conditions of \Cref{lem:LC-packing-hugeitem}. It is an adaptation of L\&C packings used in \cite{galvez2021approximating}. We shall then present an algorithm to compute the most profitable \lc-packing using a dynamic program. Finally we show that there exists \lc-packings of high profit.

\begin{figure}
    \centering
    \includegraphics[width=0.4\linewidth]{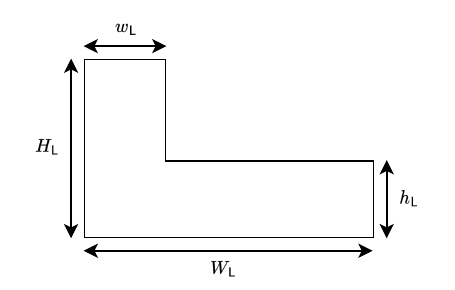}
    \caption{An $\Ls$ corridor.}
    \label{fig:l-corridor}
\end{figure}

An \emph{$\Ls$-corridor
}is an $\Ls$-shaped region $\lcor=((0,W_{\lcor})\times(0,h_{\lcor}))\cup((0,w_{\lcor})\times(0,H_{\lcor}))$
for values $H_{\lcor},W_{\lcor},h_{\lcor},w_{\lcor}\in\{1,...,N\}$
with $w_{\lcor}\le W_{\lcor}$ and $h_{\lcor}\le H_{\lcor}$ (see Figure \ref{fig:l-corridor}). We say
that $(0,W_{\lcor})\times(0,h_{\lcor})$ is the \emph{horizontal arm
}of $\lcor$ and $(0,w_{\lcor})\times(0,H_{\lcor})$ is the \emph{vertical
arm} of $\lcor$.

\begin{definition}[\lc-packing]
 \label{def:l-packing} Consider a packing of a set of items $I'\subseteq I$
inside $K$, a set of containers~$\C$, and an $\Ls$-corridor $\lcor$.
They form an \emph{\lc-packing} if the containers in $\C$ and $\lcor$
are pairwise disjoint and satisfy the following properties:
\begin{enumerate}
\renewcommand{\labelenumi}{(\theenumi)}
\item each item $i\in I'$ is packed inside $\lcor$ or one of the boxes
in $\C$,
\item  the bottom edge of the horizontal arm of $\lcor$ coincides with
the bottom boundary of $K$,
\item $W_{\Ls}=N$, $H_{\Ls}\le N/2$, $w_{\Ls}\le\epsilon N$, and $h_{\Ls}\le\epsilon N$,
\item each item $i\in I$ packed in the horizontal (respectively, vertical) arm
of $\lcor$ satisfies $w^{*}(i)>N/2$ (respectively, $h^{*}(i)>\frac{1}{2}H_{\Ls}$),
where $w^{*}(i)$ and $h^{*}(i)$ are the width and height of $i$
in the packing (i.e., $w^{*}(i)=h(i)$ and $h^{*}(i)=w(i)$ if $i$
is rotated), and
\item let $I'\aw{'}$ be the set of items packed in the vertical arm of the $\Ls$-corridor;
either $h^{*}(i)>w^{*}(i)$ for each $i\in I'\aw{'}$ or $w^{*}(i)\ge h^{*}(i)$ for each
$i\in I'\aw{'}$.
\end{enumerate}
\end{definition}

One important consequence of Properties (2), (3), and (4) \aw{of} Definition~\ref{def:l-packing}
is that \aw{even in the rotational case, each} input item $i\in I$ can be packed either only in the
horizontal arm, or only in the vertical arm, or in none of them in
\emph{any} \lc-packing. \aw{Also, we can easily compute
a corresponding partition of the items. This would not have worked using the
definition of L\&C packings from \cite{galvez2021approximating}.}
\begin{lemma}
\label{lem:partition-items}
In polynomial time we can compute a partition $I=I_{H}\dot{\cup}I_{V}\dot{\cup}I_{R}$
such that
\begin{itemize}
\item each item $i\in I_{H}$ cannot be packed in the vertical arm of $\lcor$
in any L\&C{*}-packing but in the horizontal arm of some L\&C{*}-packing,
however, only with one of its two possible orientations,
\item each item $i\in I_{V}$ cannot be packed in the horizontal arm of
$\lcor$ in any L\&C{*}-packing but in the vertical arm of some L\&C{*}-packing,
and
\item each item $i\in I_{R}$ cannot be packed in $\lcor$ in any L\&C{*}-packing.
\end{itemize}
\end{lemma}
\begin{proof}
    Let $I_H \subseteq I$ be the set of items $i\in I$ for which $\max \{h(i),w(i)\}>N/2$ but
    $\min \{h(i),w(i)\}\le \eps N$; these items can potentially be packed into the horizontal arm of the $\lcor$ and into the containers. By Property (4) of Definition \ref{def:l-packing}, if an item $i\in I_H$ is packed in the horizontal arm of an $\Ls$, it must be oriented such that $w^*(i)=\max \{h(i),w(i)\}$.
Let $I_V\subseteq I$ be the set of items
$i\in I$ for which $\max \{h(i),w(i)\} \in (\frac{1}{2}H_{\Ls},H_{\Ls}]$ and
    $\min \{h(i),w(i)\}\le \eps N$.
Since $H_{\Ls}\le N/2$, for each item $i\in I_V$ we have that $w(i)\le N/2$ and $h(i)\le N/2$. By Property (4) of Definition~\ref{def:l-packing} we are not allowed to pack these items in the horizontal arm of the $\Ls$. Finally, let $I_R:= I \setminus (I_H\cup I_V)$ be the remaining items, which can only be packed into the containers.
\end{proof}

\begin{figure}
    \centering
    \includegraphics[width=0.5\linewidth]{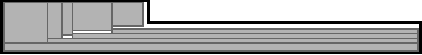}
    \caption{Items in the vertical arm may also fit after rotation.}
    \label{fig:vertical-l-arm}
\end{figure}

In contrast to $I_{H}$, there can be items $i\in I_{V}$ that could
fit inside (the vertical arm) of $\lcor$ with or without rotation.
This may happen, e.g., if $i$ is almost a square and if $H_{\lcor}$ and $h_{\lcor}$
are almost identical, see Figure~\ref{fig:vertical-l-arm}.
Inherently, the DP of \cite{galvez2021approximating} no longer works if some item $i\in I_{H}\cup I_{V}$ \aw{may} be packed inside $\lcor$ in both of its possible orientations.
However, due to Property (5)~of
Definition~\ref{def:l-packing} we are able to guess the orientation of all items in $I_{V}$ in an optimal L\&C{*}-packing, as there are
only two options for this. Now the orientation of each item is fixed and we extend the dynamic program in \cite{galvez2021approximating}
to solve the subproblem of finding the most profitable packing inside $\lcor$ and \aw{the} $O_{\eps}(1)$ \aw{additional} containers.

\begin{lemma}
\label{lem:compute-LC-packing-optimal}
For any $c\in\N$, there is
an $(1+\eps)$-approximation algorithm with a running time of $n^{O_{\eps}(c)}$
for computing the most profitable L\&C{*}-packing with at most $\aw{c}$
containers (and an $\Ls$-corridor).
\end{lemma}

In order to prove \Cref{lem:compute-LC-packing-optimal}, let $I=I_H\dot{\cup}I_V\dot{\cup}I_R$ be the partition of the input items given by Lemma \ref{lem:partition-items}. We first guess the orientation of the items of $I_V$, i.e., whether $h(i)>w(i)$ holds or $w(i)\ge h(i)$ holds for all items $i$ packed in the vertical arm of the $\lcor$. Let $\OPT_{L\& C^*}$ denote the optimal \lc-packing.
Let $\CC^h, \CC^v$ and $\CC^a$ denote the horizontal, vertical, and area containers, respectively in $\OPT_{L\& C^*}$. Consider a horizontal container $C\in \CC^h$. We permute the items inside $C$ so that the items of $\Ihor\cup \Iver$ all lie above the items of $\Irem$. Then we split $C$ into two containers at the boundary between the two sets of items. Let $\CC^h_{\text{skew}}$ be the set of containers that pack items from $\Ihor\cup \Iver$ and $\CC^h_{\text{rem}}$ be the set of containers inside which items from $\Irem$ are packed. Doing an analogous procedure for the vertical containers, we obtain the sets $\CC^v_{\text{skew}}$ and $\CC^v_{\text{rem}}$.

We call an \lc-packing to be \textit{restricted} if there exist sets $\mathcal{T}_{\mathsf{L}}, \mathcal{R}_{\mathsf{L}}$ corresponding to the $\lcor$; $\mathcal{T}_C$ for each horizontal container $C \in \CC^h_{\text{skew}}$; and $\mathcal{R}_C$ for each vertical container $C\in \CC^v_{\text{skew}}$, which can be computed in polynomial time based on the input, such that
    \begin{enumerate}[label=(\roman*)]
        \item for each item $i$ packed in the horizontal (resp. vertical) arm of the $\lcor$, the distance between the top (resp. right) edge of $i$ from the bottom (resp. left) edge of the $\lcor$ lies in $\mathcal{T}_{\mathsf{L}}$ (resp. $\mathcal{R}_{\mathsf{L}}$),
        \item for each horizontal (resp. vertical) container $C \in \CC^h_{\text{skew}}$ (resp. $\CC^v_{\text{skew}}$), the distance between the top (right) edge of any item packed in $C$ and the bottom (resp. left) edge of $C$ lies in $\mathcal{T}_C$ (resp. $\mathcal{R}_C$).        
    \end{enumerate}
    
    The idea behind the above definition is to restrict the possible positions of the items inside the containers to polynomial-sized sets. Since an item can potentially be packed in either the $\lcor$ or into one of the containers, our DP table needs to have a cell for each possible configuration of the amount of empty space remaining inside the arms of the $\lcor$ and the containers. Let \optr~denote the optimum restricted \lc-packing. Then we show the following lemma.

    \begin{lemma}
    \label{lem:restricted-l-packing}
        We have $p(\OPT_{\text{r-}L\&C^*})\ge (1-O(\epsilon))p(\OPT_{L\&C^*})$.
    \end{lemma}
    \begin{proof}
        In \cite{galvez2021approximating}, a \textit{shifting} procedure was presented that can be applied individually to each arm of the $\lcor$ to restrict the possible positions of the items to appropriate polynomial-sized sets $\TT_{\mathsf{L}}$ and $\RR_{\mathsf{L}}$. We can apply the same procedure inside the containers of $\CC^h_{\text{skew}}$ and $\CC^v_{\text{skew}}$ to obtain the sets $\TT_C$ and $\RR_C$, respectively, and are done.        
    \end{proof}

We are now ready to prove Lemma \ref{lem:compute-LC-packing-optimal}.

\begin{proof}[Proof of Lemma \ref{lem:compute-LC-packing-optimal}]
    Thanks to Lemma \ref{lem:restricted-l-packing}, it suffices to compute the optimal restricted \lc-packing. We do this using dynamic programming.
    Note that there are only polynomially-many choices for the dimensions of the $\lcor$ -- \dk{the height $H_{\Ls}$ of the vertical arm can be written as the sum of the height of some item lying in the vertical arm and a value in $\RR_{\Ls}$}, and the thickness $w_{\Ls}$ and $h_{\Ls}$ of the two arms correspond to values lying in the sets $\TT_{\Ls}$ and $\RR_{\Ls}$. We first guess the positions of the $\lcor$ and the $O_{\epsilon}(1)$ containers inside the knapsack.
    
    We first pack the items of $\Irem$. To this end, for each area container $C\in \CC^a$, we guess $a_{\text{rem}}(C)$ to be the total area occupied by items from $\Irem$, rounded up to the nearest integer multiple of $a(C)/n^2$ -- clearly there are only polynomially many choices for $a_{\text{rem}}(C)$. We now build an instance of the Generalized Assignment Problem (GAP). We define an item $R_i$ for each rectangle $i\in \Irem$, with profit $p(i)$. For each horizontal container $C\in \CC^h_{\text{rem}}$, we create a bin $R_C$ of size $s(R_C)= h(C)$. The size $s(R_i,R_C)$ for item $R_i$ w.r.t.~bin $R_C$ is defined as follows: if $R_i$ does not fit inside the container of width $w(C)$ and height $h(C)$ in either of the two possible orientations, we define $s(R_i, R_C)=\infty$; else we consider the feasible orientation for which the shorter side of $i$ is vertical, and we define $s(R_i,R_C)$ to be the height of $i$ in this orientation. Symmetrically, we define $s(R_C)$ and $s(R_i,R_C)$ for the vertical containers $C \in \CC^v_{\text{rem}}$. For an area container $C\in \CC^a$, we create a bin $R_C$ of size $s(R_C) = a_{\text{rem}}(C)$, and define the size of item $R_i$ w.r.t.~bin $R_C$ as $s(R_i,R_C)=a(i)$ if there exists an orientation of rectangle $i$ such that $w(i)\le \epsilon w(C)$ and $h(i) \le \epsilon h(C)$; otherwise we set $s(R_i,R_C)=\infty$. Applying Lemma \ref{lem:GAP-lemma} to this GAP instance, we can assign a set of items from $\Irem$ having profit at least $(1-\epsilon)p(\OPT_{\text{rem}})$ into these bins from the GAP instance.

    \begin{figure}
    \centering
    \includegraphics[width=0.6\linewidth]{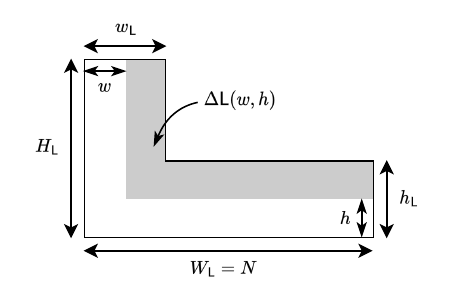}
    \caption{\dk{The shaded region corresponds to $\Delta\Ls(w,h)$.}}
    \label{fig:Lcor}
\end{figure}

    Now we shall pack items from $\Ihor \cup \Iver$ into the $\lcor$, the containers in $\CC^h_{\text{skew}}\cup \CC^{\aw{v}}_{\text{skew}}$, and in the remaining spaces in the area containers $\CC^a$. For this, we define certain regions inside the $\Ls$ and the containers, similar to \cite{galvez2021approximating}. Recall that the bottom left corner of the $\lcor$ lies at the origin $(0,0)$. For $w\in [0,w_\Ls]$ and $h\in [0,h_\Ls]$, let $\Delta\Ls(w,h) := ([w,w_\Ls]\times [h,H_\Ls])\cup ([w,N]\times [h,h_\Ls])$ (see Figure \ref{fig:Lcor}). \dk{Next, for a horizontal container $C$, we define $\Delta C(h)$ to be the region inside $C$ consisting of points lying at a distance of at most $h$ from the base of $C$. Formally, if the bottom left corner of $C$ lies at $(x_C,y_C)$, then $\Delta C(h):= [x_C,x_C+w(C)]\times [y_C,y_C+h]$.}
    Analogously, we define $\Delta C(w)$ for a vertical container $C$.

    Let $h_1, h_2,\ldots, h_{|\Ihor|}$ be the items of $\Ihor$ sorted in non-increasing order of width, and $v_1,v_2,\ldots,v_{|\Iver|}$ be the items of $\Iver$ sorted in non-increasing of height, breaking ties arbitrarily. Each DP cell is indexed by the following parameters:
    \begin{itemize}
        \item an integer $i\in [|\Ihor|]$,
        \item an integer $j\in [|\Iver|]$,
        \item values $t_{\Ls} \in \TT_\Ls$ and $r_\Ls \in \RR_\Ls$,
        \item values $t_C \in \TT_C$ for each $C\in \CC^h_{\text{skew}}$,
        \item values $r_C \in \RR_C$ for each $C\in \CC^v_{\text{skew}}$,
        \item non-negative integer values $k_C \in \left[\max\{\frac{a_{\text{rem}}(C)}{a(C)/n^2}-1,0\},n^2\right]$ for each $C\in \CC^a$.
    \end{itemize}
    $\DP[i,j,t_\Ls,r_\Ls,\{t_C\}_{C\in \CC^h_{\text{skew}}},\{r_C\}_{C\in \CC^v_{\text{skew}}}, \{k_C\}_{C\in \CC^a}]$ will store the maximum profit of a subset $\OPT' \subseteq \{h_i,\ldots,h_{|\Ihor|}\}\cup \{v_j,\ldots,v_{|\Iver|}\}$ such that
    \begin{enumerate}[label=(\roman*)]
        \item \dk{there is a subset $\OPT''\subseteq \OPT'$ that can be packed in the region $\Delta\Ls(r_\Ls,t_\Ls) \cup \bigcup_{C\in \CC^h_{\text{skew}}} \Delta C(t_C)\cup \bigcup_{C\in \CC^v_{\text{skew}}} \Delta C(r_C)$, and}
        \item \dk{there exists a partition $\{I_C\}_{C \in \CC^a}$ of the items of $\OPT'\setminus \OPT''$ such that $\sum_{i\in I_C} \left(\left\lfloor \frac{a(i)}{\frac{a(C)}{n^2}}\right\rfloor\cdot \frac{a(C)}{n^2}\right)\le (1-\frac{k_C}{n^2})a(C)$ holds for all $C \in \CC^a$.}
    \end{enumerate}
    \dk{The second condition in the above definition says that the sum of the areas of the assigned items rounded down to the nearest integer multiple of $a(C)/n^2$, should not exceed the free area of $(1-\frac{k_C}{n^2})a(C)$ remaining inside the area container.} In order to compute the above DP table entry, whose optimal solution is denoted by $\OPT'$, we consider the maximum among several cases:
    \begin{itemize}
        \item If $h_i\not\in \OPT'$, we go to the entry $\DP[i+1,j,t_\Ls,r_\Ls,\{t_C\}_{C\in \CC^h_{\text{skew}}},\{r_C\}_{C\in \CC^v_{\text{skew}}}, \{k_C\}_{C\in \CC^a}]$.
        \item If $v_j\not\in \OPT'$, we go to the entry $\DP[i,j+1,t_\Ls,r_\Ls,\{t_C\}_{C\in \CC^h_{\text{skew}}},\{r_C\}_{C\in \CC^v_{\text{skew}}}, \{k_C\}_{C\in \CC^a}]$.
        \item Assume that both $h_i,v_j \in \OPT'$.
        \begin{enumerate}[label=(\alph*)]
            \item If $h_i$ lies in some horizontal container $\widehat{C}$, then let $\overline{t_{\widehat{C}}}$ be the smallest value in $\TT_{\widehat{C}}$ such that $\overline{t_{\widehat{C}}} \ge h(h_i)+t_{\widehat{C}}$. Then $\DP[i,j,t_\Ls,r_\Ls,\{t_C\}_{C\in \CC^h_{\text{skew}}},\{r_C\}_{C\in \CC^v_{\text{skew}}}, \{k_C\}_{C\in \CC^a}] = p(h_i)+ \DP[i+1,j,t_\Ls,r_\Ls,\{t'_C\}_{C\in \CC^h_{\text{skew}}},\{r_C\}_{C\in \CC^v_{\text{skew}}}, \{k_C\}_{C\in \CC^a}]$, where $t'_C = \overline{t_C}$ if $C=\widehat{C}$, and $t'_C = t_C$, otherwise.
            \item The case when $h_i$ lies in a vertical container is analogous (note that we have to rotate $h_i$ in that case).
            \item For the case when $h_i$ lies in an area container $\widehat{C}$, we let $\overline{k(\widehat{C})}=\lfloor \frac{a(h_i)}{a(\widehat{C})/n^2}\rfloor$, and then $\DP[i,j,t_\Ls,r_\Ls,\{t_C\}_{C\in \CC^h_{\text{skew}}},\{r_C\}_{C\in \CC^v_{\text{skew}}}, \{k_C\}_{C\in \CC^a}] = p(h_i)+ \DP[i+1,j,t_\Ls,r_\Ls,\{t_C\}_{C\in \CC^h_{\text{skew}}},\\ \{r_C\}_{C\in \CC^v_{\text{skew}}}, \{k'_C\}_{C\in \CC^a}]$, where $k'(C) = k(C)+\overline{k(C)}$ if $C=\widehat{C}$, and $k'(C)=k(C)$, otherwise.
            \item Analogously, we handle the cases when $v_j$ lies in some horizontal, vertical, or area container.
            \item Finally, we have the case when both $h_i$ and $v_j$ lie inside the $\Ls$. Here, it must be the case that either there exists a horizontal guillotine cut separating $h_i$ from the rest of the packing inside the $\Ls$, or there exists a vertical guillotine cut separating $v_j$. In the former case, let $\overline{t_\Ls}$ be the smallest value in $\TT_\Ls$ such that $\overline{t_\Ls} \ge h(h_i)+t_\Ls$. Then $\DP[i,j,t_\Ls,r_\Ls,\{t_C\}_{C\in \CC^h_{\text{skew}}},\{r_C\}_{C\in \CC^v_{\text{skew}}}, \\ \{k_C\}_{C\in \CC^a}] = p(h_i) + \DP[i+1,j,\overline{t_\Ls},r_\Ls,\{t_C\}_{C\in \CC^h_{\text{skew}}},\{r_C\}_{C\in \CC^v_{\text{skew}}}, \{k_C\}_{C\in \CC^a}]$. The latter case is symmetric.
        \end{enumerate}
    \end{itemize}

    \dk{The DP computes a feasible packing of a subset of $I$ into the $\lcor$ and the horizontal and vertical containers, and an assignment of items into the area containers. We now compute a feasible packing inside the area containers. Note that since at most
    $n$ items are packed inside any area container $C$, the total area of the items $I_C$ assigned to $C$ is at most $(1+O(\frac{1}{n}))a(C)<(1+\epsilon)a(C)$. If $a(I_C) \le (1-2\epsilon)a(C)$, then we can directly pack ${I}_C$ into $C$ using Lemma \ref{lem:NFDH-guarantee}; otherwise we partition ${I}_C$ into groups of total area at least $2\epsilon\cdot a(C)$, i.e., we iteratively pick items into a group until their total area exceeds $2\epsilon\cdot a(C)$, and then restart the procedure to create another group (the last group may have a smaller total area). Since the area of each item is at most $\epsilon^2\cdot a(C)$, the number of groups is at least $\frac{1-2\epsilon}{2\epsilon+\epsilon^2}\ge \Omega(1/\epsilon)$. We delete the group having minimum profit among the ones with total area at least $2\epsilon\cdot a(C)$, and let ${I}'_C$ be the remaining items. Then $p({I}'_C) \ge (1-O(\epsilon))p(I_C)$ and $a({I}'_C)\le (1-2\epsilon)a(C)$, implying that ${I}'_C$ can be packed into $C$ using NFDH (Lemma \ref{lem:NFDH-guarantee}).}
    
    Since the computation at each DP cell only involves values stored at $O_{\epsilon}(1)$ other cells, it can be done in $O_{\epsilon}(1)$ time. Finally, since the number of DP cells is polynomially bounded, we are done.
\end{proof}

In the following subsection, we shall exploit the power of \lc-packings to obtain an approximation ratio strictly less than $1.5$ for the case when the conditions of \Cref{lem:LC-packing-hugeitem} are satisfied.

\subsection{There is a huge item \texorpdfstring{$\hugeitem$}{i*} with $N-w(\hugeitem) \le 2\epsilon^2 (N-h(\hugeitem))$}
\label{sec:large-huge-item}
We prove \Cref{lem:LC-packing-hugeitem} in this subsection. Recall that \lc-packings were defined only for a square knapsack (see Definition \ref{def:l-packing}). For convenience, in the following lemma we extend this definition (without any alteration to Properties (1)-(5)) to also apply to a rectangular knapsack, i.e., a knapsack of width $N$ and arbitrary height (not exceeding $N/2$).  Later, we shall identify a certain rectangular region inside $K$ and apply the lemma to the rectangular region, in order to obtain an \lc-packing inside the square knapsack.

\begin{restatable}{lemma}{arbitraryknapsack}
\label{lem:arbitrary-knapsack}
    Let $I$ be a set of items packed inside a rectangular box $\mathcal{R}$ with $w(\mathcal{R})=N$ and $h(\mathcal{R})\le N/2$. Then there exists a \lc-packing of a subset $I'\subseteq I$ inside the box such that $p(I')\ge (22/43-\epsilon)p(I)$.
\end{restatable}

For ease of presentation, we defer the proof of Lemma \ref{lem:arbitrary-knapsack} to \Cref{sec:arbitrary-knapsack-proof}, and prove Lemma \ref{lem:LC-packing-hugeitem} first. Consider the optimal packing $\OPT$ inside the knapsack. Note that the guarantee of Lemma \ref{lem:delete-big-item} continues to hold, i.e., if the item $\hugeitem$ is discarded, we can obtain a container packing of the remaining items.

As in \Cref{sec:vertical-arms-large}, we divide the region surrounding the huge item $\hugeitem$ into the top, bottom, left, and right arms. Let $h_t, h_b,w_{\ell},w_r$ denote the thickness of the top, bottom, left and right arms, respectively. Let $\OPT_{\text{top}}, \OPT_{\text{bottom}}$ be the items completely lying in the top and bottom arm, respectively. Among the remaining items, let $\OPT_{\text{left}}, \OPT_{\text{right}} \subseteq \OPT\setminus (\OPT_{\text{top}}\cup \OPT_{\text{bottom}})$ be the items lying in the left and right arm, respectively. Let
\begin{itemize}
    \item $\OPT_L^h := \{i\in \OPT_{\text{top}} \mid h_i > \epsilon^2h_t \text{ and } w_i > \epsilon^2 N\} \cup \{i\in \OPT_{\text{bottom}}\mid h_i > \epsilon^2 h_b \text{ and } w_i > \epsilon^2 N\}$,
    \item $\OPT_H^h := \{i\in \OPT_{\text{top}} \mid h_i \le \epsilon^2h_t \text{ and } w_i > \epsilon^2 N\} \cup \{i\in \OPT_{\text{bottom}}\mid h_i \le \epsilon^2 h_b \text{ and } w_i > \epsilon^2 N\}$,
    \item $\OPT_V^h := \{i\in OPT_{\text{top}} \mid h_i > \epsilon^2h_t \text{ and } w_i \le\epsilon^2 N\} \cup \{i\in \OPT_{\text{bottom}}\mid h_i > \epsilon^2 h_b \text{ and } w_i \le \epsilon^2 N\}$, and
    \item $\OPT_S^h := \{i\in \OPT_{\text{top}} \mid h_i \le \epsilon^2h_t \text{ and } w_i \le\epsilon^2N\} \cup \{i\in \OPT_{\text{bottom}}\mid h_i \le \epsilon^2 h_b \text{ and } w_i \le \epsilon^2N\}$.
\end{itemize}
Intuitively, these four sets consist of the items that are large, horizontal, vertical, and small compared to the dimensions of the arm where they lie, respectively. Analogously, the items in $\OPT_{\text{left}}\cup \OPT_{\text{right}}$ are classified into the sets $\OPT_L^v, \OPT_H^v,\OPT_V^v,\OPT_S^v$.

We now present three ways of restructuring the optimal packing and show that the best among the three, together with Lemma \ref{lem:delete-big-item} gives us a better than $3/2$-approximation. The first two of these will be container packings (Lemmas \ref{lem:rotate-vertical} and \ref{lem:klaus}) and the last one is an \lc-packing (Lemma \ref{lem:lc-packing}).

\begin{lemma}
\label{lem:rotate-vertical}
    There exists a container packing of profit at least $(1-\epsilon)(p(\OPT_H^h)+p(\OPT_S^h)+p(\OPT_{\text{left}})+ p(\OPT_{\text{right}}))+p(\hugeitem)$.
\end{lemma}
\begin{proof}
    \begin{figure}
    \centering
    \includegraphics[width=0.8\linewidth]{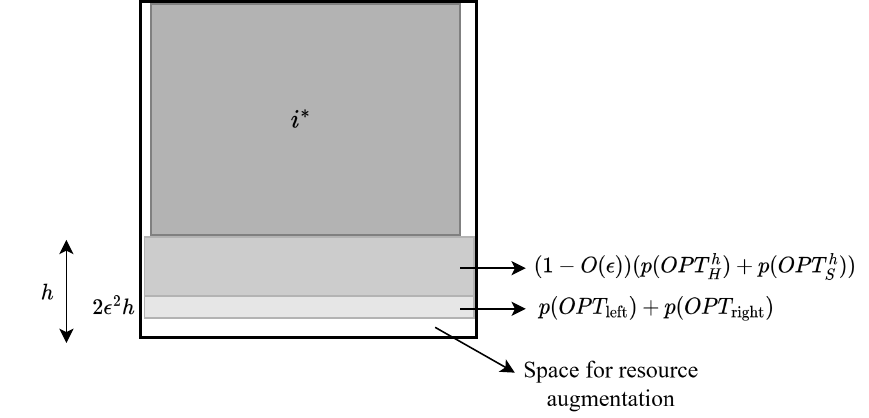}
    \caption{\dk{Packing corresponding to Lemma \ref{lem:rotate-vertical}.}}
    \label{fig:repack-vertical-arms}
\end{figure}
    We temporarily remove the items of $\OPT_{\text{left}}\cup \OPT_{\text{right}}$, and by repositioning the item $\hugeitem$ so that its upper edge touches the top edge of the knapsack, we may assume that the items in the top and bottom arms are packed inside a box of height $h$ and width $N$ below the item $\hugeitem$. We consider a random horizontal strip $\mathcal{S}$ of height $\epsilon h$ inside the box, and delete all items intersecting $\mathcal{S}$. Note that each item in $\OPT_H^h \cup \OPT_S^h$ survives with probability at least $1-O(\epsilon)$, hence we get a profit of at least $(1-O(\epsilon))(p(\OPT_H^h)+p(\OPT_S^h))$. Now since $w\le 2\epsilon^2 h$, we can pack the items of $\OPT_{\text{left}}\cup \OPT_{\text{right}}$ inside the strip $\mathcal{S}$. Finally, we use the remaining empty height of $(\epsilon - 2\epsilon^2)h$ for resource augmentation (Lemma \ref{lem:resource-augmentation}) in order to obtain a container packing with the desired profit guarantee (see Figure \ref{fig:repack-vertical-arms}).
\end{proof}

Next, we show a container packing of all the large and vertical items lying in the four arms. This is achieved by an application of the following inequality from \cite{jansen1999improved}.

\begin{lemma}[\cite{jansen1999improved}]
\label{lem:jansen-porkolab}
        Let $p_1\ge p_2\ge \ldots \ge p_n>0$ be a sequence of real numbers and $P=\sum_{i=1}^{n} p_i$. Let $c$ be a nonnegative integer and $\epsilon >0$. If $n = c^{\Omega(1/\epsilon)}$, then there is an integer $k \le c^{O(1/\epsilon)}$ such that $p_{k+1}+\ldots p_{k+ck} \le \epsilon P$.
\end{lemma}

\begin{lemma}
\label{lem:klaus}
    There exists a container packing of profit at least $(1-\epsilon)(p(\OPT_V^h)+p(\OPT_V^v))+p(\OPT_L^h)+p(\OPT_L^v)+p(\hugeitem)$.
\end{lemma}
\begin{proof}
\begin{figure}
    \centering
    \includegraphics[width=0.3\linewidth]{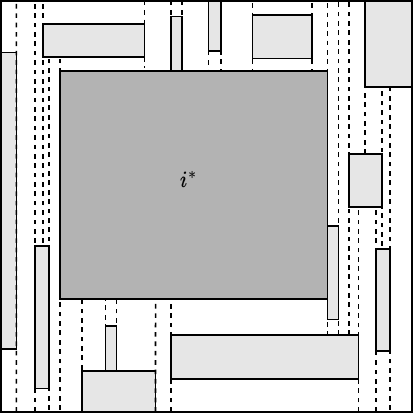}
    \caption{Extending the vertical edges of untouchable items partitions the remaining region of the knapsack into cells}
    \label{fig:klaus-trick}
\end{figure}
    Consider the packing of the items in $\OPT_L^h \cup \OPT_L^v \cup \OPT_V^h \cup \OPT_V^v$ and the item $\hugeitem$ inside the knapsack. Note that $|\OPT_L^h|,|\OPT_L^v| \le 2/\epsilon^4$. Let $k \ge 1/\epsilon^5$ be a constant depending on $\epsilon$ to be determined later. We mark the item $\hugeitem$, all items in $\OPT_L^h \cup \OPT_L^v$ and the $k$ most profitable items in $\OPT_V^h \cup \OPT_V^v$ as \textit{untouchable} (i.e., we create a separate container for each such item). Next, we extend the left and right edges of each untouchable item until it touches the top/bottom edge of another untouchable item or the top/bottom boundary of the knapsack (see Figure \ref{fig:klaus-trick}).
    We discard all items in $\OPT_V^h \cup \OPT_V^v$ that are intersected by these lines. Note that the number of such lines is bounded by $2(k+4/\epsilon^4+1) \le 4k$ and each line intersects at most $2/\epsilon^2$ items of $\OPT_V^h$ and at most $1/\epsilon^2$ items of $\OPT_V^v$. Thus, the number of discarded items is $O(1/\epsilon^2)k$. Next, observe that the region of the knapsack not occupied by the untouchable items is now partitioned into at most $O(1/\epsilon^2)k$ \textit{cells} that contain items from $\OPT_V^h \cup \OPT_V^v$. We consider any such cell $C$ and let $w(C)$ denote the width of $C$. We draw $1/\epsilon-1$ equidistant vertical lines inside $C$ that partition $C$ into strips of width $\epsilon\cdot w(C)$, and discard all items intersecting these lines. The number of such items is bounded by $O(1/\epsilon^2)$. We now discard items in the strip having minimum profit, and use the resulting empty space of width $\epsilon\cdot w(C)$ inside the cell for resource augmentation (Lemma \ref{lem:resource-augmentation}) to get a container packing. Thus, for each cell, we discard at most $O(1/\epsilon^2)$ items and save an $(1-\epsilon)$-fraction of the remaining profit. Hence, the total number of discarded items is bounded by $O(1/\epsilon^4)k$.
    Invoking Lemma \ref{lem:jansen-porkolab} with $c=O(1/\epsilon^4)$ and $p_i$'s being the profits of the items in $\OPT_V^h \cup \OPT_V^v$ in non-increasing order (breaking ties arbitrarily), gives a value of $k$ such that the profit of the discarded items is bounded by $\epsilon\cdot (p(\OPT_V^h)+p(\OPT_V^v))$. Thus, we obtain the desired profit guarantee of the lemma.
\end{proof}

The packings in Lemmas \ref{lem:rotate-vertical} and \ref{lem:klaus} were container packings. As shown in Lemma \ref{lem:1.5-hardness}, it is not possible to beat factor $3/2$ using only container packings. In the following lemma, we exploit the power of \lc-packings in order to pack the horizontal and small items in the vertical arms into $O_{\epsilon}(1)$ containers together with an \lc-packing of a subset of the items packed in the horizontal arms.

\begin{lemma}
\label{lem:lc-packing}
    There exists a \lc-packing of profit at least $(1-\epsilon)(p(\OPT_H^v)+p(\OPT_S^v)) + (22/43-\epsilon)(p(\OPT_{\text{top}})+ p(\OPT_{\text{bottom}}))+p(\hugeitem)$.
\end{lemma}
\begin{proof}
    Consider the packing of the items of $\OPT_H^v \cup \OPT_S^v$ lying in the left arm. Note that by our classification of items, no such item completely lies in the top or bottom arm (otherwise they would have been placed in the set $\OPT_{\text{top}}$ or $\OPT_{\text{bottom}}$). Therefore, all these items can be completely packed into a box of width $w_{\ell}$ and height $h(\hugeitem)+2\epsilon^2N$, since the height of each such item is bounded by $\epsilon^2N$. We draw a random horizontal strip $\mathcal{S}$ of height $\epsilon N$ inside the box and discard all items intersecting $\mathcal{S}$. Clearly, the profit of the discarded items is at most $O(\epsilon)\cdot(p(\OPT_H^v)+p(\OPT_S^v))$. Now we push down each remaining item lying above the strip $\mathcal{S}$ by $\epsilon^2N$ and push up the items lying below $\mathcal{S}$ by $\epsilon^2N$, so that the remaining items now fit inside a box of width $w_{\ell}$ and height at most $h(\hugeitem)$. Observe that there is still an empty horizontal strip of height $(\epsilon - 2\epsilon^2)N$ inside the box, which we utilize for resource augmentation (Lemma \ref{lem:resource-augmentation}) in order to obtain a container packing. Doing an analogous procedure in the right arm, we obtain a container packing of profit at least $(1-O(\epsilon))(p(\OPT_H^v)+p(\OPT_S^v))+p(\hugeitem)$ inside a box $B$ of width $N$ and height $h(\hugeitem)$.

    Next we reposition the items inside the knapsack so that the upper boundary of the box $B$ now touches the top edge of the knapsack, and the items of $\OPT_{\text{top}}\cup \OPT_{\text{bottom}}$ are packed below $B$ inside a box of width $N$ and height $h$. Thus, by applying Lemma \ref{lem:arbitrary-knapsack}, we obtain an \lc-packing inside the knapsack with the desired profit guarantee of the lemma.
\end{proof}

Combining Lemmas \ref{lem:rotate-vertical}, \ref{lem:klaus}, and \ref{lem:lc-packing}, we obtain the following.

\begin{lemma}
\label{lem:keep-massive}
    There exists an \lc-packing of profit at least $(44/87-\epsilon)(p(\OPT_{\text{top}})+p(\OPT_{\text{bottom}})+p(\OPT_{\text{left}})+p(\OPT_{\text{right}}))+p(\hugeitem)$.
\end{lemma}
\begin{proof}
    Let $\OPT_{L\& C^*}$ be the maximum profitable \lc-packing. We have the following bounds.
    \begin{adjustwidth}{-0.5cm}{-1cm}
\begin{alignat*}{2}
22p(\OPT_{L\& C^*}) 
    &\ge (1-\epsilon)\big(22p(\OPT_H^h) + 22p(\OPT_S^h) + 22p(\OPT_{\text{left}}) \\
    &\qquad\qquad\quad  + 22p(\OPT_{\text{right}})\big) + 22p(\hugeitem) &\quad& \text{[Lemma \ref{lem:rotate-vertical}]}\\    
22p(\OPT_{L\& C^*}) 
    &\ge (1-\epsilon)\big(22p(\OPT_V^h) + 22p(\OPT_V^v)\big) + 22p(\OPT_L^h)\\
    &\qquad\qquad\quad + 22p(\OPT_L^v) + 22p(\hugeitem) &\quad& \text{[Lemma \ref{lem:klaus}]} \\    
43p(\OPT_{L\& C^*}) 
    &\ge (1-O(\epsilon))\big(43p(\OPT_H^v) + 43p(\OPT_S^v) + 22p(\OPT_{\text{top}})\\
    &\qquad\qquad\quad + 22p(\OPT_{\text{bottom}})\big) + 43p(\hugeitem) &\quad& \text{[Lemma \ref{lem:lc-packing}]}
\end{alignat*}
\end{adjustwidth}
    Since $\OPT_{\text{top}}\cup \OPT_{\text{bottom}} = \OPT_L^h \cup \OPT_H^h \cup \OPT_V^h \cup \OPT_S^h$ and $\OPT_{\text{left}}\cup \OPT_{\text{right}} = \OPT_L^v \cup \OPT_H^v \cup \OPT_V^v \cup \OPT_S^v$, adding the above inequalities gives
    \[ 87p(\OPT_{L\& C^*})\ge (1-O(\epsilon))(44p(\OPT_{\text{top}})+44p(\OPT_{\text{bottom}})+44p(\OPT_{\text{left}})+44p(\OPT_{\text{right}}))+87p(\hugeitem),\]
    which completes the proof.
\end{proof}

Finally, combining Lemma \ref{lem:delete-big-item} together with Lemma \ref{lem:keep-massive} proves \Cref{lem:LC-packing-hugeitem}.

\LCpackinghugeitem*
\begin{proof}
    Again letting $\OPT_{L\& C^*}$ denote the maximum profitable \lc-packing, from Lemma \ref{lem:keep-massive}, we have $p(\OPT_{L\& C^*})\ge (44/87-\epsilon)(p(\OPT)-p(\hugeitem))+p(\hugeitem)$, and therefore
    \[ 87p(\OPT_{L\& C^*})\ge (44-87\epsilon)p(\OPT)+(43+87\epsilon)p(\hugeitem)\]
    Also, from Lemma \ref{lem:delete-big-item}, we have
    \[43p(\OPT_{L\& C^*})\ge (1-\epsilon)(43p(\OPT)-43p(\hugeitem))\]
    Adding the above two inequalities, we get $130p(\OPT_{L\& C^*})\ge (87-O(\epsilon))p(\OPT)$, thus completing the proof.
\end{proof}

It remains to prove \Cref{lem:arbitrary-knapsack}.

\subsubsection{Proof of \Cref{lem:arbitrary-knapsack}}
\label{sec:arbitrary-knapsack-proof}

We apply the corridor decomposition framework (Lemma \ref{lem:corr-decomposition}) to the packing inside $\mathcal{R}$, thus obtaining a partition of $\RR$ into $O_{\epsilon, \epsilon_{\text{large}}}(1)$ corridors (there is a subtle technicality here -- since Lemma \ref{lem:corr-decomposition} was defined for a square knapsack, we cannot directly apply it to $\RR$; so we first scale $\RR$ together with the packing inside it so that $h(\RR)=N$, apply Lemma \ref{lem:corr-decomposition} to this scaled packing, and then rescale $\RR$ back to its original size). Note that the height of any horizontal subcorridor is bounded by $\epsilon_{\text{large}}\cdot h(\mathcal{R})$ and the width of any vertical subcorridor is at most $\epsilon_{\text{large}}N$. Similar to \Cref{sec:corr-decomposition-appendix}, the length of a horizontal (resp. vertical) subcorridor is defined as the length of the shorter horizontal (resp. vertical) edge bounding the subcorridor. A horizontal (resp. vertical) subcorridor $C$ is said to be long if its length exceeds $N/2$ (resp. $\frac{1}{2}h(\RR)$). We define $I_{LT}$ to be the thin items lying in long subcorridors, and let $I_{ST}$ be the remaining thin items in short subcorridors. Analogously, we classify the fat items into $I_{LF}$ and $I_{SF}$. Let $\OPT_{L\&C^*}$ denote the maximum profitable \lc-packing. Note that the guarantees of Lemmas \ref{lem:lcpacking-focs} and \ref{lem:new-processing} on $p(\OPT_{L\& C^*})$ continue to hold.

    \begin{figure}
    \centering
    \includegraphics[width=0.6\linewidth]{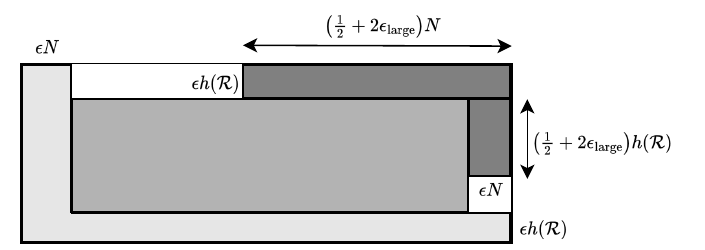}
    \caption{The light gray region corresponds to the boundary $\Ls$. Items of $I_{LT_{\text{short}}}\cup I_{ST}$ are packed in the dark gray boxes. The gray box packs items of $I_{SF}$.}
    \label{fig:L-packing}
\end{figure}

Recall that the packings in Lemmas \ref{lem:lcpacking-focs} and \ref{lem:new-processing} are in fact container packings, i.e., they correspond to the case of a degenerate \lcor. Let $I_{LT_{\text{long}}}\subseteq I_{LT}$ be the set of horizontal items with width more than $(\frac{1}{2}+2\epslarge)N$ and vertical items of height more than $(\frac{1}{2}+2\epslarge)h(\mathcal{R})$. Let $I_{LT_{\text{short}}} = I_{LT}\setminus I_{LT_{\text{long}}}$.
    \begin{lemma}
    \label{lem:pack-st}
        There exists a container packing of a subset of items in $I_{LT_{\text{short}}}\cup I_{ST}$ of profit at least $(1-\epsilon)p(I_{LT_{\text{short}}}\cup I_{ST})$ into two boxes of dimensions $(\frac{1}{2}+2\epslarge)N\times \epsilon h(\RR)$ and $\epsilon N \times (\frac{1}{2}+2\epslarge)h(\RR)$.
    \end{lemma}
    \begin{proof}
        Recall that the items of $I_{LT}\cup I_{ST}$ were obtained by processing the corridors with parameter $\epsilon_{\text{box}}$ (i.e., the items lying inside a region whose height was only $\epsilon_{\text{box}}$-fraction of the subcorridor were marked as thin). 
        Therefore, the total height of the horizontal items in $I_{LT_{\text{short}}}\cup I_{ST}$ is at most $\epsilon_{\text{box}}h(\RR)\cdot c_{\epsilon,\epsilon_{\text{large}}} \le \epsilon^4\cdot h(\RR)$, and they can be packed into a box of width $(\frac{1}{2}+2\epslarge)N$ and height $\epsilon^4\cdot h(\RR)$. We increase the height of the box to $\epsilon h(\RR)$ and utilize the extra height of $(\epsilon-\epsilon^4)\cdot h(\RR)$ in order to obtain a container packing using Lemma \ref{lem:resource-augmentation}. In an analogous way, the total width of the vertical items in $I_{LT_{\text{short}}}\cup I_{ST}$ can be bounded by $\epsilon^4 N$, and there exists a container packing of (a subset of) these items into a box of height $(\frac{1}{2}+2\epslarge)h(\mathcal{R})$ and width $\epsilon N$.
    \end{proof}

    \begin{lemma}
    \label{lem:pack-sf}
        There exists a container packing of a subset of items in $I_{SF}$ of profit at least $(1/2-O(\epsilon))p(I_{SF})$ inside a box of dimensions $(1-2\epsilon)N\times (1-2\epsilon)h(\RR)$.
    \end{lemma}
    \begin{proof}
        \dk{Consider the packing of the items of $I_{SF}$ inside the box $\RR$. Let $\SM^h$ be a random horizontal strip of height $2\epsilon\cdot h(\RR)$, and $\SM^v$ be a random vertical strip of width $3\epsilon N$ inside $\RR$. We delete items of $I_{SF}$ intersecting either $\SM^h$ or $\SM^v$. Let $I'_{SF}$ be the items that are not deleted. Each horizontal item in $I_{SF}$ is intersected by $\SM^h$ with probability at most $O(\epsilon)$ and by $\SM^v$ with probability at most $1/2+O(\epsilon)$. Therefore, it survives with probability at least $1/2-O(\epsilon)$. In a similar way, each vertical item in $I_{SF}$ survives with probability at least $1/2-O(\epsilon)$. Hence the expected profit of $I'_{SF}$ is at least $(1/2-O(\epsilon))p(I_{SF})$, and $I'_{SF}$ can be packed inside a $(1-3\epsilon)N\times (1-2\epsilon)h(\RR)$ box. By increasing the width of this box by $\epsilon N$, we can obtain a container packing of (a subset of) $I'_{SF}$ using Lemma \ref{lem:resource-augmentation}.}
    \end{proof}

    \dk{We place the boundary \lcor-corridor aligned with the left and bottom boundary of $\RR$, with the vertical arm having a width of $\epsilon N$, and the horizontal arm having a height of $\epsilon h(\RR)$, as shown in Figure \ref{fig:L-packing}. Next, we place the two boxes that pack items from $I_{LT_{\text{short}}}\cup I_{ST}$ given by Lemma \ref{lem:pack-st}, aligned with the top and right boundary of the knapsack. This leaves an empty rectangular region of width $(1-2\epsilon)N$ and height $(1-2\epsilon)h(\mathcal{R})$ at the interior of the knapsack where we place the box that packs items of $I_{SF}$ (Lemma \ref{lem:pack-sf}).}

    \dk{We next pack items from $\ILTlong$ into the boundary \lcor-corridor. Consider the packing of $\ILTlong$ in the optimal solution. For each horizontal item $i\in \ILTlong$ that has no vertical item to its top (resp. bottom), we shift $i$ up (resp. down) until it touches another horizontal item or the top (resp. bottom) boundary of $\RR$. We iterate this process as long it is possible to move some horizontal item. We then perform a symmetric process for the vertical items of $\ILTlong$. At the end, the items of $\ILTlong$ are packed in four \textit{stacks} at the boundaries of the knapsack. Recall that we need to ensure the condition that either $h(i)>w(i)$ or $w(i)\ge h(i)$ holds for all items $i$ packed in the vertical arm of the \lcor.}

    To this end, we let $I_t, I_b, I_{\ell}, I_r$ be the items of $\ILTlong$ lying in the top, bottom, left and right stacks, respectively. We further classify $I_{\ell}$ into two categories: let $I_{\ell}^{h>w} := \{i \in I_{\ell} \mid h(i) > w(i)\}$, and $I_{\ell}^{w\ge h} := I_{\ell}\setminus I_{\ell}^{h>w}$. Analogously the items in $I_r$ are classified into $I_r^{h>w}$ and $I_r^{w\ge h}$. \dk{Note that since the items in $I_{\ell}\cup I_r$ occupied only an $\epsilon_{\text{box}}$-fraction of the width of a subcorridor, and the number of subcorridors is bounded by $c_{\epsilon,\epslarge}$, the total width of $I_{\ell}\cup I_r$ is at most $\epsilon_{\text{box}}N\cdot c_{\eps,\epslarge} \le \epsilon^4N$ by our choice of $\eps_{\text{box}}$. Therefore, if $I_t \cup I_b$ is discarded, we can pack the items in $I_{\ell}\cup I_r$ one beside the other inside a box of width $\epsilon^4N$. Since the width of the vertical arm of the \lcor~is $\epsilon N$, we can use the remaining width of $(\epsilon-\epsilon^4)N$ in order to obtain a container packing of a subset of $I_{\ell}\cup I_r$ of profit at least $(1-\epsilon)p(I_{\ell}\cup I_r)$ (Lemma \ref{lem:resource-augmentation}), which we can place inside the vertical arm of the \lcor. This gives the following lemma.}

    \begin{lemma}
    \label{lem:all-vertical}
        We have $p(\OPT_{L\& C^*})\ge (1-\epsilon)(p(I_{\ell})+p(I_r)+p(I_{ST}\cup I_{LT_{\text{short}}})+\frac{1}{2}p(I_{SF}))$.
    \end{lemma}

    \begin{lemma}
    \label{lem:delete-vert-arm}
        We have the following lower bounds on $p(\OPT_{L\& C^*})$.
        \begin{itemize}
            \item $p(\OPT_{L\& C^*}) \ge (1-\epsilon)(p(I_t)+p(I_b)+p(I_{j}^{h>w})+p(I_{ST}\cup I_{LT_{\text{short}}})+\frac{1}{2}p(I_{SF}))$,
            \item $p(\OPT_{L\& C^*}) \ge (1-\epsilon)(p(I_t)+p(I_b)+p(I_{j}^{w\ge h})+p(I_{ST}\cup I_{LT_{\text{short}}})+\frac{1}{2}p(I_{SF}))$,
        \end{itemize}
        for $j\in \{\ell,r\}$.
    \end{lemma}
    \begin{proof}
        For any set $J\in \{I_{\ell}^{h>w},I_{\ell}^{w\ge h}, I_r^{h>w}, I_r^{w\ge h}\}$, observe that items in $I_t\cup I_b\cup J$ are packed inside a U-shaped region at the boundary of $\RR$. Using a result of \cite{galvez2021approximating}, they can be rearranged into an \lcor. In the remaining area of $\RR$, we can pack a profit of $(1-\epsilon)(p(I_{ST}\cup I_{LT_{\text{short}}})+\frac{1}{2}p(I_{SF}))$ as discussed before (see Figure \ref{fig:L-packing}).
    \end{proof}

    Combining Lemmas \ref{lem:all-vertical} and \ref{lem:delete-vert-arm}, we obtain the following.

    \begin{lemma}
    \label{lem:lt-st-sf-pack}
        We have $p(\OPT_{L\& C^*})\ge (1-\epsilon)(\frac{4}{7}p(I_{LT})+p(I_{ST})+\frac{1}{2}p(I_{SF}))$.
    \end{lemma}
    \begin{proof}
        Adding the four inequalities (corresponding to $j\in \{\ell,r\}$) in Lemma \ref{lem:delete-vert-arm} with 3 times the inequality in Lemma \ref{lem:all-vertical}, we get
        \begin{align*}
            7p(\OPT_{L\& C^*}) &\ge (1-\epsilon)\left(4p(I_{t})+4p(I_b)+4p(I_{\ell})+4p(I_r)+7p(I_{ST}\cup I_{LT_{\text{short}}})+\frac{7}{2}p(I_{SF})\right) \\
            &\ge (1-\epsilon)\left(4p(I_{LT})+7p(I_{ST})+\frac{7}{2}p(I_{SF})\right),
        \end{align*}
        and we are done.       
    \end{proof}

    Combining Lemmas \ref{lem:lcpacking-focs}, \ref{lem:new-processing} and \ref{lem:lt-st-sf-pack}, we prove Lemma \ref{lem:arbitrary-knapsack}.

\arbitraryknapsack*
\begin{proof}
        We have
        \begin{alignat*}{2}
            2p(\OPT_{L\&C^*}) &\ge (1-\epsilon)(2p(I_{LF})+p(I_{SF})+p(I_{ST})) &\quad& \text{[Lemma \ref{lem:lcpacking-focs}(i)]} \\
            20p(\OPT_{L\& C^*}) &\ge (1-\epsilon)\left(20p(I_{LF})+20p(I_{SF})+10p(I_{LT})\right) &\quad& \text{[Lemma \ref{lem:new-processing}]} \\
            21p(\OPT_{L\& C^*}) &\ge (1-\epsilon)\left(12p(I_{LT})+21p(I_{ST})+\frac{21}{2}p(I_{SF})\right) &\quad& \text{[Lemma \ref{lem:lt-st-sf-pack}]}
        \end{alignat*}
        Adding the above inequalities, we obtain
        \begin{align*}
            43p(\OPT_{L\& C^*}) &\ge (1-\epsilon)\left(22p(I_{LF})+ \frac{63}{2}p(I_{SF})+22p(I_{LT})+22p(I_{ST})\right) \\
            &\ge (22-22\epsilon)p(I),
        \end{align*}
        which completes the proof.
    \end{proof}

\paragraph{Handling small items.}
\dk{Similar to \Cref{sec:non-rotation-improved}, we have to repack a set of items $\OPT'_{\text{small}}$ consisting of the small items that did not completely lie inside the corridors, and the small items that were discarded while processing the corridors. For the packings corresponding to Lemmas \ref{lem:lcpacking-focs}(i) and \ref{lem:new-processing} (which are container packings), we pack the small items as described in \Cref{sec:non-rotation-improved}. It only remains to handle the packings corresponding to Lemma \ref{lem:lt-st-sf-pack}. Recall that for packing items of $I_{SF}$ in Lemma \ref{lem:pack-sf}, we discarded items of $I_{SF}$ intersecting a random horizontal strip of height $2\eps\cdot h(\RR)$ and a random vertical strip of width $3\eps N$. Then the remaining items of $I_{SF}$ could be packed inside a box of dimensions $(1-3\eps)N\times (1-2\eps)h(\RR)$. We can increase the width of this box by $\eps N$ and use a width of $\eps N/2$ for resource augmentation. This still leaves an empty region $\SM$ of width $\eps N/2$ and height $(1-2\eps)h(\RR)$. Since $a(\OPT'_{\text{small}})\le \eps^3N\cdot h(\RR)$ and each item in $\OPT'_{\text{small}}$ has width at most $\epssmall N$ and height at most $\epssmall h(\RR)$, they can be easily packed using NFDH inside the region $\SM$.}

\section{Hardness}
\label{sec:finehard}
In this section, we prove Theorem \ref{theorem:hardness-of-2dkr}. Given a collection of $n$ integers, the $k$-\textsc{Sum} problem asks to determine whether there are $k$ among them that sum to $0$. The $k$-\textsc{Sum} conjecture states the following.

\begin{conjecture}[$k$-\textsc{Sum} conjecture, \cite{abboud2013exact}]
    There does not exist a $k\ge 2$, a $\delta >0$, and a randomized algorithm that succeeds (with high probability) in solving $k$-\textsc{Sum} in time $O(n^{\lceil k/2\rceil-\delta})$.
\end{conjecture}

We reduce a variant of the $k$-\textsc{Sum} problem to 2DK and 2DKR. We denote this variant by $k$-\textsc{PartSum};  it is defined as follows. In $k$-\textsc{PartSum} we are given a multiset $\AM$ of $n$ positive integers (in particular, they are not necessarily pairwise different) and an integer $k$.
We want to find a (multi)-subset $\overline{T} \subseteq A$ with $|\overline{T}|=k$ such that it can be split into two (multi)-subsets $T_1$ and $T_2$ such that the sum of the numbers in both multi-subsets is equal. 
First we show that $k$-\textsc{PartSum} cannot be solved in $O(n^{\frac{k-1}{2}-\delta})$ time, for any odd $k\ge 3$ and $\delta >0$ assuming the $k$-\textsc{Sum} conjecture.

\begin{lemma}
 Assuming the $k$-\textsc{Sum} conjecture, for any odd $k\ge 3$ and $\delta>0$, there does not exist an algorithm for $k$-\textsc{PartSum} running in time $O(n^{\frac{k-1}{2}-\delta})$.    
\end{lemma}
\begin{proof}
    Let $k$ be an odd integer and suppose we are given an instance of $(k-1)$-\textsc{Sum} whose input numbers are the multiset 
    $\AM=\{a_1,\ldots,a_n\}$. We consider the $k$-\textsc{PartSum} instance $\overline{\AM}=\{a_1+M',\ldots, a_n+M'\}\cup \{(k-1)M'\}$, where $M'=k\cdot (\max_{a\in \AM} |a|)+1$. Note that all values in $\overline{\AM}$ are strictly positive. We show that $k-1$ numbers in $\AM$ sum to $0$ if and only if there exist disjoint subsets $T_1$ and $T_2$ of $\overline{\AM}$ of equal sum and such that $|T_1|+|T_2|=k$.

    Suppose there exists a (multi-)subset $T=\{a_{i_1}, \ldots,a_{i_{k-1}}\}\subset \AM$ such that $a_{i_1}+\ldots + a_{i_{k-1}}=0$. Then we define  $T_1=\{a_{i_1}+M',\ldots, a_{i_{k-1}}+M'\}$ and $T_2 = \{(k-1)M'\}$, and we are done. 
    Conversely, suppose there exists $\overline{T} \subseteq \overline{\AM}$ with $|\overline{T}|=k$ such that it can be split into two subsets $T_1$ and $T_2$ of equal sum.
    W.l.o.g.~let $|T_2|\ge |T_1|$, and since $k$ is odd, we have that $|T_2|-|T_1|\ge 1$. Now if $(k-1)M' \not\in T_1\cup T_2$, we can find $a_{i_1}, \ldots, a_{i_{|T_1|}}$ and $a_{j_1},\ldots, a_{j_{|T_2|}}$ such that $a_{i_1}+\ldots + a_{i_{|T_1|}}+|T_1|M' = a_{j_1}+\ldots + a_{j_{|T_2|}}+|T_2|M'$, and therefore $(a_{i_1}+\ldots + a_{i_{|T_1|}})-(a_{j_1}+\ldots + a_{j_{|T_2|}}) = (|T_2|-|T_1|)M' \ge M'$.
This is a contradiction since the absolute value of the LHS in the above equality is at most $k\cdot (\max_{a\in \AM} |a|)<M'$. Therefore $(k-1)M'\in T_1\cup T_2$. Suppose that $(k-1)M' \in T_1$.  If $T_2$ consists of at most $k-2$ elements, they sum to at most $(k-2)\cdot (\max_{a\in \AM} |a|) + (k-2)M' < (k-1)M'$ \aw{which is} a contradiction. Hence we must have $|T_2|=k-1$ and $T_1= \{(k-1)M'\}$. \aw{Assuming that} $T_2=\{a_{i_1}+M',\ldots, a_{i_{k-1}}+M'\}$, \aw{we have that} $a_{i_1}+\ldots + a_{i_{k-1}}=0$. Finally, if $(k-1)M' \in T_2$, by a similar argument, we would have that $|T_1|=k-1$ and $T_2 =\{(k-1)M'\}$, which would contradict our assumption that $|T_2|\ge |T_1|$. 

    Since our reduction takes $O(n)$ time only, the existence of an $O(n^{\frac{k-1}{2}-\delta})$-time algorithm for $k$-\textsc{PartSum} would imply an algorithm with the same running time bound for $(k-1)$-\textsc{Sum}, contradicting the $k$-\textsc{Sum} conjecture.
\end{proof}

For our reduction, we assume that $k$ is a sufficiently large odd integer ($k\ge 9$ suffices; \aw{note that if $k=O(1)$ we can simply solve the instance in time $n^{O(1)}$ to optimality within our reduction and then map it to trivial yes- or no-instance}) and let $\AM$ be an instance of $k$-\textsc{PartSum}.
We create an instance of 2DK or 2DKR as follows. Let $M:= \max_{a\in \AM} a$.
We consider a square knapsack of side length $N:= 2Mk^4$. For each $a\in \AM$, we create two rectangles $R_a$ and $R'_a$ with $w(R_a)=\frac{N}{k}+a$ and $h(R_a)=\frac{N}{2}-a$, and $w(R'_a)=\frac{N}{k}-a$ and $h(R'_a)=\frac{N}{2}+a$, respectively.
We assign a profit of $1$ to each rectangle \aw{which yields an instance in the cardinality case}. Let $X:= \{R_a\}_{a\in \AM}$ and $X':= \{R'_a\}_{a\in \AM}$. We first make a few simple observations about the dimensions of these rectangles.

\begin{lemma}
\label{lem:rectangle-properties}
The following statements hold for each $i\in X\cup X'$.
\begin{enumerate}[label=(\roman*)]
    \item $w(i)+h(i)=(\frac{1}{2}+\frac{1}{k})N$.
    \item $w(i) \in ((\frac{1}{k}-\frac{1}{k^4})N,(\frac{1}{k}+\frac{1}{k^4})N)$ and $h(i) \in ((\frac{1}{2}-\frac{1}{k^4})N,(\frac{1}{2}+\frac{1}{k^4})N)$.
    \item The area of $i$ \aw{is contained in the interval} $((\frac{1}{2k}-\frac{1}{k^4})N^2,(\frac{1}{2k}+\frac{1}{k^4})N^2)$.
\end{enumerate}
\end{lemma}

\begin{corollary}
\label{cor:bound-on-opt}
    Any feasible packing contains at most $2k$ rectangles.
\end{corollary}
\begin{proof}
    Follows directly from Lemma \ref{lem:rectangle-properties}(iii), since the area of $2k+1$ rectangles would exceed $(2k+1)(\frac{1}{2k}-\frac{1}{k^4})N^2>N^2$ \aw{since $k\ge 3$.}
\end{proof}

\aw{Let $\OPT$ denote an optimal packing for the constructed instance of 2DK or 2DKR, respectively.}
We first establish \aw{that yes-instances of $k$-\textsc{PartSum} are mapped to instances of 2DK or 2DKR for which $|\OPT|=2k$.}

\begin{lemma}
    If there exist values $a_1,a_2,\ldots, a_k \in \AM$ such that $\sum_{j=1}^{m} a_j=\sum_{j={m+1}}^{k} a_j$ holds for some index $m\in [k-1]$, then $|\OPT|=2k$ \aw{and there is an optimal solution that does not rotate any item}.
\end{lemma}
\begin{proof}
    \begin{figure}
    \centering
    \includegraphics[width=0.3\linewidth]{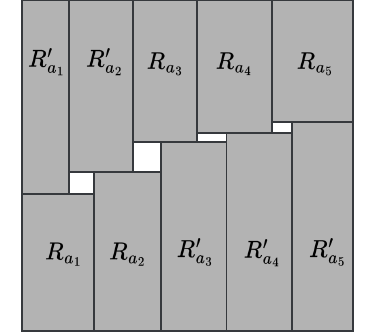}
    \caption{\dk{Optimal packing for $k=5$ and $m=2$.}}
    \label{fig:hard-instance}
\end{figure}
    We assume w.l.o.g.~that $a_1\ge \ldots \ge a_m$ and $a_{m+1} \le \ldots \le a_k$. We construct a feasible packing of the rectangles $\{R_{a_j},R'_{a_j}\}_{j\in [k]}$.
    Observe that $\sum_{j=1}^{m} w(R'_{a_j})+\sum_{j=m+1}^{k} w(R_{a_j})=N$ and similarly $\sum_{j=1}^{m} w(R_{a_j})+\sum_{j=m+1}^{k} w(R'_{a_j})=N$. Also $h(R_{a_j})+h(R'_{a_j})=N$, for all $j\in [k]$. For any rectangle $i$, we shall let $\textrm{left}(i)$ (resp. $\textrm{right}(i)$) denote the $x$-coordinate of the left (resp. right) edge of $i$ (assuming the bottom left corner of the knapsack lies at the origin). We place the rectangles $R'_{a_1},\ldots, R'_{a_m}$ followed by $R_{a_{m+1}},\ldots, R_{a_k}$ from left to right in this order, \aw{pushed to the left as much as possible,} so that the top edge of each rectangle touches the top boundary of the knapsack\aw{, see Figure \ref{fig:hard-instance}}. Note that $h(R'_{a_1})\ge h(R'_{a_2})\ge \ldots \ge h(R'_{a_m})\ge h(R_{a_{m+1}})\ge h(R_{a_{m+2}})\ge \ldots \ge h(R_{a_k})$, and therefore the bottom edges of these rectangles form a monotonically increasing curve from the left to the right knapsack boundary. Next, starting from the left boundary of the knapsack we place the rectangles $R_{a_1},\ldots,R_{a_{m-1}}$ from left to right in this order with the bottom edge of each \aw{rectangle} touching the bottom boundary of the knapsack. Observe that this is feasible since $w(R_{a_j})>w(R'_{a_j})$ holds for all $j\in [k]$, and therefore, \aw{for each $j\in [k]$ we can argue that} after the rectangles $R_{a_1},\ldots, R_{a_{j-1}}$ have been placed, we will have $\textrm{right}(R_{a_{j-1}})>\textrm{right}(R'_{a_{j-1}})$, and thus the available vertical space at the $x$-coordinate $\textrm{right}(R_{a_{j-1}})$ is at least $N-h(R'_{a_j})=h(R_{a_j})$, implying that $R_{a_j}$ can be packed.  Next, starting from the right knapsack boundary, we place the rectangles $R'_{a_k},R'_{a_{k-1}}\ldots, R'_{a_{m+1}}$ from right to left in this order, such that the bottom edge of each rectangle touches the bottom boundary of the knapsack. Again since $w(R_{a_j})>w(R'_{a_j})$ holds, the packing is feasible by a similar argument as before. 
    It remains to pack $R_{a_m}$. We make the following two observations about the packing constructed \aw{until} now. 
    \begin{itemize}
        \item $\textrm{right}(R_{a_{m-1}}) > \textrm{left}(R'_{a_m})$; \aw{this holds} because $\textrm{right}(R_{a_{m-1}})=\sum_{j=1}^{m-1} w(R_{a_j})>\sum_{j=1}^{m-1} w(R'_{a_j})=\textrm{left}(R'_{a_m}).$
        \item $\textrm{left}(R'_{a_{m+1}})>\textrm{right}(R'_{a_m})$; \aw{this holds} because $\textrm{left}(R'_{a_{m+1}})=N-\sum_{j=m+1}^{k} w(R'_{a_j}) > N- \sum_{j=m+1}^{k} w(R_{a_j}) = \textrm{right}(R'_{a_m})$, \aw{using that $\sum_{j=1}^{m} w(R'_{a_j})+\sum_{j=m+1}^{k} w(R_{a_j})=N$}.
    \end{itemize}
    The \aw{two properties} above together imply that the rectangle $R_{a_m}$ can be packed touching the bottom knapsack boundary, with the left edge of $R_{a_m}$ lying at the $x$-coordinate $\textrm{right}(R_{a_{m-1}})$. Letting $\textrm{top}(i), \textrm{bottom}(i)$ denote the $y$-coordinates of the top and bottom edges of \aw{each} rectangle $i$, observe that $\textrm{top}(R_{a_m})=h(R_{a_m})=N-h(R'_{a_m})=\textrm{bottom}(R'_{a_m})$, and therefore the top edge of $R_{a_m}$ touches the bottom edge of $R'_{a_m}$. Also for any \aw{$j\in \{m+1,\ldots, k\}$}, we have $\textrm{bottom}(R_{a_j})=N-h(R_{a_j})\ge N-h(R'_{a_m})= \textrm{bottom}(R'_{a_m})$ and thus $R_{a_m}$ does not intersect with the rectangle $R_{a_j}$.  Therefore, the rectangles $\{R_{a_j},R'_{a_j}\}_{j\in [k]}$ can be packed inside the knapsack and we have $|\OPT|=2k$.
\end{proof}

Next, our goal is to \aw{prove the following lemma}, which would complete the reduction.

\begin{lemma}
\label{lem:only-if-direction}
    If $|\OPT|=2k$, then there exist values $a_1,a_2,\ldots, a_k \in \AM$ and index $m\in [k-1]$ such that $\sum_{j=1}^{m} a_j=\sum_{j={m+1}}^{k} a_j$.
\end{lemma}

\aw{Assume that $|\OPT|=2k$.} As we shall see, the main difficulty in analyzing 2DKR arises from item rotations; the argumentation for 2DK is much simpler. We shall first establish that if $|\OPT|=2k$, then the optimal packing for 2DKR either rotates all items or none.
For an item $i$, let $w^*(i)$ and $h^*(i)$ denote the width and height of $i$ in the packing of $\OPT$. We say that rectangle $i$ is \textit{oriented vertically} if $h^*(i)=h(i)$ and $w^*(i)=w(i)$ holds \aw{(note that then $h^*(i)>w^*(i)$)},
and \textit{oriented horizontally} otherwise.

\begin{lemma}
\label{lem:orientation}
    $\OPT$ does not contain both a horizontally oriented and a vertically oriented rectangle.
\end{lemma}

In order to prove Lemma \ref{lem:orientation}, we assume for the sake of contradiction that 
$\OPT$ contains both a horizontally oriented and a vertically oriented rectangle. We replace each rectangle by a smaller rectangle of dimensions $(\frac{1}{k}-\frac{1}{k^4})N \times (\frac{1}{2}-\frac{1}{k^4})N$, that completely lies inside the original rectangle (note that this is always possible due to Lemma \ref{lem:rectangle-properties}(ii)). Let $\optshr$ denote the packing of these $2k$ identical rectangles (with different orientations) inside the knapsack. We have the following guarantee on the free area, i.e., the area not occupied by any rectangle, inside the knapsack.
    \begin{lemma}
    \label{lem:free-area}
        The area not occupied by the items in $\OPT_{\mathrm{shrink}}$ inside the knapsack is at most $\frac{2N^2}{k^3}$.
    \end{lemma}
    \begin{proof}
        The total area of the $2k$ (shrinked) rectangles is $2k(\frac{1}{k}-\frac{1}{k^4})(\frac{1}{2}-\frac{1}{k^4})N^2 > 2k(\frac{1}{2k}-\frac{1}{k^4})N^2=(1-\frac{2}{k^3})N^2$. Therefore, the free area is at most $\frac{2N^2}{k^3}$.
    \end{proof}

    Next, we show the following lemma, which will be useful in our argumentation later.

    \begin{lemma}
    \label{lem:bound-intersected-rectangles}
        Assume that a horizontal (resp. vertical) line $\ell$ intersects both a horizontally oriented and a vertically oriented rectangle. Then $\ell$ does not intersect any other horizontally (resp. vertically) oriented rectangle. Also, the total number of vertically (resp. horizontally) oriented rectangles intersected by $\ell$ is bounded by $\frac{k-1}{2}$.
    \end{lemma}
    \begin{proof}
        Consider a horizontal line $\ell$ intersecting both a horizontally oriented and a vertically oriented rectangle. If $\ell$ intersects another horizontally oriented rectangle, from Lemma \ref{lem:rectangle-properties}(ii), the width of the knapsack would be at least $2(\frac{1}{2}-\frac{1}{k^4})N+(\frac{1}{k}-\frac{1}{k^4})N>N$, which is a contradiction. If $\ell$ intersects at least $\frac{k+1}{2}$ vertically oriented rectangles, the total width of the horizontally oriented rectangle and $\frac{k+1}{2}$ vertically oriented rectangles would exceed $(\frac{1}{2}-\frac{1}{k^4})N+\frac{k+1}{2}(\frac{1}{k}-\frac{1}{k^4})N>N$, where the inequality follows since $k\ge 9$. This gives a contradiction again. Therefore, $\ell$ intersects at most $\frac{k-1}{2}$ vertically oriented rectangles, and we are done.
    \end{proof}

    Consider the packing of $\optshr$. We push each rectangle down as much as possible, i.e., apply gravity to the packing of $\optshr$. We show the following lemma.
    
    \begin{lemma}
    \label{lem:structural-property-of-line}
        If a horizontal line $\ell$ contains the bottom edge of a vertically oriented rectangle and does not intersect any horizontally oriented rectangle, then the region $[0,N]\times [y_{\ell},y_{\ell}+(\frac{1}{2}-\frac{1}{k^4})N]$ contains exactly $k$ vertically oriented rectangles placed one beside the other, and does not contain any horizontally oriented rectangle, where $y_{\ell}$ denotes the $y$-coordinate of the line $\ell$.
    \end{lemma}

    \begin{proof}
        Let $i^*$ be a vertically oriented rectangle whose bottom edge is contained in $\ell$. Suppose that there is a horizontally oriented rectangle $i$ whose bottom edge is also contained in $\ell$. Let $\overline{\ell}$ be the horizontal line $y=y_{\ell}+\frac{N}{k^4}$. Then $\overline{\ell}$ intersects both $i$ and the vertically oriented rectangle $i^*$, and therefore by Lemma \ref{lem:bound-intersected-rectangles}, $\overline{\ell}$ can intersect at most $\frac{k-1}{2}$ vertically oriented rectangles. Hence, the width of the portion of $\overline\ell$ intersected by rectangles is at most $(\frac{1}{2}-\frac{1}{k^4})N+(\frac{k-1}{2})(\frac{1}{k}-\frac{1}{k^4})N< (1-\frac{1}{2k})N$. Hence, the total width of the non-intersected portions on the line $\overline\ell$ is at least $\frac{N}{2k}$. \dk{Observe that above each such portion, a height of $(\frac{1}{k}-\frac{1}{k^4})N-\frac{N}{k^4} = (\frac{1}{k}-\frac{2}{k^4})N$ does not intersect any other rectangle. }
        The free area inside the knapsack is then at least $\frac{N}{2k}(\frac{1}{k}-\frac{2}{k^4})N>\frac{N^2}{4k^2}\ge \frac{2N^2}{k^3}$, where the second inequality holds since $k\ge 9$. This contradicts Lemma \ref{lem:free-area}. Hence, no horizontally oriented item has its bottom edge contained in $\ell$.

        \dk{Since the items of $\optshr$ were pushed down as much as possible, there can be no horizontally oriented rectangle intersecting the region $[0,N]\times [y_{\ell},y_{\ell}+(\frac{1}{2}-\frac{1}{k^4})N]$ (otherwise it would have been pushed down so that its bottom edge touches $\ell$).}        
        Now if the number of vertically oriented rectangles whose bottom edge is contained in $\ell$ is less than $k$, we would obtain a free width of at least $N-(k-1)(\frac{1}{k}-\frac{1}{k^4})N>\frac{N}{k}$ on $\overline\ell$ and therefore a free area of at least $\frac{N}{k}(\frac{1}{2}-\frac{2}{k^4})N > \frac{N^2}{4k}$, again contradicting Lemma \ref{lem:free-area}. Hence, we must have exactly $k$ vertically oriented rectangles lying on $\overline\ell$. In particular, since we pushed down all rectangles in $\optshr$ as much as possible, this implies that the region  $[0,N]\times [y_{\ell},y_{\ell}+(\frac{1}{2}-\frac{1}{k^4})N]$ does not contain any horizontally oriented rectangles at all. This completes the proof of the lemma.
    \end{proof}

    We are now ready to prove Lemma \ref{lem:orientation}.

    \begin{figure}
    \centering
    \includegraphics[width=0.4\linewidth]{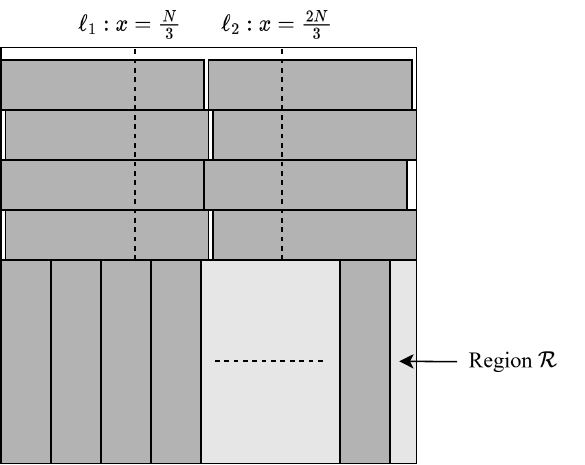}
    \caption{Figure for Lemma \ref{lem:orientation}.}
    \label{fig:single-orientation}
\end{figure}

    \begin{proof}[Proof of Lemma \ref{lem:orientation}]
        Let $i^*$ be a vertically oriented rectangle with the minimum $y$-coordinate of the bottom edge, and let $y_{\min}$ be the $y$-coordinate of the bottom edge of $i^*$. Let $\RR_1:=[0,N]\times [0,y_{\min}]$. Since all rectangles below the line $y=y_{\min}$ are horizontally aligned, each having a height of $(\frac{1}{k}-\frac{1}{k^4})N$, it follows that $y_{\min}$ is a multiple of $(\frac{1}{k}-\frac{1}{k^4})N$ and in particular, no horizontally oriented rectangle is intersected by $y=y_{\min}$. Then by Lemma \ref{lem:structural-property-of-line}, the region $\RR_2:= [0,N]\times [y_{\min},y_{\min}+(\frac{1}{2}-\frac{1}{k^4})N]$ contains exactly $k$ vertically oriented rectangles. We swap the regions $\RR_1$ and $\RR_2$ inside $K$, so that now the region $\RR:= [0,N]\times [0,(\frac{1}{2}-\frac{1}{k^4})N]$ consists of $k$ vertically oriented rectangles. We assume w.l.o.g.~that these $k$ rectangles inside $\RR$ are pushed as much to the left as possible. Consider the vertical lines $\ell_1\colon x=N/3$ and $\ell_2 \colon x=2N/3$. Let $n_1$ and $n_2$ be the number of horizontally oriented rectangles intersected by $\ell_1$ and $\ell_2$, respectively. Notice that any horizontally oriented rectangle must intersect at least one of $\ell_1$ or $\ell_2$. Since both $\ell_1$ and $\ell_2$ intersect some vertically oriented rectangle in $\RR$
        (as the total width of the rectangles inside $\RR$ is $k(\frac{1}{k}-\frac{1}{k^4})N>\frac{2N}{3}$
        and neither $N/3$ nor $2N/3$ is an integral multiple of the width of the vertical rectangles), they can intersect at most $\frac{k-1}{2}$ horizontally oriented rectangles each due to Lemma \ref{lem:bound-intersected-rectangles}, and thus both $n_1$ and $n_2$ are bounded from above by $\frac{k-1}{2}$. Hence, the number of horizontally oriented rectangles in the region $K\setminus \RR$ is at most $k-1$, and so there must be at least one vertically oriented rectangle $R^v$ in $K\setminus \RR$. Let
        \aw{$\ell$ be the horizontal line with $y$-coordinate
        $y=(\frac{1}{2}-\frac{1}{k^4})N$, i.e.,} the horizontal line passing through the top boundary of $\RR$. We push down the rectangles in the region $K\setminus \RR$ as much as possible. Observe that now if the bottom edge of $R^v$ is not contained in $\ell$, then there must exist a horizontally oriented rectangle $R^h$ in $K\setminus \RR$ such that the bottom edge of $R^v$ touches the top edge of $R^h$. Then we can find a vertical line that intersects the horizontally oriented rectangle $R^h$, the vertically oriented rectangle $R^v$, and a vertically oriented rectangle lying in $\RR$, contradicting Lemma \ref{lem:bound-intersected-rectangles}. Hence, the bottom edge of $R^v$ must be contained in $\ell$. But then Lemma \ref{lem:structural-property-of-line} implies that there must be $k$ vertically oriented rectangles whose bottom edges are contained in $\ell$. Together with the $k$ rectangles inside $\RR$, this gives $2k$ vertically oriented rectangles inside the knapsack, contradicting our assumption that there is at least one horizontally oriented rectangle.
    \end{proof}

Before proving Lemma \ref{lem:only-if-direction}, we state a rectangle packing lemma from \cite{nadiradze2016approximating}, which will be useful in the proof.

\begin{lemma}[see Lemma 4.3 in \cite{nadiradze2016approximating}]
\label{lem:sorted-packing-lemma}
    Suppose there exists a packing of a collection of rectangles $\TT\dot{\cup} \B$ inside the knapsack, where each rectangle in $\TT$ (resp. $\B$) touches the top (resp. bottom) boundary of the knapsack. Then there exists a feasible packing of $\TT\dot{\cup} \B$, where the rectangles of $\TT$ (resp. $\B$) are stacked one next to the other from left to right (resp. right to left) sorted non-increasingly by height, starting from the left (resp. right) boundary of the knapsack.
\end{lemma}

We are now ready to prove Lemma \ref{lem:only-if-direction}. Note that by Lemma \ref{lem:orientation}, we can assume w.l.o.g.~that the rectangles in $\OPT$ are all vertically oriented. Hence, the following proof works both for 2DK and 2DKR.

\begin{proof}[Proof of Lemma \ref{lem:only-if-direction}]
    Since any vertical line intersects at most two rectangles, we can shift all rectangles vertically so that each of them touches either the top or bottom boundary of the knapsack. Since the width of each rectangle is at least $(\frac{1}{k}-\frac{1}{k^4})N$ (Lemma \ref{lem:rectangle-properties}(ii)), it follows that both the top and bottom knapsack boundaries are touched by exactly $k$ rectangles each. Applying Lemma \ref{lem:sorted-packing-lemma} to this packing, we obtain another feasible packing of $\OPT$ where the rectangles touching the top (respectively, bottom) boundary are sorted non-increasingly by height from left to right (respectively, right to left).

    Let $\Rs_{t_1},\Rs_{t_2},\ldots,\Rs_{t_k}$ (resp. $\Rs_{b_1},\Rs_{b_2},\ldots, \Rs_{b_k}$) be the rectangles from left to right touching the top (resp. bottom) boundary. Then notice that for any $j\in [k]$, if there exists a horizontal line $\ell$ intersecting both $\Rs_{t_j}$ and $\Rs_{b_j}$, then $\ell$ must also intersect the rectangles $\Rs_{t_{1}},\ldots,\Rs_{t_{j-1}}$ and $\Rs_{b_{j+1}},\ldots, \Rs_{b_{k}}$. Thus $\ell$ intersects $k+1$ rectangles from $X\cup X'$, a contradiction, since the total width of $k+1$ rectangles would exceed $N$. Hence it holds that $h(\Rs_{t_j})+h(\Rs_{b_j}) \le N$, for all $j\in [k]$. From Lemma \ref{lem:rectangle-properties}(i), we then have $w(\Rs_{t_j})+w(\Rs_{b_j}) = (\frac{N}{2}+\frac{N}{k}-h(\Rs_{t_j}))+(\frac{N}{2}+\frac{N}{k}-h(\Rs_{b_j}))\ge \frac{2N}{k}$. Summing over $j\in [k]$, we obtain
    \begin{equation}
        \label{eqn:sum-of-widths}
        \sum_{j\in [k]} w(\Rs_{t_j})+\sum_{j\in [k]} w(\Rs_{b_j})\ge 2N.
    \end{equation}
    But since the rectangles $\{\Rs_{t_j}\}_{j\in [k]}$ (respectively, $\{\Rs_{b_j}\}_{j\in [k]}$) all touch the top (respectively, bottom) boundary of the knapsack, we have $\sum_{j\in [k]} w(\Rs_{t_j}) \le N$ (resp. $\sum_{j\in [k]} w(\Rs_{b_j}) \le N$), and therefore from \eqref{eqn:sum-of-widths}, it must be the case that $\sum_{j\in [k]} w(\Rs_{t_j}) = \sum_{j\in [k]} w(\Rs_{b_j})=N$.

    Note that by our construction of the rectangles, we have $w(R'_a)<\frac{N}{k}<w(R_{\overline{a}})$ and $h(R'_a)>\frac{N}{2}>h(R_{\overline{a}})$, for any $a,\overline{a} \in \AM$. Since $\sum_{j\in [k]} w(\Rs_{t_j})=N$, there must be at least one rectangle each from $X$ and $X'$. As the rectangles $\{\Rs_{t_j}\}$ were sorted non-increasingly by height, there must exist an $m \in [k]$ such that $\Rs_{t_1}, \ldots, \Rs_{t_m} \in X'$ and $\Rs_{t_{m+1}},\ldots, \Rs_{t_k} \in X$. 
     Note that we can exclude $m=0$ since among $\Rs_{t_1}, \ldots, \Rs_{t_m}$ there is at least one rectangle each from $X$ and $X'$. 
    Let $a_1,\ldots,a_k \in \AM$ be such that $\Rs_{t_j}$ corresponds to the rectangle $R'_{a_j}$, for $j\in \{1,\ldots,m\}$, and to $R_{a_j}$, for $j\in \{m+1,\ldots,k\}$. Since $\sum_{j\in [k]} w(\Rs_{t_j})=N$, we obtain $\sum_{j=1}^{m} (\frac{N}{k}-a_j) + \sum_{j=m+1}^{k} (\frac{N}{k}+a_j) =N$, and therefore $\sum_{j=1}^{m} a_j = \sum_{j=m+1}^{k} a_j$, thus completing the proof.
\end{proof}

Suppose there exists a $(1+\epsilon)$-approximation for 2DK or 2DKR with a running time of $O(n^{\frac{1}{9\eps}})$ time. Observe that since for our instance $|\OPT|\le 2k$, running the $(1+\epsilon)$-approximation with $\epsilon=\frac{1}{3k}$ would return an optimal packing. This would give us an algorithm for $k$-\textsc{PartSum} with a running time of $O(n^{3k/9})=O(n^{k/3})$, violating the $k$-\textsc{Sum} conjecture. This proves Theorem \ref{theorem:hardness-of-2dkr}.

\bibliographystyle{alpha}
\bibliography{ref}

\appendix

\section{Technical Tools}

In this section, we describe two simple packing subroutines. We also describe the Generalized Assignment Problem (GAP) in \Cref{subsec:gap}.

\subsection{Next-Fit-Decreasing-Height (NFDH)}
\label{subsec:nf}
NFDH is a shelf-based packing algorithm introduced in \cite{coffman1980performance}. Given a set of items $I$ and a rectangular region $C$, the algorithm first sorts the items in non-increasing order of height and then places the items sequentially from left to right on the base of $C$ until the next item no longer fits. At this point, a new shelf is started by drawing a horizontal line at the height of the first (i.e., the tallest) item on the current shelf, and the process is repeated until no further items can be packed. See \Cref{fig:nfdh}.

\begin{figure}
    \centering
    \includegraphics[width=0.3\linewidth]{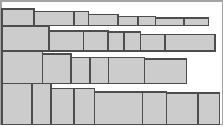}
    \caption{An example packing using NFDH with four layers.}
    \label{fig:nfdh}
\end{figure}

\begin{lemma}[\cite{galvez2021approximating}]
\label{lem:NFDH-guarantee}
    Let $C$ be a rectangular region of height $h$ and width $w$. Assume that, for some given parameter $\epsilon\in (0,1)$, for each $i\in I$ one has $w(i)\le \epsilon\cdot w$ and $h(i)\le \epsilon\cdot h$. Then NFDH is able to pack in $C$ a subset $I'\subseteq I$ of area at least $a(I')\ge \min\{a(I),(1-2\epsilon)wh\}$ in $O(n\log n)$ time. In particular, if $a(I)\le (1-2\epsilon)wh$, all items in $I$ are packed.
\end{lemma}

\subsection{Steinberg's Algorithm}
\label{subsec:stein}
Steinberg's algorithm is another commonly used area-based packing algorithm for rectangles. Informally, it states that a given collection of rectangles can be packed into a rectangular region whose area is twice the total area of the input rectangles, under certain mild assumptions on the dimensions of the input rectangles.
\begin{lemma}[\cite{steinberg1997strip}]
\label{lem:steinberg-original}
    Let $C$ be a rectangular region of height $h$ and width $w$, and let $I$ be a collection of rectangles. Let $w_{\max} = \max_{i\in I} w(i)$, and $h_{\max} = \max_{i\in I} h(i)$. If it holds that $w_{\max}\le w$, $h_{\max} \le h$, and $2\sum_{i\in I} w(i)h(i) \le wh - (2w_{\max}-w)_+(2h_{\max}-h)_+$, then it is possible to pack all rectangles in $I$ into $C$ in $O\left(\frac{n\log^2 n}{\log\log n}\right)$ time. Here $x_+=\max(x,0)$.
\end{lemma}

\subsection{Generalized Assignment Problem (GAP)}
\label{subsec:gap}
GAP is a generalization of the classical one-dimensional knapsack problem. Here we are given a collection of $k$ knapsacks with associated capacities  $\{c_j\}_{j\in [k]}$, and a set of $n$ items where each item $i$ has a given size $s_{ij}$ and a given profit $p_{ij}$ for each knapsack $j$. The goal is to find a maximum profitable packing of a subset of items, satisfying the capacity constraints, i.e., for each knapsack $j$, the total size of the items assigned to $j$ must not exceed $c_j$. The general case of GAP admits a tight $(\frac{e}{e-1}+\eps)$-approximation \cite{fleischer2011tight}. However, for the special case when $k=O(1)$, there exists a PTAS.
\begin{lemma}[\cite{galvez2021approximating}]
\label{lem:GAP-lemma}
    There exists an algorithm for the Generalized Assignment Problem (GAP) with $k$ knapsacks that runs in $n^{O(k/\epsilon^2)}$ time and returns a solution that has profit at least $(1-\epsilon)p(\OPT)$, for any fixed $\epsilon >0$.
\end{lemma}

\end{document}